\documentclass[conference]{IEEEtran}
\usepackage[compatibility=false]{caption}

\usepackage{enumitem}
\usepackage{algorithm}
\usepackage{algpseudocode}
\usepackage{amsmath}
\usepackage{tikz}
\usetikzlibrary{shapes}
\usepackage{adjustbox} 
 \usepackage[subtle,tracking=normal]{savetrees}

\usepackage{subfiles}
\usepackage{amsthm}
\usepackage{amssymb}
\usepackage{mathtools}
\usetikzlibrary{shapes.geometric, arrows.meta, positioning, calc, backgrounds, fit}
\usetikzlibrary{decorations.pathreplacing}
\usepackage{multirow}
\usepackage{booktabs}
\usepackage{float}
\usepackage{graphicx}

\usepackage{subcaption}
\usepackage{url}
\newtheorem{theorem}{Theorem}
\newtheorem*{theorem*}{Theorem}
\newtheorem{remark}{Remark}
\newtheorem{lemma}{Lemma}
\newtheorem*{lemma*}{Lemma}

\newcommand{\poly}{\operatorname{poly}}

\newcommand{\IP}{\text{Pr}}
\newcommand{\IE}{\mathbb{E}}

\newcommand{\bY}{\mathbf{Y}}
\AtBeginDocument{%
  }

\ifCLASSINFOpdf
\else
\fi
\begin{document}
%
\title{Provably Secure Diffusion Watermarking Against Removal Adversaries with Batched Detector Access}

\author{
\IEEEauthorblockN{De Zhang Lee}
	\IEEEauthorblockA{National University of Singapore}
	\and
	\IEEEauthorblockN{Ee-Chien Chang}
	\IEEEauthorblockA{National University of Singapore}
        }

\maketitle

\begin{abstract}
 Recent latent-based watermarking methods can embed undetectable watermarks into AI-generated content. However, these approaches are vulnerable to removal attacks as their watermark detection boundaries are exposed, prompting the use of secret long-term transformations to hide these boundaries. We show that such defenses are insecure under a stronger threat model motivated by publicly reachable detectors: an adversary with batched detector access can  reconstruct the transformations, permanently re-exposing the decision boundary. To address this, we propose a watermarking scheme based on a hybrid of long-term and one-time secret transformations, which prevents the reuse of structure across samples and keeps the long-term secret  hidden. Besides undetectability, we show that this scheme is provably secure against removal adversaries with batched detector access, in the sense  that no probabilistic polynomial-time adversary can achieve a non-negligible advantage over an adversary restricted to adding isotropic white-noise with a $\sqrt{2}\times$ larger perturbation budget, and we show this $\sqrt{2}$ factor to be tight. Evaluations on Stable Diffusion 2.1 and XL demonstrate the attack against the existing defense, and show that our scheme preserves generation quality and withstands common perturbations and targeted removal attacks.
\end{abstract}


%
\IEEEpeerreviewmaketitle

\maketitle

\section{Introduction}

The remarkable quality of Latent Diffusion Models (LDMs) has raised significant provenance and misuse concerns, driving interest in watermarking AI-generated content (AIGC) \cite{zhao2025sok}. Notably, \cite{prc} proposed the notion of undetectability in latent-based watermarks, ensuring no polynomial-time algorithm can distinguish between watermarked and unwatermarked content.

However, recent work   \cite{ccs} demonstrated that watermark undetectability does not guarantee unremovability. Undetectability only ensures that watermarked and unwatermarked samples are indistinguishable, but it does not conceal the detector's decision boundary (e.g. the sign patterns of the individual latents in the PRC watermark \cite{prc}, exploited in \cite{ccs}). Thus, an adversary holding a known watermarked sample can perturb it directly toward this exposed boundary with small perturbations. In response, a boundary-hiding defense was proposed \cite{ccs} using a secret long-term transformation, with a proof that no polynomial-time removal attack outperforms adding random white-noise, under the assumption of a passive adversary, i.e. who observes watermarked samples but has no detector access.

The passive adversary assumption was reasonable when watermark detectors were not publicly reachable, which may not reflect current deployment settings. For example,  Google's SynthID
\cite{synthid_arxiv}, which watermarks Gemini's AIGC, already permits public detector
access \cite{synthid_website}, with an announced detection API
\cite{verge_synthid_c2pa} that would enable batched queries. Detectors may
also be reverse-engineered, e.g. \texttt{reverse-SynthID}
\cite{reverse_synthid_github}, a popular open-source approximation of the
SynthID detector. Nor is this an isolated development. Driven in part by the EU AI Act's transparency requirements~\cite{euaiact2024}, other providers are following suit; Anthropic, for example, has announced a detection API for its text watermarks~\cite{anthropic2026watermark}. This suggests that publicly reachable detectors will become the norm across modalities. 
Thus, a defense proved secure only against a passive adversary offers no guarantee once this access exists, regardless of how that access is provisioned (e.g. a public API, a reverse-engineered proxy, or a leaked detector).

In this work, we present an attack on the boundary-hiding defense under this
stronger, more realistic threat model, i.e., a watermark removal adversary with batched (non-adaptive) detector access, which we call REM2 (REM1 corresponds to the passive adversary in \cite{ccs}). We focus on batched rather than adaptive access for two reasons. First, batched queries reflect realistic deployment conditions: public detection APIs are typically rate-limited per account, but a non-adaptive query set can be parallelized across accounts and time, whereas an adaptive strategy is inherently sequential. Second, as we show, batched access already suffices to break the existing defense.  (We discuss the adaptive analogue, REM3, in Section \ref{subsection:rem3}).

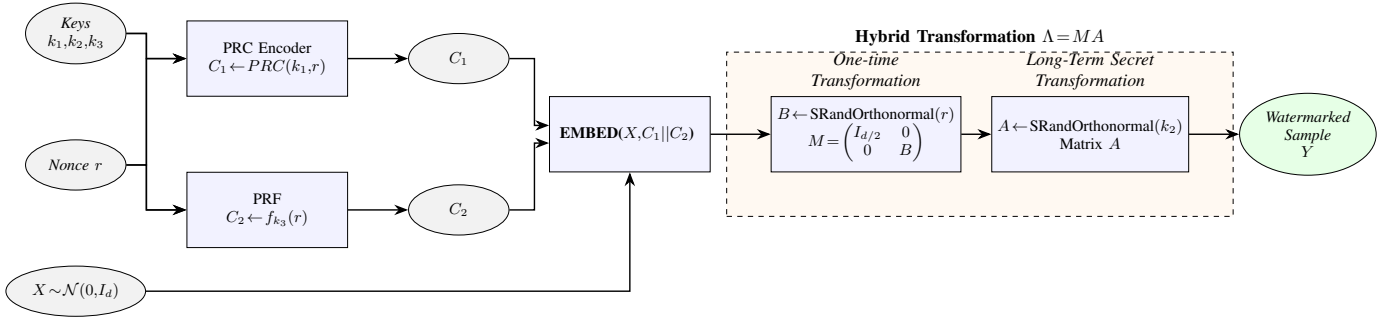
\begin{figure*}[t]
    \centering
    \begin{adjustbox}{width=\textwidth}
        \begin{tikzpicture}[
            node distance=1.5cm and 1.2cm,
            proc/.style={rectangle, draw, fill=blue!5, minimum width=3.2cm, minimum height=1.5cm, align=center, font=\small},
            val/.style={ellipse, draw, fill=gray!10, minimum width=2cm, minimum height=1cm, align=center, font=\small\itshape},
            arrow/.style={-{Stealth[scale=1.2]}, thick}
        ]

        \node[val] (r) {Nonce $r$};
        \node[val] (k) [above=of r] {Keys\\$k_1, k_2, k_3$};
        \node[val] (x_in) [below=of r] {$X \sim \mathcal{N}(0, I_d)$};

        \node[proc] (prc) [right=of k, yshift=-0.5cm] {PRC Encoder \\ $C_1 \leftarrow PRC(k_1, r)$};
        \node[proc] (prf) [below=of prc] {PRF \\ $C_2 \leftarrow f_{k_3}(r)$};

        \node[val] (c1) [right=of prc] {$C_1$};
        \node[val] (c2) [right=of prf] {$C_2$};

        \node[proc] (embed) [right=1.8cm of $(c1)!0.5!(c2)$] {\textbf{EMBED($X, {C_1 || C_2}$)} };

        \node[proc, label={[font=\itshape, align=center]above:One-time\\Transformation}] (short_term) [right=1.2cm of embed] {$B \leftarrow \text{SRandOrthonormal}(r)$\\ $M=\begin{pmatrix} I_{d/2} & 0 \\ 0 & B \end{pmatrix}$};
        
        \node[proc, label={[font=\itshape, align=center]above:Long-Term Secret\\Transformation}] (long_term) [right=0.6cm of short_term] {$A \leftarrow \text{SRandOrthonormal}(k_2)$\\Matrix $A$};
        
        \node[val, fill=green!10] (output) [right=1cm of long_term] {Watermarked\\Sample\\$Y$};

        \draw[arrow] (k.east) -- ++(0.4,0) |- (prc.west);
        \draw[arrow] (k.east) -- ++(0.4,0) |- (prf.west);
        \draw[arrow] (r.east) -- ++(0.4,0) |- (prc.west);
        \draw[arrow] (r.east) -- ++(0.4,0) |- (prf.west);

        \draw[arrow] (prc) -- (c1);
        \draw[arrow] (prf) -- (c2);

        \draw[arrow] (c1.east) -- ++(0.5,0) |- ([yshift=5pt]embed.west);
        \draw[arrow] (c2.east) -- ++(0.5,0) |- ([yshift=-5pt]embed.west);

        \draw[arrow] (x_in.east) -| (embed.south);

        \draw[arrow] (embed) -- (short_term) ;
        \draw[arrow] (short_term) -- (long_term);
        \draw[arrow] (long_term) -- (output);

        \begin{scope}[on background layer]
            \node[draw, dashed, fill=orange!5, inner sep=25pt, 
                  label={[font=\bfseries]above:Hybrid Transformation $\Lambda = MA$}, 
                  fit=(short_term) (long_term)] {};
        \end{scope}

        \end{tikzpicture}
    \end{adjustbox}
    \caption{ Flowchart of  $\texttt{Sample}(\lambda, d, k)$ (Algorithm \ref{alg:sample}), for sampling watermarked latents.  The scheme leverages fresh randomness $r$ to derive a unique one-time transformation for each sample. In the hybrid transformation, this one-time transformation dynamically obfuscates the detector’s decision boundary, while the long-term secret transformation $A$ securely hides the per-sample random nonce $r$. }
    \label{fig:watermark_sample}
\end{figure*}

Lifting REM1 to REM2 is not a hypothetical strengthening: we demonstrate how a REM2 adversary can fully break the boundary-hiding defense by recovering the secret long-term transformation, thereby re-exposing the detector's decision boundaries permanently. As the
transformation is long-term, this one-time recovery cost is amortized across
all future attacks, constituting a complete break of the defense. The extension from REM1 to REM2 is analogous in spirit to the upgrade from IND-CPA to IND-CCA1 in encryption: an adversary gains oracle access to a secret-keyed algorithm, and an IND-CPA-secure scheme is not necessarily IND-CCA1-secure.

We approach watermark removal from the standpoint of provable security. Similar to certified robustness against adversarial attacks \cite{cohen2019certified},  we seek guarantees that hold against an entire class of adversaries rather than fixed attacks.
Since our adversary queries the
detector to learn the hidden decision boundary, security is a question of boundary
secrecy rather than geometry, and is naturally studied through a
cryptographic lens. Concretely, we (i) formalize detection and removal adversaries via explicit
security games (IND0/IND1 and REM1/REM2, respectively); (ii) prove an
information-theoretic guarantee for an idealized instantiation
(Theorem~\ref{thm:rem2-secure}); and (iii) lift it to a computational
guarantee for the deployed construction via standard cryptographic primitives
(Theorem~\ref{thm:computational_rem2}). Unlike prior schemes, which assume a
passive adversary \cite{ccs} or guarantee only undetectability
\cite{gs,prc}, we formally prove removal security against an adversary with batched
detector access, via cryptographic reduction rather than a certified-radius
argument.

Our construction combines a long-term and a one-time transformation. Inspired
by hybrid encryption, where a long-term public-key scheme exchanges a
one-time session key, our scheme generates a fresh one-time transformation
per sample. This dynamically obfuscates the detector's decision boundary, while
the long-term transformation hides the per-sample nonce and is used during
detection to retrieve it (Figure~\ref{fig:watermark_sample}). Because the
effective transformation is fresh per sample and its orientation hidden, any
adversarial perturbation is rendered independent of the transformation. Our worst-case security analysis proves that an adversary can at best match
an isotropic white-noise adversary with a $\sqrt{2}\times$ larger budget, and
we show this constant is tight.

Our empirical evaluation focuses on three key aspects. First, we demonstrate that our proposed attack enables an adversary with batched detector access to successfully recover the secret transformation, thereby re-exposing the detection boundary. Second, we show that our latent-based watermark integrates seamlessly into existing LDM pipelines while preserving high image quality. Third, we verify our scheme's robustness against common benign perturbations, such as Gaussian noise and JPEG compression, regeneration attacks \cite{zhao2024invisible}, as well as against the targeted removal attack specifically designed for undetectable latent-based watermarks proposed in \cite{ccs}.

Our central contribution is a provably secure watermarking scheme, developed as follows:
\begin{enumerate}
 \item A formal threat model for latent-based watermarks is introduced, unifying detectability and removability notions that prior work treats only implicitly (Section \ref{subsection:threat_model}).
 
     \item An attack against boundary hiding \cite{ccs} is demonstrated under the REM2 threat model, showing how an adversary can successfully recover the underlying long-term secret transformation employed in the existing  boundary-hiding defense (Section \ref{subsection:attack_on_defense}).
     
    \item A provably secure latent-based watermarking construction. Under the REM2 threat model,  no PPT adversary removes the watermark with non-negligible advantage over a white-noise baseline granted a $\sqrt{2}\times $ larger perturbation budget (Theorems \ref{thm:rem2-secure} \& \ref{thm:computational_rem2}).

    \item Empirical evaluations demonstrate the proposed attack, verify that the new watermarking scheme preserves high image quality, and measure the overall robustness of the framework (Section \ref{section:evals}).
\end{enumerate}

\section{Background}
\subsection{Latent Diffusion Model}

A typical latent diffusion model (LDM) (e.g. Stable Diffusion \cite{rombach2021highresolution}), consists of two main components: (1) denoising model (e.g., DDIM \cite{ddim}, DDPM \cite{ddpm}), and (2) variational autoencoder (VAE) \cite{vae}. The model takes as input: (1) a text prompt, (2) generation hyperparameters (e.g., the number of timesteps), and (3) a random latent vector sampled from the VAE latent space.

At a high level, the generative process is as follows: First, the input text prompt is converted into a text embedding using a text encoder (e.g. CLIP \cite{clip}). Next, a randomly sampled starting point is iteratively denoised by the diffusion model, conditional on the text embedding and generation hyperparameters. Finally, the denoised latent representation is decoded by the VAE decoder into the pixel space.

For analysis, let us define the LDM as a deterministic function
$L : \mathcal{C} \times \mathcal{L}_v  \rightarrow \mathcal{L}_p$ 
that  takes as input the following: (1) a condition $c \in \mathcal{C}$, and (2) a starting point $s \in \mathcal{L}_v$, sampled from some distribution $\mathcal{D}$. The LDM outputs the generated image  in the pixel  space $\mathcal{L}_p$.  That is, conditional on condition $c$ (i.e. prompt and meta-parameters) and a randomly sampled starting point $s$, the output of the LDM is 
$
L(c, s) \in \mathcal{L}_p.
$
In Stable Diffusion (SD), $\mathcal{D} = \mathcal{N}(0, I_d)$. 
The output space $\mathcal{L}_v$ is a real valued $d$-dimensional  vector space, where $d= 4\times 64\times 64= 16,384$ and $d = 4 \times 128 \times 128 = 65,536$ in SD1.5/2.1 and SDXL respectively. 

Furthermore, the $L$ typically admits an approximate inversion function
$
L^{-1} : \mathcal{L}_p \times \mathcal{C} \rightarrow \mathcal{L}_v.
$
That is, given a condition $c \in \mathcal{C}$ and an image $I = L(s, c) \in \mathcal{L}_p$ in the pixel domain, the inverse mapping $\hat{s} = L^{-1}(I, c) \in \mathcal{L}_v$ outputs an approximation of the original latent $s$. 
Recent work (e.g. \cite{cosda, roar}) has focused on improving the quality of such inversion procedures, showing that $\hat{s}$ can be constructed with small reconstruction error even when the exact $c$ is unknown. 

\subsection{Latent-Based Watermarking}

Latent-based watermarking schemes typically involve two distributions: an underlying distribution $\mathcal{D}$ and a keyed watermarked distribution $\mathcal{D}_k$, where $k$ is a secret key known only to the watermark owner. The procedures for generating watermarked samples and detecting watermarks are as follows.

\paragraph{Generation}
Instead of sampling the initial latent from $\mathcal{D}$, the starting point $s$ is sampled from the keyed watermarked distribution $\mathcal{D}_k$. That is, given a condition $c \in \mathcal{C}$, the watermarked output is
$
I_w = L(s, c),
$
where $s \sim \mathcal{D}_k$.

\paragraph{Detection}
Given an image $I$ generated by the LDM, an approximate latent starting point $\hat{s}$ is recovered via the inversion
$
\hat{s} = L^{-1}(I, c_*),
$
where $c_*$ is the inversion condition. In practical settings, the exact generation condition may be unknown; thus, one often sets $c_* = \emptyset$ (the null prompt). Empirical evidence from recent work suggests that even under the null prompt, $\hat{s}$ remains close to the original $s$.

\subsubsection*{Undetectability}

Some latent-based watermarks (e.g., Tree-Ring watermark \cite{treering}) are detectable, meaning an adversary with access to polynomial samples from the watermarked distribution $\mathcal{D}_k$ can reliably distinguish a new watermarked sample from the original distribution $\mathcal{D}$. There are various notions of undetectability, depending on the threat model, described in the IND0 and IND1 games in Section \ref{subsection:threat_model}.


\subsection{Pseudorandom Error Correcting Codes} \label{subsection:prc_code}
Pseudorandom Error-Correcting (PRC) codes \cite{christ2024pseudorandomerrorcorrectingcodes} are used as cryptographic primitives to construct undetectable watermarks.
It consists of two components: an encoder and decoder. 

\paragraph{Encoder}
The encoder $Enc : \{0,1\}^n \times \{0, 1\}^{n_1} \rightarrow \{0, 1\}^{n_2}$  is a probabilistic algorithm, which takes in a $n$-bit key $k$ and a $n_1$-bit message $m$
respectively, and outputs a codeword of length $n_2$. The output is pseudorandom in the sense that, for any $k$ and $m$, $Enc(k, m)$ is computationally indistinguishable from a truly random binary sequence for any PPT adversary who does not know $k$.

\paragraph{Decoder}
The decoder
$Dec : \{0,1\}^n \times \ \{0,1\}^{n_2} \rightarrow \{0, 1\}^{n_1} \cup \{ \bot \}$
takes in a $n$-bit key and $n_2$-bit codeword, and outputs the $m$-bit
watermark message, or $\bot$ if the decoding failed.
The scheme is error-correcting in the sense that if $c = Enc(m, k)$, there exists a $\tau$ s.t. 
\[
d_H(c, \hat{c}) \leq \tau \;\Rightarrow\; Dec(k, \hat{c}) = m,
\]
where $d_H$ denotes the Hamming distance, and $\tau$ depends on the parameters of the PRC scheme. That is, as long as the number flips between the PRC codeword $c$ and $\hat{c}$ does not exceed  $\tau$, the decoder correctly recovers the message $m$.

\subsubsection*{Relationship with PRC Watermarking Scheme \cite{prc}}

Consistent with existing implementations, let $\mathcal{D} = \mathcal{N}(0, I_d)$. Conditional on a secret key $k$ and watermark message $m$, the PRC latent-based watermarking scheme samples from $\mathcal{D}_k$ as follows:
\begin{enumerate}
    \item Sample $x = (x_1, \dots, x_d) \sim \mathcal{N}(0, I_d)$ and compute the codeword $c = (c_1, \dots, c_d) = Enc(k, m)$.
    
    \item For $i = 1, \dots, d$, set
    \[
    x_i =
    \begin{cases}
        |x_i|, & \text{if } c_i = 1,\\
        -|x_i|, & \text{if } c_i = 0.
    \end{cases}
    \]
    
    \item Output $x$.
\end{enumerate}

Similarly, to detect the PRC watermark in a sample $s = (s_1, \dots, s_d)$:
\begin{enumerate}
    \item Recover the codeword $c = (c_1, \dots, c_d)$, where $c_i = 1(s_i \geq 0)$
    
    \item Perform decoding and recover $\hat{m} = Dec(k, c)$.
    
    \item Declare $s$ as watermarked if and only if $\hat{m} = m$.
\end{enumerate}

It can be shown that, if  $c$ is pseudorandom (i.e.  computationally indistinguishable from a random binary sequence), then $x$ is  computationally indistinguishable from $\mathcal{N}(0,I_d)$. 

\subsection{Boundary-Hiding Defense} \label{subsection:bh}

In latent-based watermarking, \cite{ccs} showed that undetectability alone does not guarantee unremovability. Specifically, even if $\mathcal{D}$ and $\mathcal{D}_k$ are computationally indistinguishable, an adversary may still be able to deduce the decision boundary of the detector $\mathcal{O}$. 
 In a \emph{known-watermark attack}, an adversary perturbs a known watermarked latent $s \sim \mathcal{D}_k$ into an unwatermarked state $\tilde{s}$ such that:
\[
\mathcal{O}(\tilde{s}) = 0
\quad \text{and} \quad
d(s, \tilde{s}) \text{ is small},
\]
for some distance metric $d(\cdot,\cdot)$. 
To mitigate this, \cite{ccs} proposed a boundary-hiding defense that masks the decision boundary via a secret linear transformation. Specifically, consider the transformation
$
T(s) = sA,
$
where $A \in \mathbb{R}^{d \times d}$ is a secret, invertible orthonormal matrix. During detection, a received latent $s'$ is first inverted via $T^{-1}(s')$ before watermark decoding. Since the adversary does not know $A$, $\mathcal{O}$'s decision boundary 
is hidden. \cite{ccs} showed that a watermark removal adversary with a polynomial number of watermarked samples (without detector access) cannot outperform an adversary limited to adding white-noise. 

\subsection{Summary of Notations Used}
For convenience, notations introduced throughout this paper (both introduced earlier and later), are summarized in Table~\ref{table:notation}.
\begin{table}
\centering
\caption{Summary of notations used throughout the paper.}
\label{table:notation}
\begin{tabular}{@{}ll@{}}
\toprule
\textbf{Notation} & \textbf{Description} \\ \midrule
$L, L^{-1}$ & LDM generation and inversion functions \\
$\mathcal{C}, \mathcal{L}_v, \mathcal{L}_p$ & Condition, latent ($\mathbb{R}^d$), and pixel spaces \\
$\mathcal{D}, \mathcal{D}_k$ & Base and watermarked latent distributions \\
$\mathcal{O}(\cdot)$ & Watermark detector oracle \\
$A$ & Long-term secret transformation ($\mathbb{R}^{d \times d}$) \\
$B_i$ & One-time transformation ($\mathbb{R}^{\frac{d}{2} \times \frac{d}{2}}$) for sample $i$ \\
$\Lambda_i$ & Effective hybrid transformation for sample $i$ \\
PRC & Pseudorandom error-correcting code \\
$\tau_1, \tau_2$ & Error tolerance thresholds \\
$d_H(\cdot, \cdot)$ & Hamming distance metric \\
$\mathcal{A}, {Adv}$ & PPT adversary and its advantage \\
\bottomrule
\end{tabular}%
\end{table}

\section{Problem Setup}

Let $I = L(e, s)$ denote the output of a LDM $L$ with conditioning $e$ and initial latent $s$.
In latent-based watermarking, $s$ is sampled from a watermarked distribution $\mathcal{D}_k$ parameterized by a secret key $k$. There are 3 phases involved: \textit{Setup}, \textit{Sample}, \textit{Detect}.

\paragraph{Setup}
During the setup phase, the system generates the relevant keys $k$, which parameterize the watermarked distribution $\mathcal{D}_k$.

\paragraph{Sample}
To generate a watermarked image, an embedding $e$ and  starting point $s \sim \mathcal{D}_k$, is provided as input to $L$. The model then outputs the generated image $I = L(e, s)$.

\paragraph{Detect}
To detect watermark in an image, diffusion inversion is performed on the image to approximate $s$, denoted as $\hat{s}$, which is then passed into a watermark detector $\mathcal{O} : \mathbb{R}^d \to \{0,1\}$. The detector outputs $\mathcal{O}(\hat{s}) = 1(\hat{s} \sim \mathcal{D}_k)$
\subsection{Threat Model} \label{subsection:threat_model}
\begin{table*}[t]
\centering
\caption{Comparison of threat models and adversary capabilities between prior undetectable latent-based watermarking schemes (Gaussian Shading (GS) \cite{gs}, GS++ \cite{gspp},  PRC watermark \cite{prc}, and boundary-hiding (BH) \cite{ccs} defense) and our proposed construction. Compare to BH defense which achieves IND1-REM1 security,  our scheme achieves  IND1-REM2 security.}
\label{tab:threat-models}
\renewcommand{\arraystretch}{1.5} 
\begin{tabular}{|l|l|p{5.5cm}|c|c|c|}
\hline
\multirow{2}{*}{\textbf{Threat Model}} & \multirow{2}{*}{\textbf{Work}} & \multirow{2}{*}{\textbf{Goal}} & \multicolumn{3}{c|}{\textbf{Adversary's Capabilities}} \\
\cline{4-6}
& & & \textbf{1 Sample} & \textbf{$>1$ Samples} & \textbf{Detector Access} \\
\hline
IND0  & GS & Undetectability (Single Sample) & $\checkmark$ & $\times$ & $\times$ \\
\hline
IND1  & PRC, GS++, BH, Ours & Undetectability (Polynomial Samples) & $\checkmark$ & $\checkmark$ & $\times$ \\
\hline
REM1 & BH & Unremovability  (Polynomial Samples) & $\checkmark$ & $\checkmark$ & $\times$ \\
\hline
REM2  & Ours &  Unremovability  (Polynomial Samples \newline + Batched  Detector Access) & $\checkmark$ & $\checkmark$ & $\checkmark$ \\
\hline
\end{tabular}
\end{table*}


The adversary's primary goals are watermark detection and removal. While some prior work relies on implicit security assumptions, we formalize the adversary's capabilities into an explicit hierarchy of threat models to provide a unified evaluation framework. Higher indices denote stronger adversaries:
\begin{enumerate}
    \item \textbf{Watermark Detectability:} {IND0}, {IND1}
    \item \textbf{Watermark Removal:} {REM1}, {REM2}
\end{enumerate}
Table \ref{tab:threat-models} summarizes these threat models, which are formally defined here.
In this paper, we consider schemes which achieve at least {IND0}-security (single-sample indistinguishability), thus excluding detectable watermarking schemes (e.g. Tree-Ring \cite{treering}).

\subsubsection{IND0: Single Sample Indistinguishability}
 Given a watermark distribution $\mathcal{D}_k$ parameterized by a secret key $k$ and a target distribution $\mathcal{D}$:
\begin{enumerate}
    \item The challenger uniformly samples $b \leftarrow \{0, 1\}$.
    \item If $b = 0$, the challenger outputs a sample from $\mathcal{D}_k$. Else, the challenger outputs a sample from $\mathcal{D}$.
\end{enumerate}
After receiving the sample, a distinguisher guesses the value of $b$, outputs a guess $b'$, and wins if and only if $b' = b$. A scheme is IND0-secure if for all PPT distinguishers:
$$ \left|\Pr[b' = b] - \frac{1}{2}\right| \leq {negl}(\lambda) $$
where ${negl}(\cdot)$ is a negligible function.

\subsubsection{IND1: Multiple Sample Indistinguishability}
IND1 is an extension of IND0 from 1 sample to polynomial number of samples, described as follows.  Given a watermark distribution $\mathcal{D}_k$ parameterized by a secret key $k$ and a target distribution $\mathcal{D}$:
\begin{enumerate}
    \item The challenger uniformly samples $b \leftarrow \{0, 1\}$.
    \item The adversary submits an integer $t$ to the challenger, where $t$ is bounded by a polynomial in the security parameter $\lambda$.
    \item The challenger sends $t$ unwatermarked samples sampled independently from $\mathcal{D}$, and $t$ watermarked samples sampled independently from $\mathcal{D}_k$
    \item If $b = 0$, the challenger outputs a sample sampled independently $\mathcal{D}_k$. Else, the challenger outputs a sample sampled independently from $\mathcal{D}$.
\end{enumerate}
A scheme is IND1-secure, if for all PPT distinguishers:
$$ \left|\Pr[b' = b] - \frac{1}{2}\right| \leq {negl}(\lambda) $$
where ${negl}(\cdot)$ is a negligible function.

\subsubsection{REM1: Known-Watermark Attack}
Unlike the earlier games, REM1 deals with removal attacks. Consider the following security game:
\begin{enumerate}
    \item The adversary receives a collection of watermarked and unwatermarked samples $Y_w$ and $Y_{nw}$, sampled independently from $\mathcal{D}_k$ and $\mathcal{D}$ respectively. $|Y_w|$ and $|Y_{nw}|$ are polynomial in $\lambda$.
    \item The adversary is given a fresh sample $Y \sim \mathcal{D}_k$.
    \item The adversary computes a perturbed sample $Y' = Y + \eta$ subject to a distortion budget $\sigma$ and passes it to the challenger
    \item The challenger outputs $\mathcal{O}(Y')$
\end{enumerate}
Let $\mathcal{A}_w(Y, \sigma)$ denote a baseline algorithm that outputs a perturbed sample by adding randomly sampled white-noise of magnitude $\sigma$ to $Y$. Let $\mathcal{A}(Y, Y_w, Y_{nw}, \sigma)$ denote a PPT adversary that observes $Y_w$ and $Y_{nw}$, and outputs a perturbed sample by adding a computed perturbation of magnitude at most $\sigma$ to $Y$. A scheme is $c$-REM1-secure for constant $c \geq 1$ iff for all PPT adversaries $\mathcal{A}$ and all valid perturbation budgets $\sigma > 0$:
$$
\resizebox{\linewidth}{!}{
$\Pr[\mathcal{O}(\mathcal{A}(Y, Y_w, Y_{nw}, \sigma) = 0] - \Pr[\mathcal{O}(\mathcal{A}_w(Y, c\sigma)) = 0]  \le negl(\lambda) $
}
$$
where ${negl}(\cdot)$ is a negligible function.

\subsubsection{REM2: Chosen Noise Attack + Batch Detector Access} \label{subsubsection:rem2}
This threat model extends REM1 by granting the adversary batched oracle access to the detector $\mathcal{O}$.  In this scenario, the adversary has access to a polynomial number of watermarked samples $\mathbf{Y} = \{Y_1, \dots, Y_t\}$ and can issue a polynomial number of queries to the detector to observe its decision boundary. This is formalized in the following security game:
\begin{enumerate}
    \item The adversary receives a collection of watermarked and unwatermarked samples $Y_w$ and $Y_{nw}$, sampled independently from $\mathcal{D}_k$ and $\mathcal{D}$ respectively. $|Y_w|$ and $|Y_{nw}|$ are polynomial in $\lambda$.
    \item The adversary submits a sequence of $s$ queries $q_1, \dots, q_s$ to the detector oracle $\mathcal{O}$, where $s$ is bounded by a polynomial in $\lambda$.
    \item The challenger returns the detector outputs $\mathcal{O}(q_1), \dots, \mathcal{O}(q_s)$.
    \item The adversary is given a fresh watermarked sample $Y \sim \mathcal{D}_k$.
    \item The adversary outputs a perturbed sample $Y' = Y + \eta$
    \item Challenger outputs $\mathcal{O}(Y')$
\end{enumerate}
Let $\mathcal{A}_w(Y, \sigma)$ denote a baseline algorithm that outputs a perturbed sample by adding randomly sampled white-noise of magnitude $\sigma$ to $Y$, and let $\mathcal{A}^{\mathcal{O}}(Y, \mathbf{Y}, \sigma)$ denote a PPT adversary with oracle access to $\mathcal{O}$ that outputs a perturbed sample by adding a computed perturbation of magnitude at most $\sigma$ to $Y$. A watermarking scheme $c$-REM2-secure for constant $c \geq 1$, iff for all PPT adversaries $\mathcal{A}^{\mathcal{O}}$, $Y \sim \mathcal{D}_k$ and $\sigma > 0$:
$$ \Pr[\mathcal{O}(\mathcal{A}^{\mathcal{O}}(Y, \mathbf{Y}, \sigma)) = 0] -\Pr[\mathcal{O}(\mathcal{A}_w(Y, c\sigma)) =0]   \le {negl}(\lambda) $$
where ${negl}(\cdot)$ is a negligible function

\section{Attacks on Existing Undetectable Latent-based Watermarking Schemes}

\subsection{Attack on Gaussian Shading Watermark}
Gaussian Shading (GS) theoretically achieves IND1 security by employing a stream cipher initialized with a secret key and a random nonce. Provided this nonce is randomly sampled for every watermarked sample generation, the resulting sequence is computationally indistinguishable from $\mathcal{N}(0, I_d)$. 

However, to circumvent the prohibitive computational cost of exhaustive nonce searching during verification, practical GS deployments frequently rely on a restricted nonce pool or static session nonces. An adversary who intercepts a polynomial number of watermarked samples generated with reused nonces can successfully distinguish between watermarked and unwatermarked distributions  \cite{prc, clwe-mark}. This breaks IND1 security, and the scheme is  only IND0-secure. 

\subsection{Attack on PRC Watermark}  \label{subsection:attack_prc}
Recent work \cite{ccs} proves that IND1 security does not imply REM1 security by demonstrating a ``known-watermark attack'' against the PRC scheme. Specifically, the PRC watermark exposes its decision boundary to adversaries by directly embedding the codeword into the signs of the starting latent vector. Consequently, an adversary with a restricted perturbation budget $\sigma$ can strategically flip the signs of the latent components with the smallest magnitudes. This targeted approach achieves a non-negligible advantage over a baseline adversary restricted to adding  white-noise of equivalent magnitude $\sigma$.
This removal strategy also generalizes to GS++, as both schemes inadvertently leak their watermark detection boundaries through similar underlying detector mechanisms.

\subsection{Proposed Attack on Boundary-Hiding Defense} \label{subsection:attack_on_defense}

The boundary-hiding defense detailed in Section \ref{subsection:bh} is designed to obscure the watermark detector's decision boundary, which  mitigates the known-watermark attack outlined in Section \ref{subsection:attack_prc}. However, if an adversary recovers or sufficiently approximates the transformation matrix $A$, the decision boundary is re-exposed, which reduces back to the known-watermark attack setting.

Under the REM2 threat model, we describe an algorithm to approximate $A$. Let $Y = T(X) = XA$ be a watermarked sample, and let $\tau$ be the detector's threshold (as described in Section \ref{subsection:prc_code}). Because the decision boundary is an intersection of hyperplanes defined by the orthonormal rows of $A$, the attack aims to isolate the normal vectors to these bounding hyperplanes using filtered random noise.

\begin{enumerate}
    \item \textbf{Noise Magnitude Calibration.} The adversary first determines a target perturbation scale $\sigma$ designed to push the watermarked sample just past the decision boundary with a small, controlled probability of evasion. For a detector with a sign-flip threshold $\tau$, this target magnitude scales proportionally to $\tan(\pi \tau)$.
        \begin{itemize}
            \item For $G \sim \mathcal{N}(0, I_d)$ and $H \sim \mathcal{N}(0, \sigma^2 I_d)$, where both $G$ and $H$ are independent, it can be shown that setting $\sigma = \tan(\pi \tau)$ guarantees that the expected proportion of sign flips between $G$ and $G+H$ is exactly $\tau$.
            \item Because the latent dimension $d$ is typically very large, the empirical proportion of sign flips tightly concentrates around its expected value $\tau$ with high probability.
        \end{itemize}
    \item \textbf{Probing.} The adversary generates a massive batch of $N$ random noise vectors $\eta_i \sim \mathcal{N}(0, \sigma^2 I)$. The adversary adds these to the watermarked sample and queries the oracle for the entire batch simultaneously: $o_i = \mathcal{O}(Y + \eta_i)$.

\item \textbf{Boundary-Crossing Filtration.} The adversary retains only the perturbations that successfully evaded the detector,  which is the set
    $
    \mathcal{E} = \{ \eta_i \mid \mathcal{O}(Y + \eta_i) = 0 \}
    $
    Geometrically, $\mathcal{E}$ represents the collection of perturbations possessing sufficiently large projections onto at least $\lceil \tau d \rceil$ rows of $A$. Filtering out non-evading vectors physically truncates the Gaussian noise cloud, skewing its principal directions of variance outward along the normal vectors of these active hyperplanes.
    
\item \textbf{Independent Component Analysis (ICA).} The adversary applies ICA to the mean-centered matrix of $\mathcal{E}$ to compute 
    $
    \mathcal{V} = \{v_1, \dots, v_d\}.
    $
    By maximizing non-Gaussianity, ICA isolates the statistically independent components responsible for the large projections within $\mathcal{E}$. These extracted vectors $v_i$ serve as approximations of the individual orthonormal rows of the secret matrix $A$.
\end{enumerate}

\subsubsection{Number of Queries Needed} 
\label{subsubsection:no_operation_basic_attack}

In  Step 4, the adversary applies ICA to $\mathcal{E}$, thus the number of queries to $\mathcal{O}$ is lower bounded by $|\mathcal{E}|$.
We establish a lower bound of $|\mathcal{E}|$ by appealing to Theorem 2.4 of \cite{auddy2023largedimensionalindependentcomponent}, which gives the following information-theoretic result: Any PPT algorithm requires  $|\mathcal{E}| \in {\Omega}(d^2/ \poly \log(d))$ to successfully unmix the independent components of $A$. In other words, if $|\mathcal{E}| \in o(d^2 / \poly \log(d) )$, an adversary can at most recover $AR$, where $R \sim Haar(d)$ is a randomly sampled $d\times d$ orthonormal matrix.   Since $A$ and $AR$ share the same $Haar(d)$ distribution, the adversary gains no meaningful information regarding $A$. Thus,  $\tilde{\Omega}(d^2)$ queries are needed. 

Although this requirement makes the attack computationally prohibitive for typical values of $d$ (e.g., $16,384$ in SD2.1), we nonetheless demonstrate a polynomial-time adversary who significantly outperforms randomly guessing the transformation.




\section{Proposed Watermark Construction} \label{section:watermark_construction}

\subsection{Cryptographic Primitives}

Our watermark construction relies on the existence of the following coding and cryptographic primitives:

\begin{enumerate}
    \item \textbf{Pseudorandom Error-Correcting Code (PRC):} The details and requirements of the PRC scheme is described in Section \ref{subsection:prc_code}. 
    

    \item \textbf{Pseudorandom Function (PRF):} Given key space $\mathcal{K}$, let $\mathcal{F}=\{f_k:k \in \mathcal{K}\}$ denote a PRF family, where each member of the family is indexed by a key $k$. For a uniformly sampled $k$ and any input $r$, the function $f_k(r)$ is computationally indistinguishable from a truly random sequence.  
\end{enumerate}


\subsection{Watermark Scheme} \label{subsection:watermark}
We now formalize our watermarking scheme through three core algorithms: $\texttt{KeyGen}$, $\texttt{Sample}$, and $\texttt{Decode}$, described in Algorithms \ref{alg:keygen}, \ref{alg:sample}, and \ref{alg:decode_detect}. A high level description of these algorithms are as follows:
\begin{itemize}
\item $\mathsf{KeyGen}(\lambda, d)$:  Samples and outputs  secret key $k = (k_1, k_2, k_3)$, where $k_1$ is the  PRC key, $k_2$ is used to generate the secret fixed transformation, and $k_3$ is used as the PRF key.
\item $\mathsf{Sample}(\lambda, d, k)$: Generates a sample from the watermarked distribution $\mathcal{D}_k$.
\item $\mathsf{Decode}(\lambda, Y, k, t)$: Outputs  whether a sample $Y$ is sampled from the watermarked distribution $\mathcal{D}_k$. The parameter $t$ controls the false positive rate.
\end{itemize} 

Additionally, let us define the following functions, both used in the watermark $\texttt{Sample}$ and $\texttt{Decode}$ function
\begin{itemize}
    \item $\texttt{SRandOrthonormal}: d , \{0,1\}^\lambda \rightarrow \mathbb{R}^{d \times d}$: A function which outputs a uniformly sampled $d \times d$ orthonormal matrix conditional on a nonce uniformly sampled from $\{0,1\}^\lambda $. For a uniformly sampled $u$, $\texttt{SRandOrthonormal}(d, u)$ is computationally indistinguishable from a truly random $d \times d$ orthonormal matrix, its implementation is discussed in Appendix \ref{appendix_srandorth}.
    \item $f_k : \mathcal{M} \rightarrow \{0, 1\}^{d/2}$: A keyed PRF. Using a secret key $k$, this function takes any input $m \in \mathcal{M}$ (e.g. $\mathcal{M}$ is the set of all binary strings) and produces an output $f_k(m)$. Without $k$, $f_k(m)$ is computationally indistinguishable from a truly random sequence.
\end{itemize}
An illustration of the $\mathsf{Sample}(\lambda, d, k)$ algorithm is given in Figure \ref{fig:watermark_sample}.

\begin{algorithm}
\caption{Watermarking KeyGen}\label{alg:keygen}
\begin{algorithmic}[1]
\Procedure{KeyGen}{$\lambda$}
    \State Sample secret keys $k_1, k_2, k_3 \gets \{0, 1\}^\lambda$
    \State \Return $k = (k_1, k_2, k_3)$
\EndProcedure

\end{algorithmic}
\end{algorithm}

\begin{algorithm}
\caption{Watermark Sample} \label{alg:sample}
\begin{algorithmic}[1]

\Procedure{Sample}{$\lambda, d, k$}
    \State $k_1, k_2, k_3 \leftarrow k$
    \State $r \gets \{0,1\}^{n_r}$ \Comment{Sample random nonce}
    \State $\mathbf{X} \sim \mathcal{N}(0, \mathbf{I}_d)$ \Comment{Sample Gaussian vector}
    \State $\mathbf{A} \gets \mathsf{SRandOrthonormal}\left(d, k_2\right)$
    \State $\mathbf{B} \gets \mathsf{SRandOrthonormal}\left(\frac{d}{2}, r\right)$ \Comment{Sample orthonormal matrix}
    \State $C_1 \gets \mathsf{PRC}(k_1, r)$
    \State $C_2 \gets f_{k_3}(r)$ \label{alg_line_sample_prf}
    \State $C \gets (C_1, C_2)$ \Comment{Concatenate binary sequence} \label{line:codeword} \label{alg_line_random_bin}
    \State $\mathbf{X}' \gets \texttt{Embed}(\mathbf{X}, C)$ \label{alg_line_embed}
    \State Form block matrix $\mathbf{M} = \begin{pmatrix} 
    \mathbf{I}_{d/2} & \mathbf{0} \\ 
    \mathbf{0} & \mathbf{B} 
    \end{pmatrix}$
    \State $\mathbf{Y} \gets \mathbf{X}'\mathbf{M}\mathbf{A} $ \label{alg_sample_line_mult}
    \State \Return $\mathbf{Y}$
\EndProcedure
\Statex
\Procedure{Embed}{$\mathbf{X}, S$}
    \For{$i = 1, \dots, |\mathbf{X}|$}
        \State $\mathbf{X}[i] \gets |\mathbf{X}[i]| \times (-1)^{S[i]}$
    \EndFor
    \State \Return $\mathbf{X}$
\EndProcedure
\Statex

\end{algorithmic}
\end{algorithm}

\begin{algorithm}
\caption{Watermark Decode}
\label{alg:decode_detect}
\begin{algorithmic}[1]
\Procedure{Decode}{$\lambda, \mathbf{Y}, k, t$}
    \State $k_1, k_2, k_3 \leftarrow k$
    \State $\mathbf{A} \gets \mathsf{SRandOrthonormal}\left(d, k_2\right)$
    \State $\boldsymbol{\zeta} \gets  \mathbf{Y}\mathbf{A}^T$  \label{alg_decode_line_mult1}  \Comment{Reversing the global mask $A$}
    
    \Statex \Comment{Extract signs to recover $\hat{C}_1$}
    \For{$i = 1, \dots, d/2$}
        \State $\hat{C}_1[i] \gets \begin{cases} 
        0 & \text{if } \boldsymbol{\zeta}[i] \geq 0 \\ 
        1 & \text{if } \boldsymbol{\zeta}[i] < 0 
        \end{cases}$
    \EndFor
    
    \State $\hat{r} \gets \mathsf{PRC}^{-1}(k_1, \hat{C}_1)$ 
    \State $\hat{\mathbf{B}} \gets \mathsf{SRandOrthonormal}\left(\frac{d}{2}; \hat{r}\right)$ 
    \State Form transpose block matrix $\mathbf{M}^\top = \begin{pmatrix} 
    \mathbf{I}_{d/2} & \mathbf{0} \\ 
    \mathbf{0} & \hat{\mathbf{B}}^\top 
    \end{pmatrix}$
    \State $\hat{\mathbf{X}} \gets \boldsymbol{\zeta} \mathbf{M}^\top$ \label{alg_decode_line_mult2}
    
    \Statex \Comment{Extract signs to recover $\hat{C}_2$}
    \For{$j = 1, \dots, d/2$}
        \State $\hat{C}_2[j] \gets \begin{cases} 
        0 & \text{if } \hat{\mathbf{X}}[d/2 + j] \geq 0 \\ 
        1 & \text{if } \hat{\mathbf{X}}[d/2 + j] < 0 
        \end{cases}$
    \EndFor
    
        \Statex \Comment{$d_H(\cdot, \cdot)$ = Hamming distance}
\State \Return $1(d_{H}(\hat{C}_2, f_{k_3}(\hat{r})) < t)$
        \label{alg_decode_return}
\EndProcedure

\end{algorithmic}
\end{algorithm}

\subsubsection{Encoding Multi-bit Messages}
The construction described in Algorithms \ref{alg:keygen}, \ref{alg:sample}, and \ref{alg:decode_detect} embeds a single-bit watermark (i.e., a binary indicator of whether a sample is watermarked). We can easily extend this construction to accommodate multi-bit messages. Let $\mathcal{M}$ denote a set of allowable messages, and let $\texttt{ECC.Encode}$ and $\texttt{ECC.Decode}$ represent the encoding and decoding functions of an error-correcting code.
\paragraph{Sample}
To encode a message $m \in \mathcal{M}$, replace Line \ref{alg_line_sample_prf} in Algorithm \ref{alg:sample} with $C_2 \gets f_{k_3}(r) \oplus \texttt{ECC.Encode}(m)$. Because $f_{k_3}(r)$ is a pseudorandom sequence, this XOR operation effectively acts as a one-time pad. This encrypts the ECC codeword, ensuring that the embedded sequence remains pseudorandom and preserves the statistical properties of the baseline scheme.

\paragraph{Decode} To recover the message, replace Line \ref{alg_decode_return} of Algorithm \ref{alg:decode_detect} with return $\texttt{ECC.Decode}(\hat{C}_2 \oplus f_{k_3}(\hat{r}))$.

If a multi-bit message is used, the FPR rates will depend on the configuration of the error correcting code scheme.


\subsection{Per-sample Transformations}
As shown in Section \ref{subsection:attack_on_defense}, relying solely on the fixed transformation $A$ allows a REM2 adversary to reconstruct the transformation. To prevent the adversary from learning $A$, we enforce two requirements:
\begin{enumerate}
    \item \textbf{Per-Sample Transformations}: The applied transformation must remain strictly unique for each generated sample.
    \item \textbf{Obfuscated Fixed Transformation}: The adversary's interaction with the system should not reveal any information about the fixed transformation. 
\end{enumerate}
To analyze the adversary's capabilities, let us write $A$ as
\[
A=
\begin{pmatrix} A_1 & A_2 \\ A_3 & A_4 \end{pmatrix}
\] 
As per the $\texttt{Sample}(\lambda, d, k)$ function in Algorithm \ref{alg:sample},  watermarked sample $Y_i$ can be rewritten as $Y_i = X_i {\Lambda}_i$, where the effective orthonormal transformation matrix is:
\begin{equation} \label{eqn:mask}
{\Lambda}_i = \begin{pmatrix} {I}_{d/2} & \mathbf{0} \\ \mathbf{0} & {B}_i \end{pmatrix} \begin{pmatrix} A_1 & A_2 \\ A_3 & A_4 \end{pmatrix} = \begin{pmatrix} A_1 & A_2 \\ {B}_i A_3 & {B}_i A_4 \end{pmatrix}
\end{equation}

\subsubsection{Detector}
For any adversarial perturbation ${\eta}$ added to ${Y}_i$, the detector evaluates ${X}_i + {\eta} \Lambda_i^\top$. Define the detector as follows:
\begin{equation}
    \mathcal{O}(Y) = \texttt{Decode}(\lambda, Y, k, t) .
\end{equation}

Let $d_H(\cdot, \cdot)$ denote the Hamming distance, and the codewords $(C_1, C_2)$ be as defined in Algorithm \ref{alg:sample}. If $\widetilde{C}_1$ and $\widetilde{C}_2$ represent the codewords recovered from the perturbed sample ${Y}_i + \eta$, we can formally define the indicator variables $\alpha$ and $\beta$ as:
\begin{equation}
\alpha = 1\left[ \frac{d_H(\widetilde{C}1, C_1)}{|C_1|} > \tau_1 \right], \quad \beta = 1\left[ \frac{d_H(\widetilde{C}2, C_2)}{|C_2|} > \tau_2 \right],
\end{equation}
where $\alpha = 1$ denotes the event where more than $\tau_{1}$ proportion of signs in $C_1$ are flipped, and likewise for $\beta = 1$. Conceptually, we can view $\tau_1$ as the maximum proportion of bit flips the PRC codeword $C_1$ can sustain while still being successfully authenticated. Similarly, $\tau_2$ represents the tolerable error rate for the $C_2$ portion, beyond which the decoder fails and returns 0.
Therefore, we can express 
\begin{equation} \label{eqn:O_behavior}
    \mathcal{O}(Y) = ((\alpha = 0) \wedge (\beta = 0)).
\end{equation}

\subsubsection{Requirement on Error Thresholds}
We require the error thresholds $\tau_1$ and $\tau_2$ to satisfy:
\begin{equation} \label{eqn:threshold_tau}
    \tau_{1} - \tau_{2} > \sqrt{\frac{\ln d \ln \ln d}{d}},
\end{equation}
for sufficiently large $d$.
If $\Lambda_i$ is unknown to the adversary, any chosen perturbation $\eta$ results in an effective noise vector $\eta \Lambda_i^T$ that disperses evenly across the $d$-dimensional space, yielding uniform expected error rates across both codeword components. By standard concentration bounds, the statistical fluctuation between these components is strictly bounded by the threshold gap. Consequently, any undirected noise strong enough to trigger $\alpha = 1$ will simultaneously trigger $\beta = 1$ with all but negligible probability, preventing the adversary from isolating the state $\alpha=1$ and $\beta=0$. We formalize this intuition and apply this threshold in the security analysis detailed in Section \ref{subsection:secanalysis_non_adaptive}.

\section{Security Analysis} \label{section:sec_analysis}

In this section, we prove that our proposed watermarking scheme achieves IND1-REM2 security, as defined in Section~\ref{subsection:threat_model}. Throughout the following analysis, we assume the scheme is initialized by sampling a secret key via $\mathsf{KeyGen}(\lambda)$.

\subsection{Security Under IND1 (Undetectability) }
To formally establish watermark undetectability, we demonstrate that no PPT adversary can achieve a non-negligible advantage in the IND1 security game (Section \ref{tab:threat-models}). Consistent with standard diffusion models, we let $\mathcal{D} = \mathcal{N}(0, I_d)$.

The function $\texttt{Sample}(\lambda, d, k)$ in Algorithm \ref{alg:sample} generates samples from $\mathcal{D}_k$. The indistinguishability of $\mathcal{D}_k$ from $\mathcal{D}$ follows directly from its construction. First, the codeword $C$ is constructed by concatenating two sequences that are computationally indistinguishable from a truly random binary sequence (Line \ref{alg_line_random_bin}). Next, a sample from $\mathcal{D}_k$ is generated by drawing a base sample from $\mathcal{D}$, preserving its magnitude, but assigning its signs according to $C$ (Line \ref{alg_line_embed}).

Consequently, any PPT adversary which distinguishes $\mathcal{D}_k$ from $\mathcal{D}$ could trivially distinguish $C$ from a uniformly random sequence. Therefore, $\mathcal{D}_k$ is computationally indistinguishable from $\mathcal{D}$, and the scheme is IND1-secure.

\subsection{Watermark Removal and Adversary's Advantage}

Given a watermarked sample $Y$, the adversary's goal is to compute a perturbed sample $\widetilde{Y}= Y + \eta$ for some sufficiently small $\eta$ (e.g., $\|\eta\|_2 \leq \epsilon$ for a predetermined budget $\epsilon$) such that $\mathsf{Decode}(\lambda, \widetilde{Y}, k, t)$ returns \text{False}.

Let $\mathcal{A}$ be a probabilistic polynomial-time (PPT) adversary operating under the REM2 threat model. We evaluate $\mathcal{A}$ at the challenge phase, meaning it has already completed Steps (1)--(3) of the REM2 security game (Section \ref{subsubsection:rem2}). First, we define the success probability of any adversary $\mathcal{A}$, bounded by a perturbation budget $\epsilon$, as:
\begin{equation} \label{eqn:success}
\begin{aligned} 
\text{Succ}(\mathcal{A}, \epsilon)
= \Pr_{\substack{k \leftarrow \mathsf{KeyGen}(\lambda) \\ Y \leftarrow \mathsf{Sample}(\lambda, d, k) \\ \widetilde{Y} \leftarrow \mathcal{A}(Y, \epsilon)}} & \Big[ \mathsf{Decode}(\lambda, \widetilde{Y}, k, t) \\
& \quad = \big( \text{False} \;\land\; \|\widetilde{Y} - Y\|_2 \leq \epsilon \big) \Big]. 
\end{aligned}
\end{equation}

Let $\mathcal{A}_w$ denote a baseline adversary limited to adding isotropic white-noise of magnitude $\epsilon$. We define the advantage of $\mathcal{A}$ w.r.t. the baseline $\mathcal{A}_w$ as:
\begin{equation} \label{eqn:advantage}
\text{Adv}(\mathcal{A}, \epsilon) =  \text{Succ}(\mathcal{A}, \epsilon) - \text{Succ}(\mathcal{A}_w, \sqrt{2} \epsilon) .
\end{equation}

The watermarking scheme is said to be $\sqrt{2}$-REM2-secure if and only if, for all PPT adversaries $\mathcal{A}$,
\[
\text{Adv}(\mathcal{A}, \epsilon) \leq {negl}(\lambda)
\]
where $negl(\cdot)>0$ is a negligible function. 

\subsubsection{The $\sqrt{2}$ factor}
The $\sqrt{2}$ factor establishes a strict upper bound on the adversarial efficiency operating under the REM2 threat model. Our subsequent security analysis reveals that while an adversary can isolate the $d/2$-dimensional subspace associated with either the fixed or dynamic transformation—and restrict their perturbations accordingly—they cannot recover its specific orthonormal basis. By keeping the exact decision boundaries obfuscated, the scheme successfully thwarts targeted boundary attacks (e.g. the sign-flip attack  in \cite{ccs}). As shown later in Theorem \ref{thm:rem2-secure}, this $\sqrt{2}\epsilon$ factor is a tight constant.

\subsection{Security Analysis under REM2} \label{subsection:secanalysis_non_adaptive}

Under the REM2 threat model, an adversary observes a sequence of watermarked samples $\mathbf{Y} = ({Y}_1, \dots, {Y}_t)$ and submits polynomial queries $\mathbf{Q} = (q_1, \dots, q_s)$ to the detector. Let us express each query as $q_j = Y + \eta_j$ for some $Y \in \mathbf{Y}$ and perturbation $\eta_j$. The detector evaluates $\mathcal{O}(q_j)$ and returns $0$ if and only if $\alpha_j = 1$ or $\beta_j = 1$ (as defined in (\ref{eqn:O_behavior})).

To establish REM2 security, we show that the adversary gains negligible information about the static boundary from the oracle, i.e:
\[
\IP[\alpha_j \mid \mathbf{Y}, \mathcal{O}(q_1), \dots, \mathcal{O}(q_j), \eta_j] = \IP[\alpha_j \mid \mathbf{Y}, \eta_j] + negl(d),
\]
where $\IP[\cdot]$ denotes the probability taken over $k \leftarrow \texttt{KeyGen}(\lambda)$, i.e. that the distribution of $\alpha_j$ given the oracle's responses is indistinguishable from its distribution without oracle access.

Our proof proceeds in two steps:
\begin{enumerate}
    \item \textbf{Information-Theoretic Security (Theorem~\ref{thm:rem2-secure}):} We first idealize the \texttt{Sample} algorithm by replacing all pseudorandom primitives (PRC, PRF, and \texttt{SRandOrthonormal}) with perfectly random sources, proving that the scheme is statistically secure under these idealized conditions.
    \item \textbf{Computational Security (Theorem~\ref{thm:computational_rem2}):} Next, a hybrid argument is employed to reduce the security of our real-world instantiation to the computational hardness of the underlying pseudorandom primitives, proving that no PPT adversary can achieve a non-negligible advantage.
\end{enumerate}

\subsubsection{Information-Theoretic Security}

Our proof of information-theoretic security
(Theorem~\ref{thm:rem2-secure}) relies on the following lemmas.

\begin{enumerate}
    \item \textbf{Lemma~\ref{lemma:O_Ob_indist} (Oracle Indistinguishability):}
    Establishes that the real detector is statistically indistinguishable from
    an idealized detector that evaluates only the dynamic component $(\beta)$.
    Hence, the oracle responses reveal only negligible information about the
    static component $(\alpha)$.

    \item \textbf{Lemma~\ref{lemma:basis_ind} (Basis Indistinguishability):}
    Extends Lemma \ref{lemma:O_Ob_indist} to show 
    that the fixed transformation is identifiable only up to an independent
    uniformly sampled rotation. In particular, the adversary obtains no
    non-negligible information about the orientation of the fixed basis.

    \item \textbf{Lemma~\ref{lemma:hidden_basis_isotropization}
    (Hidden-Basis Isotropization):}
    Shows that if a subspace basis is hidden by a uniformly sampled rotation,
    then any perturbation chosen independently of this rotation is transformed
    into isotropic noise with the same projected magnitude. 

    \item \textbf{Lemma~\ref{lemma:sqrt2_baseline_comparison}
    ($\sqrt2\epsilon$ Baseline Comparison):}
    Applies Lemma \ref{lemma:hidden_basis_isotropization} to show that even an
    adversary who perfectly isolates the fixed and dynamic subspaces cannot
    outperform a full-space isotropic white-noise baseline of magnitude
    $\sqrt2\epsilon$, up to negligible error.

    \item \textbf{Lemma~\ref{lemma:most_powerful}
    (Advantage Implies Distinguishability):}
    Any adversary achieving non-negligible advantage over the
    $\sqrt2\epsilon$ baseline must distinguish the true orientation of the
    fixed basis from a uniformly rotated copy. This contradicts
    Lemma~\ref{lemma:basis_ind}, and therefore no such adversary can exist.
\end{enumerate}

At a high level: by Lemma~\ref{lemma:O_Ob_indist} the oracle reveals nothing about the long-term transformation (Figure~\ref{fig:intro_alpha_beta} in Appendix \ref{appendix_proof}), so it is recoverable only up to an unknown rotation (Lemma~\ref{lemma:basis_ind}), while the one-time transformation is re-randomized per sample. Any adversarial perturbation is therefore isotropized into white-noise (Lemma~\ref{lemma:hidden_basis_isotropization}), giving no advantage over the $\sqrt{2}\epsilon$
 baseline (Lemmas~\ref{lemma:sqrt2_baseline_comparison},~\ref{lemma:most_powerful}) and establishing Theorem~\ref{thm:rem2-secure}.

For readability, the detailed proofs of these lemmas are deferred to Appendix \ref{appendix_proof}. We present their formal statements below.


\begin{lemma}[Oracle Indistinguishability] \label{lemma:O_Ob_indist}
    Suppose the threshold gap $\tau_{1} - \tau_{2}$ satisfies (\ref{eqn:threshold_tau}). For a given $Y \sim \mathcal{D}_k$ and any query $q = Y + \eta$, the output distribution of the real oracle $\mathcal{O}(q) = 1((\alpha = 0) \wedge (\beta = 0))$ is statistically indistinguishable from an oracle $\mathcal{O}_3 = 1(\beta = 0)$.
\end{lemma}

\begin{lemma}[Basis Indistinguishability]\label{lemma:basis_ind}
   Let $\Lambda := \Lambda_i$ be as defined in (\ref{eqn:mask}), and let $\Lambda^\top = [F, D]$, where
   
   $$ F = \begin{pmatrix} A_1^\top \\ A_2^\top \end{pmatrix} \quad \text{and} \quad D = \begin{pmatrix} A_3^\top B_i^\top \\ A_4^\top B_i^\top \end{pmatrix}. $$
   
   Here, $F$ and $D$ denote the fixed secret transformation and the one-time masked transformation, respectively. 
Let $\mathcal{F}$ be the set of all bases that are orthogonal to $D$. For any fixed rotation matrix $R \in Haar(d/2)$ and any measurable subset $S \subseteq \mathcal{F}$, the conditional probability over the randomized choice of the secret basis yields:
   \[
   \resizebox{\linewidth}{!}{%
  $ \Big| \IP[F \in S \mid \mathbf{Y}, \mathcal{O}(q_1), \dots] - \IP[F \in S_R \mid \mathbf{Y}, \mathcal{O}(q_1), \dots] \Big| \le \mathrm{negl}(d), $
   }
   \]
   where the set $S_R = \{ s R : s \in S\}$. 
   That is, the secret fixed transformation $F$ is only recoverable up to a random rotation.
\end{lemma}

\begin{lemma}[Hidden-basis isotropization]
\label{lemma:hidden_basis_isotropization}
Let $F\in\mathbb R^{d\times d/2}$ be an orthonormal basis, and
$R \sim Haar(d/2)$ be independent of a perturbation $\eta$. It follows that
\[
    \Delta(\eta F R, \|\eta F\|_2U_{d/2})=0.
\]
That is, an adversary that knows $F$ only up to a rotation $R$ sampled uniformly from $Orth(d/2)$ cannot outperform adding isotropic white-noise of magnitude $\|\eta F\|_2$.  
\end{lemma}

\begin{lemma}[$\sqrt2\epsilon$ Baseline Comparison]
\label{lemma:sqrt2_baseline_comparison}
Suppose an adversary $\mathcal A$ is given the two $d/2$-dimensional
subspaces corresponding to $A_F$ and $A_D$, up to rotations sampled independently from $Haar(d/2)$. Then any perturbation
$\eta$ chosen by $\mathcal A$ with $\|\eta\|_2\leq\epsilon$ has negligible advantage over
full-space isotropic white-noise of magnitude $\sqrt2\epsilon$. That is,
$
    \mathsf{Succ}(\mathcal A,\epsilon)
    \leq 
    \mathsf{Succ}(\mathcal A_w,\sqrt2\epsilon)
$
\end{lemma}

\begin{lemma}[Advantage Implies Basis Distinguishability] \label{lemma:most_powerful}
Continuing from the setting in Lemma \ref{lemma:basis_ind}. If there exists an adversary $\mathcal{A}$ operating under the REM2 threat model such that  $\mathsf{Adv}(\mathcal{A}, \epsilon)$ is non-negligible, then there exists a measurable subset $S \subseteq \mathcal{F}$, a rotation $R \in Haar(d/2)$, and polynomial $p(\cdot)>0$ such that:
\[
\resizebox{\linewidth}{!}{%
  $ \Big| \IP[F \in S \mid \mathbf{Y}, \mathcal{O}(q_1), \dots] - \IP[F \in S_R \mid \mathbf{Y}, \mathcal{O}(q_1), \dots] \Big| \geq \frac{1}{p(d)},$
  }
\]
where $S_R = \{ M R : M \in S\}$. That is, $\mathcal{A}$ can distinguish the specific orientation of the fixed basis $F$ from a rotated counterpart with non-negligible probability.
\end{lemma}

\begin{theorem} \label{thm:rem2-secure}
If the threshold gap $\tau_{1} - \tau_{2}$ satisfies (\ref{eqn:threshold_tau}), then for any PPT adversary $\mathcal{A}$ operating under the REM2 threat model, with access to a polynomial number of watermarked samples and queries to the detector oracle $\mathcal{O}$, 
$
    \mathsf{Adv}(\mathcal{A}, \epsilon) \leq {negl}(d).
$
\end{theorem}

\begin{proof}
We proceed by contradiction. Assume there exists a PPT adversary $\mathcal{A}$ operating under the REM2 threat model with a perturbation budget $\epsilon$ s.t.  $\mathsf{Adv}(\mathcal{A}, \epsilon)$ is non-negligible.

By Lemma \ref{lemma:O_Ob_indist}, the adversary's interaction with the real oracle $\mathcal{O}$ is statistically indistinguishable from interaction with the simulated oracle $\mathcal{O}_3$. Consequently, the oracle responses leak strictly negligible information regarding the static decision boundaries. Building upon this, Lemma \ref{lemma:basis_ind} establishes that given the adversary's view, the exact orientation of the fixed basis $F$ remains statistically indistinguishable from a randomly rotated counterpart $FR$, where $R \sim \mathrm{Haar}(d/2)$. 

However, by Lemma \ref{lemma:most_powerful}, since $\mathsf{Adv}(\mathcal{A}, \epsilon)$ is non-negligible, there exists a PPT distinguisher which distinguishes between $F$ and $FR$ with non-negligible probability, which contradicts Lemma \ref{lemma:basis_ind}. Thus, such an $\mathcal{A}$ does not exist. Thus, for all PPT $\mathcal{A}$, 
$ \mathsf{Adv}(\mathcal{A}, \epsilon) \le \mathrm{negl}(d). $
\end{proof}
\begin{remark}[Tightness of the $\sqrt2$ factor]
The $\sqrt2$ factor is tight for the full-space white-noise comparison. By
Lemma~\ref{lemma:hidden_basis_isotropization}, an adversary who knows the fixed
and dynamic subspaces can choose how to split its perturbation energy
\[
    (a,b), \qquad a^2+b^2\leq \epsilon^2,
\]
but within each subspace the perturbation is still isotropized by the hidden
rotation. Lemma~\ref{lemma:sqrt2_baseline_comparison} proves that an adversary operating within these isolated subspaces has the same success probability of the $\sqrt{2}\epsilon$ white-noise adversary, up to a negligible difference. Lemma~\ref{lemma:most_powerful} then shows that any
adversary exceeding this baseline would contradict
Lemma~\ref{lemma:basis_ind}. Thus, the $\sqrt{2}$ factor is tight. 
\end{remark}

\subsubsection{Computational Security}
Having established information-theoretic security under idealized conditions, we prove in Theorem \ref{thm:computational_rem2} that the watermarking construction detailed in Algorithms \ref{alg:keygen}, \ref{alg:sample}, and \ref{alg:decode_detect} achieves computational security against any PPT adversary. Due to space constraints, we will state Theorem \ref{thm:computational_rem2} here, and the proof is deferred to Appendix \ref{appendix_proof}.
\begin{theorem} \label{thm:computational_rem2}
If the underlying PRC encoder, PRF, and orthonormal matrix generators are secure pseudorandom primitives, and the condition in (\ref{eqn:threshold_tau}) is met, then under the REM2 threat model, for all PPT adversaries $\mathcal{A}$, 
\[
{Adv}(\mathcal{A}, \epsilon) \leq {negl}(\lambda),
\]
where $\lambda$ is the security parameter and $negl(\cdot)$ is a negligible function.
\end{theorem}

\subsection{False Positive Rates}
As outlined in Algorithm \ref{alg:decode_detect}, the decoding function $\mathsf{Decode}(\lambda, \bY, k, t)$ takes a threshold parameter $t$ that governs the false positive rate (FPR) of the detection process. Based on prior work (e.g. \cite{gspp}), for a given threshold $t$ and dimension $d > 0$, the FPR for the watermarking scheme is given by
\begin{equation}
FPR(t,d) = \frac{1}{2^{d/2}} \sum_{i=0}^{t-1} \binom{d/2}{i} \sim \exp\left(-d \left( \frac{1}{2} - \frac{2t}{d}\right)^2 \right),
\end{equation}
where the latter approximation holds for sufficiently large $d$. For a fixed $d$ and predetermined FPR (e.g. $2^{-100}$), the corresponding $t$ can be computed. 

\subsection{Errors Due To Diffusion Inversion}
Existing latent-based watermarking schemes typically rely on diffusion inversion to recover the approximate starting point, and ours is no exception. This inversion is imperfect, e.g. due to null-prompt inversion or limitations of the inversion algorithm. 
Let
$\hat{s} := s + \delta$ denote the recovered latent, where $\delta$ is the
inversion error and $s$ is the original starting point. Since the inversion process never observes $\Lambda = MA$,
$\delta$ is independent of $\Lambda$; hence $\delta\Lambda^{\top}$ is
isotropic white-noise of magnitude $\|\delta\|_2$, perturbing the fixed and
dynamic subspaces symmetrically. Inversion error is therefore statistically
equivalent to the white-noise baseline in our REM2 analysis and can be
absorbed into the detection threshold without affecting
Theorems~\ref{thm:rem2-secure} and~\ref{thm:computational_rem2}. Formally,
this corresponds to adding an isotropic term independent of $\Lambda$ to the
adversary's perturbation in the REM2 game; since the argument of
Lemma~\ref{lemma:O_Ob_indist} depends only on the threshold gap
and the symmetry of the induced flips, the analysis is unchanged. Our security
games (Section~\ref{subsection:threat_model}) are defined over latent perturbations,
consistent with \cite{prc, ccs}; adversarial exploitation of the inversion map
itself lies outside this model and is evaluated empirically in
Section~\ref{section:evals}. Finally, as \texttt{Decode} is agnostic to the
inversion algorithm, one may employ more accurate inversion techniques
(e.g., \cite{cosda, roar}).

\section{Empirical Evaluation} \label{section:evals}

We empirically evaluate the following:

\begin{enumerate}
    \item \textbf{Recovery of the secret transformation:} We demonstrate the effectiveness of our proposed attack against the boundary-hiding defense by measuring the accuracy the recovered secret transformation.
    \item \textbf{Performance on generation quality:} We  demonstrate that our hybrid latent-based watermark preserves the visual quality and text alignment of the generated images compared to unwatermarked baselines.
    \item \textbf{Robustness against perturbations:} We assess the scheme's resilience by measuring watermark recovery rates under benign  perturbations (e.g., JPEG compression and Gaussian noise), and  robustness against recent latent-based watermark removal attacks.
\end{enumerate}
\subsection{Setup}

\subsubsection*{Defense Baselines}
In Section \ref{subsection:expt_attack}, we evaluate the efficacy of the transformation recovery attack (detailed in Section \ref{subsection:attack_on_defense}) against the boundary-hiding defense proposed in \cite{ccs}. Furthermore, in Section \ref{subsection:eval_latent_removal_stealthy}, we assess the robustness of our proposed watermarking scheme against the targeted latent-flipping attack introduced in \cite{ccs}, which is explicitly designed to remove undetectable latent-based watermarks, as well as the latent-space regeneration attack described in \cite{zhao2024invisible}.


\subsubsection*{Dataset and Model Used}
We employed Stable Diffusion 2.1 \cite{rombach2021highresolution} and Stable Diffusion XL models \cite{sdxl}, with their respective default parameters, downloaded from HuggingFace. Images were generated using prompts sourced from the Stable Diffusion Prompt (SDP) dataset \cite{gustavosta_sd_prompts}.
This is consistent with the evaluations performed in prior work \cite{prc, gs, ccs}.

\subsubsection*{Coding Primitives}
We utilized the PRC codeword implementation from \cite{prc}. To support multi-bit message payloads, we applied a Low-Density Parity-Check (LDPC) error-correcting code using the Python \texttt{pyldpc} package, with their suggested default parameters. Similar to \cite{prc, gs, gspp}, we set the message length to be 512-bits in the multi-bit watermarking setting. We empirically verify that these parameters meet the condition (\ref{eqn:threshold_tau}) in Appendix \ref{appendix_im_sample_wm}

\subsubsection*{PRF and \texttt{SRandOrthonormal} Functions}
We employed the HMAC-SHA256 PRF provided in Python's \texttt{hashlib} package. The $\texttt{SRandOrthonormal}$ function  is as described in Appendix \ref{appendix_srandorth}.

\subsubsection*{System Configuration}
A Linux system with the following specifications was used --  Processor: Intel Xeon CPU E5-2620;
RAM: 256GB;
GPU: NVIDIA A100. 

\subsubsection*{Evaluation Metrics}
We use three metrics: \textbf{FID and CLIP scores} for visual quality and text alignment of generated images; \textbf{watermark recovery rate} for robustness under benign perturbations; and \textbf{attack success rate (ASR)} for the effectiveness of targeted latent-based removal attacks.

For the watermark recovery rate and ASR evaluations (Sections~\ref{subsection:eval_benign} \& \ref{subsection:eval_latent_removal_stealthy}), we run the full generation and decoding pipeline: each image is generated from a freshly sampled watermarked latent, and detection follows prior work (e.g. \cite{prc,ccs, gs, gspp}) by inverting under the null prompt (i.e., imperfect inversion) and passing the inverted latent to the detector.

\subsection{Recovery of Secret Transformation} \label{subsection:expt_attack}
As established in Section~\ref{subsubsection:no_operation_basic_attack}, the attack requires $\tilde{\Omega}(d^2)$ oracle queries. We evaluate the boundary-hiding defense for $d \in \{4096, 8192, 16384\}$ and $\tau \in \{0.01, 0.1, 0.3\}$, where $d=16{,}384$ corresponds to Stable Diffusion 2.1, allocating a budget of $d^2$ queries for the non-adaptive blind probing phase (Step~2, Section~\ref{subsection:attack_on_defense}). While the query count is large, the attack's non-adaptivity allows for efficient parallelization. Over 5 GPUs, the $d=16{,}384$ case completes in $\sim 23$ hours. 

To evaluate the accuracy of our attack, we measure how closely our extracted components (the rows of $V \in \mathbb{R}^{d \times d}$) match the target axes of the true secret transformation matrix $A$. We compute the alignment matrix $U = VA^\top$, whose diagonal entries correspond to the cosine similarity between the actual and recovered components. For an ideal extraction where the target axes are recovered perfectly (up to a sign flip), every row in $U$ will contain exactly one entry of $1$ or $-1$, with all other entries being $0$.

We quantify the success of the extraction using $\mu_{align}$, the average cosine similarity between the actual and recovered components, corresponding to the average of the diagonal values of $U$. Note that this is equivalent to Tucker's congruence coefficient, where a standard measure of factor recovery in PCA, for which a value above $0.85$ indicates good recovery and a value above $0.95$ implies near-perfect recovery~\cite{lorenzo2006tucker}. A $\mu_{align}$ score approaching $1$ indicates perfect recovery, whereas an unsuccessful extraction yields a matrix uncorrelated with $A$, causing $\mu_{align}$ to collapse toward $0$. For reference, random guessing yields an expected alignment score of $1/\sqrt{d}$.

The results are summarized in Table~\ref{tab:extraction_results}. Across all evaluated configurations of $d$ and $\tau$, the attack successfully recovers the secret transformation, achieving $\mu_{align} \geq 0.85$. 

We further evaluated our transformation recovery attack on our proposed scheme, and observed $\mu_{align} \leq 0.01$ across all values of $\tau$ and $d$, empirically demonstrating that the attack fails to recover any useful information about the fixed transformation.

\begin{table}
\centering
\caption{Extraction performance of the attack detailed in Section \ref{subsection:attack_on_defense}, evaluated under a fixed budget of $d^2$ queries across various tolerances ($\tau$) and dimensionalities ($d$). The metric $\mu_{align} \in [0,1]$, rounded to 3 decimal places,  denotes the mean maximum absolute alignment score, where $1$ indicates perfect recovery of the secret transformation $A$. The rightmost column reports the proportion of target axes successfully recovered with an individual alignment score exceeding $0.85$.}
\label{tab:extraction_results}
\begin{tabular}{cccc}
\hline
\textbf{$d$} & \textbf{$\tau$} & \textbf{$\mu_{align}$} & \textbf{\begin{tabular}[c]{@{}c@{}}Recovered\\ Axes ($>0.85$)\end{tabular}} \\ \hline
4,096          & 0.01            & 0.999                  & 100\%                                                                       \\
             & 0.10            & 0.998                  & 100\%                                                                       \\
             & 0.30            & 0.996                  & 100\%                                                                       \\ \hline
8,192          & 0.01            & 0.999                  & 100\%                                                                       \\
             & 0.10            & 0.930                  & 100\%                                                                       \\
             & 0.30            & 0.899                 & 100\%                                                                       \\ \hline
16,384          & 0.01            & 0.987                  & 100\%                                                                       \\
             & 0.10            & 0.915                  & 100\%                                                                       \\
             & 0.30            & 0.871                  & 100\%                                                                       \\ \hline
\end{tabular}%
\end{table}

\subsection{Generation Quality Under Stable Diffusion}

To empirically verify whether our watermarked latents introduce any degradation or structural bias during the reverse diffusion process, we evaluate the Fréchet Inception Distance (FID) and CLIP scores against unwatermarked baselines. The initial latents are drawn using the \texttt{Sample} procedure detailed in Algorithm \ref{alg:sample}, and generation is performed using the default hyperparameters of Stable Diffusion 2.1. Due to the substantial computational cost associated with large-scale image generation and metric evaluation, we restrict this specific analysis to SD2.1. Samples of watermarked images generated using SD2.1 and SDXL are given in Appendix \ref{appendix_im_sample_wm}.

Following the evaluation protocol established in \cite{prc}, our dataset comprises 10,000 generated images, obtained by selecting 100 random prompts and generating 100 images per prompt. Table \ref{tab:fid_clip} presents our computed FID and CLIP scores alongside the baseline scores reported in \cite{prc} for both the PRC watermark and the unwatermarked baseline. These results empirically corroborate our theoretical claims: our hybrid watermarking scheme does not degrade visual quality or text alignment, yielding outputs with fidelity identical to those initialized with standard $\mathcal{N}(0, I_d)$ noise.

\begin{table}
\centering
\caption{Mean and standard deviation of FID and CLIP scores, rounded to 4 decimal places. Results demonstrate that our hybrid watermarking scheme achieves performance comparable to the baseline with no watermark, empirically verifying that the watermarking process preserves generation quality.}
\label{tab:fid_clip}
\begin{tabular}{@{}lcccc@{}}
\toprule
 & \multicolumn{2}{c}{\textbf{FID}} & \multicolumn{2}{c}{\textbf{CLIP} } \\ \cmidrule(lr){2-3} \cmidrule(l){4-5} 
\textbf{Watermark} & \textbf{Mean} & \textbf{Std} & \textbf{Mean} & \textbf{Std} \\ \midrule
None               & 63.4625       & 0.2507      & 0.6119        & 0.0018      \\
PRC                & 63.7350       & 0.3511      & 0.6146        & 0.0014      \\
Ours               & 63.8044       & 0.3070      & 0.6099        & 0.0025      \\ \bottomrule
\end{tabular}
\end{table}

\subsection{Performance Against Benign Pixel-Level Perturbations}\label{subsection:eval_benign}
We evaluate the robustness of the watermark under benign yet realistic image perturbations, specifically:
\begin{enumerate}
    \item Additive Gaussian noise with standard deviation parameter $\sigma \in \{0.01, 0.05, 0.10\}$,
    \item JPEG compression with quality factors $q \in \{90, 75, 50\}$.
\end{enumerate}
All pixel values are assumed to lie in the normalized range $[0, 1]$. The evaluation protocol is as follows:

\begin{enumerate}
\item \textbf{Generation:} We randomly sample 100 text prompts and generate a corresponding dataset of 100 watermarked images.
\item \textbf{Perturbation:} The various described perturbations are applied to each image 
\item \textbf{Inversion \& Detection:} Conditional on null prompt, we invert each generated image back into the latent space  and its approximate starting point, which is passed to the detection algorithm to detect the watermark. 
\end{enumerate}

We consider two payload settings:
\begin{enumerate}
    \item \textbf{Single-bit watermark:} Evaluating the binary detection accuracy to determine whether an image is successfully flagged as watermarked.
    \item \textbf{Multi-bit watermark:} Embedding a randomly sampled 512-bit payload into the generated images and evaluating whether the message can be successfully extracted post-perturbation.
\end{enumerate}
The results are given in Table \ref{tab:perturbation_robustness}, which confirm that the proposed watermarking scheme remains highly robust against realistic pixel-level perturbations. Notably, the SDXL model yields slightly higher robustness, which is an expected outcome given its significantly larger ($4\times$) latent space. 
\begin{table}[htbp] 
\centering
\caption{Watermark recovery rates under various image perturbation for single-bit and multi-bit payloads across Stable Diffusion 2.1 and SDXL.}
\label{tab:perturbation_robustness}
\resizebox{\columnwidth}{!}{%
\begin{tabular}{llcccc}
\toprule
\multirow{2}{*}{\textbf{Attack}} & \multirow{2}{*}{\textbf{Param.}} & \multicolumn{2}{c}{\textbf{SD2.1}} & \multicolumn{2}{c}{\textbf{SDXL}} \\
\cmidrule(lr){3-4} \cmidrule(lr){5-6}
 & & \textbf{Single} & \textbf{Multi} & \textbf{Single} & \textbf{Multi} \\
\midrule
Baseline & - & 1.00 & 1.00 & 1.00 & 1.00 \\
\midrule
\multirow{3}{*}{JPEG Comp.} 
 & $q=90$ & 1.00 & 1.00 & 1.00 & 1.00 \\
 & $q=75$ & 0.99 & 0.95 & 1.00 & 0.95 \\
 & $q=50$ & 0.93 & 0.91 & 0.94 & 0.93 \\
\midrule
\multirow{3}{*}{Gaussian Noise} 
 & $\sigma=0.01$ & 1.00 & 1.00 & 1.00 & 1.00 \\
 & $\sigma=0.05$ & 1.00 & 0.99 & 1.00 & 1.00 \\
 & $\sigma=0.10$ & 0.96 & 0.92 & 0.98 & 0.94 \\
\bottomrule
\end{tabular}%
}
\end{table}

\subsection{Evaluation of Latent-Based Watermark Removal Attacks} \label{subsection:eval_latent_removal_stealthy}

\subsubsection{Regeneration Attack  \cite{zhao2024invisible}}
To assess robustness against generative adversaries, we employ the regeneration attack. This method targets invisible watermarks by injecting noise into the latent space and reconstructing the image via the denoising process of Latent Diffusion Models (LDMs). We include this attack because its latent-based perturbation serves as an established benchmark in recent works (e.g., \cite{prc}). To rigorously evaluate our proposed scheme, we tested it across varying attack intensities, specifically applying diffusion noise steps of 50, 60 (the attack's default), and 70. Empirically, our watermark successfully survives all configurations for both single and multi-bit setting, demonstrating strong resilience against LDM-driven image reconstruction.

\subsubsection{``Stealthy'' Removal Attack  \cite{ccs}}
We selected this  attack because it is specifically formulated to target the removal of undetectable latent-based watermarks in diffusion models. Following their evaluation protocol, the adversary's perturbation budget is strictly constrained to specific $\ell_2$-norms. To empirically measure the success rate of the removal attack, we perform the following:
\begin{enumerate}
\item \textbf{Generation:} We randomly sample 100 text prompts and generate a corresponding dataset of 100 watermarked images.
\item \textbf{Inversion:} We invert the generated image back to its original starting point.
\item \textbf{Perturbation:} We apply the respective attacks to the inverted latent representation within the bounded $\ell_2$ budget, proceeding as follows:
\begin{enumerate}
    \item For the {white-noise attack}, we inject isotropic Gaussian noise scaled to match the specified $\ell_2$ magnitude.
    \item For the {``stealthy'' attack}, we iteratively flip the signs of the latent components with the smallest magnitudes until the $\ell_2$ perturbation budget is exhausted.
\end{enumerate}
\item \textbf{Detection:} Finally, we pass the starting point through the detection algorithm to detect the watermark. 
\end{enumerate}

Similar to the experimental setup in \cite{ccs}, evaluations are performed on SD2.1. 
As demonstrated in Table \ref{tab:prc_removal_attack}, the targeted ``stealthy'' attack fails to outperform the white-noise adversary. This corroborates our theoretical guarantees in Theorem \ref{thm:rem2-secure}, demonstrating that no PPT adversary operating under the REM2 threat model can achieve a non-negligible advantage over random white-noise injection.

\begin{table}
\centering
\caption{Performance of the ``stealthy'' removal attack \cite{ccs} against an isotropic white-noise baseline on the proposed watermarking scheme under varying $\ell_2$ perturbation budgets. Values represent the attack success rate.}
\label{tab:prc_removal_attack}
\small 
\setlength{\tabcolsep}{4pt} 
\begin{tabular}{@{}lcccc@{}}
\toprule
 & \multicolumn{2}{c}{\textbf{Single-bit}} & \multicolumn{2}{c}{\textbf{Multi-bit}} \\ \cmidrule(lr){2-3} \cmidrule(l){4-5} 
\textbf{Noise} & \textbf{White-} & \textbf{``Stealthy''} & \textbf{White-} & \textbf{``Stealthy''} \\ 
\textbf{Level} & \textbf{Noise}  & \textbf{Attack}       & \textbf{Noise}  & \textbf{Attack}       \\ \midrule
8              & 0.00            & 0.00                  & 0.00            & 0.00                  \\
10             & 0.01            & 0.01                  & 0.01            & 0.01                  \\ 
30             & 0.08            & 0.08                  & 0.17            & 0.19                  \\\bottomrule
\end{tabular}
\end{table}

\section{Discussion}
\subsection{Computational Limitations and Tradeoffs}

\paragraph{Computational Overhead}
Similar to the boundary-hiding defense in \cite{ccs}, our approach introduces computational overhead (see Appendix \ref{appendix_performance_overhead}). Memory requirements increase by approximately 25\% compared to the boundary-hiding baseline due to an additional $(d/2) \times (d/2)$ one-time transformation matrix. Taken over 100 trials, the sample and decode algorithms took  0.1s and 0.2s respectively in SD 2.1, and 1.2 s and 1.9s in SDXL. Currently, our theoretical security guarantees depend on dense orthonormal transformations during the \texttt{Sample} and \texttt{Decode} steps. A possible approach is to explore structured or sparse transformations (e.g \cite{cheon2003sparse}) to reduce computational cost without compromising security. Nonetheless, this overhead is incurred only once per image and remains a fraction of the generation time itself (2.5s and 7.3s for unwatermarked generation in SD2.1 and SDXL, respectively).

\paragraph{Overhead of Secret Transformation}
As noted in Section~\ref{subsection:expt_attack}, recovery requires $d^2$ samples, which is expensive but parallelizable efficiently across multiple GPUs. Thus, motivated and computationally powerful adversary may parallelize queries across multiple black-box detector accounts, and recover the secret transformation for large $d$. Since the secret transformation is long-term, it only needs to be broken once to re-expose the detector's decision boundary, and all subsequent removal attacks can be carried out cheaply. The cost of recovering the secret transformation is therefore amortized across an unlimited number of subsequent attacks, constituting a complete break of the boundary hiding defense.

\subsection{Error-Correcting Codes for Multi-Bit Watermarking}

For our multi-bit watermark implementation, we utilize Low-Density Parity-Check (LDPC) codes, with recommended parameters from prior work. Because our hybrid scheme imposes a global, one-time orthogonal transformation, any targeted perturbation crafted by an adversary in the latent space is uniformly dispersed across all $d$ dimensions during the extraction phase. Consequently, the ideal Error-Correcting Code (ECC) for our framework must be resilient to independent, uniformly distributed noise rather than localized corruptions.
Although a comprehensive comparative analysis of ECCs falls outside the primary scope of this paper, our watermarking framework is entirely code-agnostic. Exploring advanced coding schemes to further maximize payload capacity and robustness remains a promising direction for future work.

\subsection{A Stronger Adaptive Adversary: REM3 Threat Model} \label{subsection:rem3}

A natural follow-up is to consider an even more powerful adversary, with \textit{adaptive} oracle access to the detector $\mathcal{O}$, allowing them to craft queries based on past detector responses. Let us call this the \textit{REM3} threat model. 

Let us reuse the formulation from Lemma \ref{lemma:basis_ind}, denote the hybrid transformation as $[F, D]$, where $F$ and $D$ correspond to the fixed and dynamic basis vectors, respectively. The REM3 adversary could proceed as follows:
\begin{enumerate}
    \item \textbf{Dynamic Basis Recovery:} The adversary executes the baseline attack described in Section \ref{subsection:attack_on_defense} to approximate $D$.
    \item \textbf{Subspace Isolation:} With $D$ known, the adversary constrains subsequent queries strictly to the subspace orthogonal to $D$. This effectively zeroes out the dynamic component of the transformation.
    \item \textbf{Fixed Basis Deduction:} By isolating the fixed transform in this manner, the adversary with sufficient queries is likely able to approximate $F$.
\end{enumerate}

Nonetheless, the information-theoretic bound in Section \ref{subsubsection:no_operation_basic_attack} still applies: learning a dense $d/2 \times d/2$ transformation still requires $\Omega(d^2)$ queries, which although computationally expensive, is a substantial improvement over random guessing. A REM3-secure scheme likely requires eliminating the fixed transformation entirely. This introduces non-trivial complications regarding how the dynamic transformation can be reliably encoded and extracted. We leave this as future work. 

\paragraph{Deployment-Level Mitigations Against Adaptivity}
While we do not claim a cryptographic guarantee against REM3, practical deployment constraints substantially raise the bar for this adversary. The power of an adaptive strategy stems from \emph{sequential dependence}: each query $q_{j+1}$ depends on the response to $q_j$, forming a dependency chain that, unlike a non-adaptive batch, has limited parallelization. Consequently, defenses commonly used in real-world query-access APIs against adversarial perturbations, such as per-account rate limiting \cite{douceur2002sybil}, query latency, and stateful detection \cite{li2022blacklight}, though these approaches target different attack surfaces and provide different types of guarantees. A REM3-secure scheme remains an important direction for future work. 

\section{Related Works}

\subsection{AI-Generated Content Watermarking}

Current approaches to watermarking AIGC can generally be categorized into the following categories \cite{zhao2025sok}: post-processing, in-process, and pre-processing techniques.

\paragraph{Post-processing Watermarks}
These methods apply watermarks directly to the generated media, similar to traditional digital watermarking \cite{zhao2025sok, Rahman_2013}. Recent approaches adopt learning-based frameworks, utilizing trained neural networks to both inject and extract watermarks from content \cite{Jia_2021, dnn_watermark_survey}.
There has also been recent work exploring techniques to embed watermarks into the VAE latent prior to the final upscaling step \cite{lee2026phasemarkposthocoptimizationfreewatermarking, guo2024freqmarkinvisibleimagewatermarking}.

\paragraph{In-process Watermarks}
This category modifies the generative diffusion process itself. Techniques include altering specific model weights \cite{tondi2024robustlargepayloaddnnwatermarking}, injecting watermark signals directly during the denoising phase \cite{gesny2025guidancewatermarkingdiffusionmodels}, or integrating plugins (e.g. LoRA) into the diffusion model \cite{lin2025efficientwatermarkingmethodlatent}. Typically, these approaches rely on a separately trained watermark detector models.

\paragraph{Pre-processing Watermarks}
These methods operate before the core generation begins by manipulating the initial conditions. A prominent approach involves sampling from a specifically watermarked initial distribution (e.g., Tree-Ring \cite{treering}, Gaussian Shading (GS) \cite{gs}, GS++ \cite{gspp}, and PRC \cite{prc}), which is then verified via a detection algorithm to determine if a sample originated from that distribution. 
Closely related are seed-based techniques that manipulate the initial noise vector. For instance, NoisePrints derives the starting seed using a secret key \cite{goren2026noiseprintsdistortionfreewatermarksauthorship}, whereas other proof-of-authorship schemes select the seed by passing authorship related information into a cryptographic hash \cite{lee2026proofofauthorshipdiffusionbasedaigenerated}.

\subsection{Attacks on Watermarking}

In response to large number of work on watermarking, there has been increasing interest in attacking AIGC watermarks \cite{zhao2025sok}. For example, adversaries have successfully leveraged LDMs to erase pre and post-processing watermarks \cite{zhao2024invisible}, and there have been work directly targeting specific watermarking algorithms (e.g. \cite{hu2024stablesignatureunstableremoving}). On the other hand, \cite{ccs} demonstrated that undetectable latent-based watermarks are still susceptible to removal attacks, proving that undetectability does not guarantee unremovability. 

Besides removal, watermarking schemes must withstand ambiguity attacks, where adversaries exploit coincidental statistical artifacts to falsely claim ownership \cite{ambiguity_1, ambiguity_2}.  DNN-based methods mitigate this using weight-embedded "passports" \cite{fan2019extendedversionrethinkingdeep}, which act as a digital signature to authenticate watermarks. In general, the security of latent-based schemes ultimately hinges on minimizing their False Positive Rate (FPR) to prevent attackers from finding valid watermarks through random guessing \cite{ambiguity_fpr}.


\section{Conclusion}

We showed that an adversary with batched detector access can reconstruct the secret transformation underlying existing boundary-hiding defenses. In response, we proposed a hybrid long-term/one-time transformation scheme that is provably secure under this stronger threat model. That is, no PPT adversary removes the watermark with non-negligible advantage over an isotropic white-noise adversary granted a $\sqrt{2}\times$ larger budget, and the scheme remains undetectable. Empirically, it preserves generation quality, withstands benign perturbations (e.g. JPEG compression and Gaussian noise), and is robust against the evaluated latent-based removal attacks.

\section{Ethics Considerations}

This research characterizes vulnerabilities in existing latent-based watermarking schemes and proposes a new watermarking scheme. Because our work involves the disclosure of an attack capable of removing undetectable watermarks, we have conducted a stakeholder-based ethics analysis in accordance with the principles outlined in NDSS's ``call for papers'' and the Menlo  Report.

\subsection{Stakeholders} 
We identify three primary stakeholders impacted by our research and its publication:
\begin{enumerate}
    \item \textbf{AIGC providers:} Watermarking schemes and possible removal attacks are of interest in provenance and misuse tracking. Furthermore, as AIGC provenance is a rapidly developing area, watermarking AIGC may be necessary as part of regulatory compliance (e.g. EU AI Act \cite{euaiact2024}).
    \item \textbf{Model users:} They may rely on such watermarking tools to establish provenance or authorship
    \item \textbf{The general public:} The public may rely on tools that use watermarks (such as C2PA) to establish content provenance. Removal attacks may undermine public confidence in these tools.
\end{enumerate}

\subsection{Impacts}

\paragraph{Ethical Principles} We prioritize Beneficence and Respect for Law and Public Interest. We determine that the security posture of the public and developers is best served by proactively identifying the limitations of boundary-hiding defenses before they are widely adopted and simultaneously providing a provably secure alternative.

\paragraph{Potential Harms} By demonstrating that adversaries with batched detector access can reconstruct secret transformations and recover detection boundaries, we introduce a method that could theoretically  be used to remove watermarks from watermarked content, thereby complicating claims which rely on watermarking (e.g. misuse protection)

\subsection{Mitigations and Decision} 
We concurrently propose and validate a robust countermeasure. By introducing a resilient scheme based on a hybrid of long-term and one-time secret transformations, we prevent the reuse of structure across queries and effectively neutralize the attack we disclose. Since the paper provides a defense against the proposed attack, and aims to further the understanding of this topic, we are of the opinion that the benefits greatly outweigh  potential harms.

\bibliographystyle{IEEEtran}
\bibliography{sample-base}

\appendices


\section{Implementation of $\texttt{SRandOrthonormal(d, r)}$}
\label{appendix_srandorth}

Recall from Section~\ref{subsection:watermark} that the watermark $\texttt{Sample}$ 
and $\texttt{Decode}$ functions assume that given a pseudorandom string 
$r$, the output of $\texttt{SRandOrthonormal(d, r)}$ is 
computationally indistinguishable from a matrix sampled from $\mathrm{Haar}(d)$.

The algorithm of~\cite{mezzadri2007generaterandommatricesclassical} converts a 
$d \times d$ matrix of i.i.d.\ $\mathcal{N}(0,1)$ entries into an exactly 
$\mathrm{Haar}(d)$-distributed sample as follows:

\begin{enumerate}
    \item Sample a $d \times d$ matrix $M$ where each entry is i.i.d.\ 
    $\mathcal{N}(0,1)$.
    \item Perform QR decomposition $M = QR$, where $Q$ is orthogonal and $R$ is 
    upper triangular.
    \item Output $U = Q \cdot D$, where 
    $D = \mathrm{Diag}\!\left(\frac{R_{ii}}{|R_{ii}|}\right)$ and $R_{ii}$ denotes 
    the $i$-th diagonal entry of $R$.
\end{enumerate}

The Box--Muller transform~\cite{box1958note} is a deterministic algorithm that 
converts pairs of independent uniform samples on $(0,1)$ into pairs of independent 
standard Gaussian samples. Let $\mathsf{Gauss}: \{0,1\}^\lambda \to \mathbb{R}^{d^2}$ denote the function that 
interprets $r$ as a sequence of uniform floats on $(0,1)$ and applies the Box--Muller 
transform to return $d^2$ approximately $\mathcal{N}(0,1)$ samples. When $\lambda$ is 
insufficient to supply $d^2$ floats directly, $r$ is first expanded using a secure 
pseudorandom generator before applying the transform.
We implement $\texttt{SRandOrthonormal(d, r)}$ as follows:

\begin{enumerate}
    \item Compute $Z = \mathsf{Gauss}(r)$ and arrange the $d^2$ entries into a 
    $d \times d$ matrix $M$.
    \item Apply the algorithm of~\cite{mezzadri2007generaterandommatricesclassical} 
    to $M$ and output the resulting $U$.
\end{enumerate}

We now argue that this implementation satisfies the required indistinguishability 
property. Since $r$ is pseudorandom, $\mathsf{Gauss}(r)$ is computationally 
indistinguishable from $\mathsf{Gauss}(U)$ where $U \leftarrow \{0,1\}^\lambda$ is 
truly uniform, as $\mathsf{Gauss}$ is a deterministic polynomial-time function. 
Furthermore, provided the floating-point precision is sufficiently large, the 
statistical distance between the discretized output of $\mathsf{Gauss}(U)$ and the 
continuous Ginibre ensemble is negligible. Therefore, $M$ is computationally 
indistinguishable from a truly Ginibre-distributed matrix.

It then follows from the lemma below that the output $U$ is computationally 
indistinguishable from $\mathrm{Haar}(d)$, as required.

While the following result is implicit in standard arguments, we include it for 
completeness as we are not aware of a prior reference that states it explicitly.

\begin{lemma}
\label{lemma:haar}
If $M$ is computationally indistinguishable from the Ginibre ensemble, then the 
output of $\texttt{SRandOrthonormal}(d, r)$ is computationally indistinguishable 
from $\mathrm{Haar}(d)$.
\end{lemma}

\begin{proof}
The proof follows by a cryptographic reduction. Suppose for contradiction that a PPT 
adversary $D$ could distinguish the output of $\texttt{SRandOrthonormal}(d, r)$ from 
$\mathrm{Haar}(d)$ with non-negligible advantage. We construct an adversary $B$ that 
distinguishes $M$ from a truly Ginibre-distributed matrix as follows. $B$ receives a 
challenge matrix $M'$, applies the algorithm 
of~\cite{mezzadri2007generaterandommatricesclassical}, and passes the result to $D$. 
Because this transformation maps a truly Ginibre-distributed $M'$ exactly to 
$\mathrm{Haar}(d)$ by~\cite{mezzadri2007generaterandommatricesclassical}, $B$'s 
distinguishing advantage is identical to $D$'s, which is non-negligible by 
assumption. This contradicts the assumption that $M$ is computationally 
indistinguishable from a truly random Ginibre sample.
\end{proof}

\section{Deferred Proofs from Section \ref{section:sec_analysis}} \label{appendix_proof}

\subsection{Indistinguishability of Oracles}
Figure \ref{fig:intro_alpha_beta} illustrates the intuition behind Lemma \ref{lemma:O_Ob_indist}. Whenever the threshold gap condition (\ref{eqn:threshold_tau}) holds, any perturbation strong enough to trigger $\alpha = 1$ triggers $\beta = 1$ with all but negligible probability. Hence the real oracle $O$, which rejects whenever $(\alpha = 0) \wedge (\beta = 0)$ fails, is statistically dominated by the value of $\beta$ alone, and is therefore computationally indistinguishable from the simplified oracle $O_3$, which checks only $\beta$. Since $O_3$ carries no dependence on the long-term transformation $A$, this indistinguishability prevents the adversary from extracting any information about $A$ through detector queries.

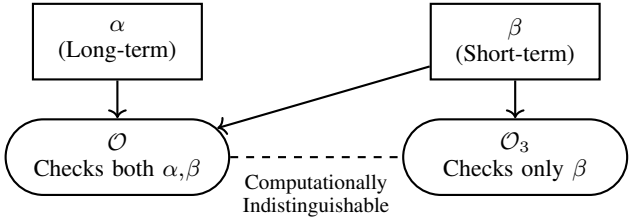
\begin{figure}
\begin{tikzpicture}[
  box/.style={rectangle,draw,thick,text width=2cm,align=center,font=\small,minimum height=1cm},
  detector/.style={rounded rectangle,draw,thick,text width=2.5cm,align=center,font=\small,minimum height=1cm},
]


\node[box] (alpha) {$\alpha$\\(Long-term)};
\node[box,right=3cm of alpha] (beta) {$\beta$\\(Short-term)};

\node[detector,below=.5cm of alpha] (det_O) {$\mathcal{O}$\\Checks both $\alpha, \beta$};
\node[detector,below=.5cm of beta] (det_O3) {$\mathcal{O}_3$\\Checks only $\beta$};

\draw[->,thick] (alpha) -- (det_O);
\draw[->,thick] (beta) -- (det_O);
\draw[->,thick] (beta) -- (det_O3);

\draw[dashed,thick] (det_O.east) -- (det_O3.west);
\node[align=center,font=\footnotesize,below=0.1cm of $(det_O.east)!0.5!(det_O3.west)$] {Computationally\\Indistinguishable};

\end{tikzpicture}

\caption{The hybrid watermark masks components with long-term ($\alpha$) and one-time per-sample ($\beta$) 
transformations. The watermark detector $\mathcal{O}$ (which checks both long-term $\alpha$ and short-term $\beta$) is computationally indistinguishable 
from $\mathcal{O}_3$ (only $\beta$), preventing adversaries from learning the fixed long-term 
transformation. Consequently, no polynomial-time adversary achieves non-negligible advantage 
in removing watermarks from fresh samples.}\label{fig:intro_alpha_beta}
\end{figure}

\subsection{``Well-Behaved'' Detector} \label{subsection:well_behaved}
Let $C$ be a valid codeword (Algorithm \ref{alg:sample}, Line \ref{line:codeword}) and let $d_H$ denote the Hamming distance. A detector $\mathcal{O}$ satisfies the ``well-behaved'' property if, for all $\hat{C}$:
\[
d_H(C, \hat{C}) \leq \tau \implies \mathcal{O}(C) = \mathcal{O}(\hat{C}).
\]

\paragraph{Necessity of the ``Well-Behaved'' Property}
To illustrate the necessity of this property, let $\mathcal{O}$ be a well-behaved detector, and consider an alternative detector $\mathcal{O}'$ defined as:
\[
\mathcal{O}'(\hat{C}) = 
\begin{cases} 
    0, & \text{if a designated bit pattern } \Gamma \text{ is present in } \hat{C}, \\
    \mathcal{O}(\hat{C}), & \text{otherwise,}
\end{cases}
\]
for some pattern $\Gamma$ where $|\Gamma| \leq \tau$. 
An adversary can ``backdoor'' $\mathcal{O}'$ by embedding $\Gamma$ into $C$ using at most $|\Gamma|$ bit flips. This yields a modified codeword $\hat{C}$ satisfying $d_H(C, \hat{C}) \leq \tau$, yet $\mathcal{O}'(C) \neq \mathcal{O}'(\hat{C})$. The ``well-behaved'' property was utilized in the security analysis of \cite{ccs}, and we likewise assume that our detector $\mathcal{O}$ exhibits this property.

\subsection{Deferred Proofs }

Let $U_{d/2}$ denote a uniformly random vector on the unit sphere
$\mathbb S^{d/2-1}$, and $Orth(m)$ be the orthonormal group in $\mathbb{R}^m$. For adversary $\mathcal A$, define
\[
    Z_{\mathcal A}
    :=
    \big(
        Y_w, Y_{nw},\,
        q_1,\ldots,q_s,\,
        \mathcal O(q_1),\ldots,\mathcal O(q_s),\,
        Y
    \big),
\]
as the view of $\mathcal{A}$ at the challange phase (step 4) of the REM2 game, where the variables 
$Y_w, Y_{nw},\,
        q_1,\ldots,q_s,\,
        \mathcal O(q_1),\ldots,\mathcal O(q_s),\,
        Y$
are as defined in the game.

\begin{lemma} \label{lemma:alpha_non_negl_beta_1}
Suppose  $\tau_1 - \tau_2$ satisfies (\ref{eqn:threshold_tau}) and let $Y \sim \mathcal{D}_k$. For all $\eta \in \mathcal{R}^d$, if there exists a polynomial $p(\cdot)$ such that $\Pr[\alpha_j = 1 \mid q_j = Y + \eta_j] \geq 1/p(d)$ for sufficiently large $d$, it implies that $\Pr[\beta = 1 \mid q_j = Y + \eta_j] = 1 - negl(d)$, where ${negl}(\cdot)$ is a negligible function.
\end{lemma}
\begin{proof}
    Let $Y = XA + \eta$, where $X \sim \mathcal{N}(0, I_d)$, $A \sim Haar(d)$. For convenience, let $\alpha := \alpha_j$, $\beta := \beta_j$ and $\IP[\cdot | q_j = Y + \eta_j] = \IP[\cdot | q_j]$. Then, the detector $\mathcal{O}(X + \eta A^T)$ returns True if and only if $\alpha = 1$ or $\beta = 1$. Let $I_i = 1(\text{Sign}(X_i) \neq \text{Sign}(X_i + u_i))$, where $U:=(u_1, \dots, u_d) = \eta A^T$. Thus, the random variables 
    \[
    N_1 = \sum_{i=1}^{d/2} I_i \text{ and } N_2 = \sum_{i=d/2 + 1}^d I_i
    \]
    denote the number of sign flips in the first and second half of $X$ due to $\eta$, respectively. Let $I = \{I_1, \dots, I_d\}$. 

    First, we prove that $N_1$ and $N_2$ share the same distribution. Since $U$ is rotationally invariant (i.e., for any orthonormal matrix $R$, since $A$ and $AR$ share the same distribution, $U$ and $UR$ share the same distribution), it follows that $U$ is exchangeable. That is, for any permutation $\pi : \{1, \dots, d\} \rightarrow \{1, \dots, d\}$, both
    \[
    (u_1, \dots, u_d) \text{ and } (u_{\pi(1)}, \dots, u_{\pi(d)})
    \]
    share the same distribution. Therefore, for any partitions $I'$, $I''$ of $I$ where $|I'| = |I''| =\frac{d}{2}$, it follows that
    \[
    \sum_{i \in I'} I_i \text{ and } \sum_{i \in I''} I_i
    \] 
    share the same distribution. It thus follows that $N_1$ and $N_2$ follow the same distribution.

    With this established, we are ready to prove that 
    \begin{equation} \label{eqn:target_lemma_imply}
        \left(\IP[\alpha = 1|q_j] \geq \frac{1}{p(d)}\right) \Rightarrow \left( \IP[\beta=1|q_j] = 1-negl(d)\right)
    \end{equation}
    Since $N_1$ and $N_2$ share the same distribution, they have the same expectation $\mu := \IE[N_1 |q_j] = \IE[N_2 |q_j]$.
    
    To bound the concentration of $N_1$, we condition on the rotation $U$. Conditioned on $U$, the coordinates $X_i$ remain independent, and hence $I_1, \dots, I_{d/2}$ are independent 0-1 random variables. By Hoeffding's inequality, we have
    \begin{equation} \label{eqn:hoeffding_1}
        \begin{split}
            \IP[\alpha=1|q_j] &= \IP[N_1 > q | \eta] \\ 
            &= \IP[N_1 > \mu + t |q_j] \text { \quad where $t:= q-\mu$} \\
            &\leq \exp \left(- \frac{2t^2}{d/2} \right) 
        \end{split},
    \end{equation}
    where the first line is due to the ``well-behaved'' property (as described in Section \ref{subsection:well_behaved}). 
    Therefore, for $\IP[\alpha = 1|q_j] \geq \frac{1}{p(d)}$ to hold, the following must hold:
    \[
    \frac{1}{p(d)} \leq \exp \left(- \frac{4t^2}{d} \right).
    \]
     This implies that $t \leq \sqrt{\frac{d}{4} \ln p(d)}$. Since $\ln p(d) \in \mathcal{O}(\ln d)$, we have $q = \mu + t \leq \mu + \mathcal{O}(\sqrt{d \ln d})$.

    Next, recall that $q = \tau_1 \frac{d}{2}$ and $w = \tau_2 \frac{d}{2}$. Since  $\tau_1 - \tau_2 \geq \sqrt{\frac{\ln d \ln \ln d}{d}}$, it follows that $q - w \ge \frac{1}{2}\sqrt{d \ln d \ln \ln d}$. Therefore:
    \[
    \begin{split}
    w &\leq q - \frac{1}{2}\sqrt{d \ln d \ln \ln d}\\
    &\leq \mu + \mathcal{O}(\sqrt{d \ln d}) - \frac{1}{2}\sqrt{d\ln d \ln \ln d} \\
    &= \mu - \sqrt{d \ln d \ln \ln d},
    \end{split}
    \]
    for sufficiently large $d$. 
    By Hoeffding's inequality (similar to (\ref{eqn:hoeffding_1})), we have
    \[
    \begin{split}
        & \IP\left[N_2 \leq \mu - \sqrt{d \ln d \ln \ln d}|q_j\right]\\
        &= \IP[-N_2  \geq - \mu+  \sqrt{d \ln d \ln \ln d} | q_j]\\
        & \leq \exp \left( -\frac{2 \left(\sqrt{d \ln d \ln \ln d}\right)^2}{d/2} \right)   \\
        &= \exp \left( -4(\ln d \ln \ln d) \right) \\
        &= \frac{1}{d^{4\ln \ln d}}
    \end{split}
    \]
    which is negligible in $d$. Thus, 
    \[
      \IP[\beta=1|q_j]= \IP\left[N_2 > \mu - \sqrt{d \ln d \ln \ln d} |q_j\right] = 1 - negl(d),
    \]
    and the claim in (\ref{eqn:target_lemma_imply}) follows. 
\end{proof}

\begin{lemma*} [\emph{Lemma \ref{lemma:O_Ob_indist}}]
    Suppose the threshold gap $\tau_{1} - \tau_{2}$ satisfies (\ref{eqn:threshold_tau}). For a given $Y \sim \mathcal{D}_k$ and any adversarial query $q = Y + \eta$, the output distribution of the real oracle $\mathcal{O}(q) = 1((\alpha = 0) \wedge (\beta = 0))$ is statistically indistinguishable from an oracle $\mathcal{O}_3 = 1(\beta = 0)$.
\end{lemma*}
\begin{proof}

We proceed via a hybrid argument. Recall  from (\ref{eqn:O_behavior})
that for a given $Y \sim \mathcal{D}_k$ and query $q  = Y + \eta$
\[
\mathcal{O}(q) = 1 \iff ((\alpha = 0) \wedge (\beta = 0)).
\]
Let $\eta^*$ be the smallest value of $||\eta||_2$ s.t. there exists a polynomial $p(\cdot)$ where $\Pr[\alpha = 1 \mid \|\eta\| \ge \eta^*] \geq \frac{1}{p(d)}$. 
Consider the following 3 hybrid experiments. 
\begin{enumerate}
    \item $H_1$ (the real system): In this hybrid, the oracle behaves exactly as described above.
    \item $H_2$ (norm bounded perturbation oracle): In this hybrid, the oracle deterministically returns $0$ if $||\eta|| \geq \eta^*$. Else, the oracle behaves as described in (\ref{eqn:O_behavior}). Denote this oracle as $\mathcal{O}_2$
    \item $H_3$ (oracle which ignores $\alpha$). In this hybrid, the oracle returns $1$ if and only if $\beta = 0$, regardless of the value of $\alpha$. Denote this oracle as $\mathcal{O}_3$
\end{enumerate}
Because the adversary's view in each hybrid experiment is entirely determined by the sequence of oracle responses, it suffices to show that the statistical distance between the output distributions of $\mathcal{O}$ and $\mathcal{O}_2$, as well as between $\mathcal{O}_2$ and $\mathcal{O}_3$, is negligible in $d$.

For any query $q_i = X + \eta$, the oracles $\mathcal{O}$ and $\mathcal{O}_2$ differ strictly when $\|\eta\| \ge \eta^*$. By definition of $\eta^*$, $\Pr[\alpha = 1 \mid (\|\eta\| \ge \eta^*)] \ge 1/p(d)$. By Lemma \ref{lemma:alpha_non_negl_beta_1}, it follows that $\Pr[\beta = 0 \mid( \|\eta\| \ge \eta^*)] = negl_1(d)$. Since $\mathcal{O}_2(q_i)$ deterministically outputs $0$ in this regime and the real oracle $\mathcal{O}(q_i)$ outputs $1$ if and only if both components succeed, the probability of divergence is bounded by:
\begin{equation} \label{eqn:oracle_1_2}
\begin{split}
\Pr[\mathcal{O}(q_i)  \neq \mathcal{O}_2(q_i)]  &\le \Pr[(\alpha = 0) \land (\beta = 0) \mid (\|\eta\| \ge \eta^*)] \\
& \le \Pr[\beta = 0 \mid (\|\eta\| \ge \eta^*)]\\
&= negl_1(d) 
\end{split}
\end{equation}

Next, consider the divergence between $\mathcal{O}_2$ and $\mathcal{O}_3$. We analyze this across the two regimes of the perturbation magnitude $||\eta||_2$:

\paragraph{For $||\eta||_2 \ge \eta^*$} The oracle $\mathcal{O}_2(q_i)$ deterministically outputs $0$, while $\mathcal{O}_3(q_i)$ outputs $1$ if and only if $\beta = 0$. By Lemma \ref{lemma:alpha_non_negl_beta_1}, $\IP[\beta = 0 \mid ||\eta||_2 \ge \eta^*] = negl_1(d)$.

\paragraph{For $||\eta||_2 < \eta^*$} The oracle $\mathcal{O}_2(q_i)$ outputs $1$ if and only if $(\alpha = 0) \wedge (\beta = 0)$, whereas $\mathcal{O}_3(q_i)$ outputs $1$ if $\beta = 0$. The oracles differ strictly when the static component fails but the dynamic component does not. By the definition of $\eta^*$, we have $\IP[\alpha = 1 \mid ||\eta||_2 < \eta^*] = negl_2(d)$.

Therefore, the probability that the two oracles disagree on any query $q_i$ is bounded by the union of these negligible events:
\begin{equation} \label{eqn:oracle_2_3}
\begin{split}
    &\IP[\mathcal{O}_2(q_i) \neq \mathcal{O}_3(q_i)] \\
    &\le \IP[\beta = 0 \mid ||\eta||_2 \ge \eta^*] + \IP[\alpha = 1 \mid ||\eta||_2 < \eta^*] \le negl(d) 
\end{split}
\end{equation}
Since the probability of divergence for each query is negligible, the statistical distance between the outputs of $\mathcal{O}_2$ and $\mathcal{O}_3$ over a polynomial number of queries remains negligible in $d$.

To conclude the proof, let us denote $\Delta(\cdot, \cdot)$ as the total variation distance over the output distribution of the input oracles. From (\ref{eqn:oracle_1_2}), we have
\[
\Delta(\mathcal{O}, \mathcal{O}_2) \leq negl(d)
\]
and similarly from  (\ref{eqn:oracle_2_3})
\[
\Delta(\mathcal{O}_2, \mathcal{O}_3) \leq negl(d),
\]
for some negligible function negl. Therefore, 
\[
\Delta(\mathcal{O}, \mathcal{O}_3) \leq 2 \cdot negl(d)
\]
and the claim follows. 

\end{proof} 

\begin{lemma*} [\emph{Lemma \ref{lemma:basis_ind}}]
   Let $\Lambda := \Lambda_i$ be as defined in (\ref{eqn:mask}), and let $\Lambda^\top = [F, D]$, where
   
   $$ F = \begin{pmatrix} A_1^\top \\ A_2^\top \end{pmatrix} \quad \text{and} \quad D = \begin{pmatrix} A_3^\top B_i^\top \\ A_4^\top B_i^\top \end{pmatrix}. $$
   
   Here, $F$ and $D$ denote the fixed secret transformation and the one-time masked transformation, respectively. 
Let $\mathcal{F}$ be the set of all bases that are orthogonal to $D$. For any fixed rotation matrix $R \in Haar(d/2)$ and any measurable subset $S \subseteq \mathcal{F}$, the conditional probability over the randomized choice of the secret basis yields:
   \[
   \resizebox{\linewidth}{!}{%
  $ \Big| \IP[F \in S \mid \mathbf{Y}, \mathcal{O}(q_1), \dots] - \IP[F \in S_R \mid \mathbf{Y}, \mathcal{O}(q_1), \dots] \Big| \le \mathrm{negl}(d), $
   }
   \]
   where the set $S_R = \{ s R : s \in S\}$. 
   That is, the secret fixed transformation $F$ is only recoverable up to a random rotation.
\end{lemma*}

\begin{proof}
    It suffices to show that $F$ is conditionally independent of both the watermarked samples $\mathbf{Y}$ and oracle responses $\mathcal{O}(q_1), \dots$. 

    By Lemma \ref{lemma:O_Ob_indist}, the outputs of the actual oracle $\mathcal{O}$ are statistically indistinguishable from an oracle $\mathcal{O}_3$ that evaluates solely on the dynamic component $\beta$. For any query $q = Y_i + \eta$, the perturbation's projection onto the effective transformation is:
    $$ \eta \Lambda^\top = (\eta F \quad \eta D_i) $$
    Because $\mathcal{O}_3$ evaluates strictly on $\beta$, the oracle responses depend exclusively on the projection $\eta D_i$. 

    Since the fixed transformation $A \sim \text{Haar}(d)$, it follows that $F$ consists of $d/2$ random $d$-dimensional components which are mutually orthonormal. Thus, for any rotation matrix $R \sim {Haar}(d/2)$ and any matrix $f \in \mathbb{R}^{d \times d/2}$ with orthonormal columns, the prior distribution of $F$ is rotationally invariant, i.e.
    $$ \IP[F \in S] = \IP[F \in S_R] = \IP[FR^T \in S]. $$
    Thus, both
    \[
    (\eta F \quad \eta D_i)  \text{\quad and \quad }(\eta F R \quad \eta D_i) 
    \]
    are statistically indistinguishable. 
    
To conclude the proof, it suffices to show that $F$ and $FR$ are conditionally indistinguishable given $D_i$. Let $M_i$ denote the orthonormal block-diagonal masking matrix such that $\Lambda_i = M_i A$. Because $A \sim \text{Haar}(d)$ and left-multiplication by a fixed orthonormal matrix preserves the Haar measure, the product $\Lambda_i$ is also $\text{Haar}(d)$-distributed. Therefore, $\Lambda_i^\top = [F, D_i]$ is a Haar-distributed orthonormal basis of $\mathbb{R}^d$.
Since for any $R \sim \text{Haar}(d/2)$, $FR$ spans the exact same orthogonal complement of $D_i$ (i.e. if $F$ and $D$ are orthogonal, then $FR$ and $D$ are also orthogonal), it follows that
    $$ \IP[F \in \mathcal{F} \mid D_i] = \IP[F \in \mathcal{F}R \mid D_i] $$
    for any set $\mathcal{F}$ orthogonal to $D_i$. That is, all choices of
    \[
    \{\mathcal{F} r : r \in Supp(Haar(d/2))\}
    \] 
    are equally likely, where $Supp(Haar(d/2))$ denotes the support of the 
    $Haar(d/2)$ distribution, i.e. all $d/2 \times d/2$ orthonormal matrices. Since the oracle responses depend only on $D_i$ (by Lemma \ref{lemma:O_Ob_indist}), and the passive samples $\mathbf{Y}$ are drawn from an isotropic Gaussian distribution $\mathcal{N}(0, I_d)$ regardless of $F$, the full observation history yields the same indistinguishability. The claim thus follows. 
\end{proof}
\begin{lemma*}[\emph{Lemma \ref{lemma:hidden_basis_isotropization}}]
Let $F\in\mathbb R^{d\times d/2}$ be an orthonormal basis, and
$R \sim Haar(d/2)$ be independent of a perturbation $\eta$. It follows that
\[
    \Delta(\eta F R, \|\eta F\|_2U_{d/2})=0.
\]
That is, an adversary that knows $F$ only up to a rotation $R$ sampled uniformly from $Orth(d/2)$ cannot outperform adding isotropic white-noise of magnitude $\|\eta F\|_2$.  
\end{lemma*}
\begin{proof}
Condition on any realization of $\eta \in \mathbb{R}^d$.
 Since $R\sim \mathrm{Haar}(d/2)$, the vector
$\eta F R$ is uniformly distributed on $\mathbb S^{d/2-1}$. Hence
$
 \Delta(   \eta FR , \|\eta F\|_2U_{d/2} ) = 0.
$ 
Since the equality holds for all realization of $\eta$, the total variation distance between
the two distributions is zero.
\end{proof}
\begin{lemma*}[\emph{Lemma \ref{lemma:sqrt2_baseline_comparison}}]
Suppose an adversary $\mathcal A$ is given the two $d/2$-dimensional
subspaces corresponding to $A_F$ and $A_D$, up to rotations sampled independently from $Haar(d/2)$. Then any perturbation
$\eta$ chosen by $\mathcal A$ with $\|\eta\|_2\leq\epsilon$ has negligible advantage over
full-space isotropic white-noise of magnitude $\sqrt2\epsilon$. That is,
\[
    \mathsf{Succ}(\mathcal A,\epsilon)
    \leq 
    \mathsf{Succ}(\mathcal A_w,\sqrt2\epsilon)
\]
\end{lemma*}

\begin{proof}
For any perturbation $\eta$ with $\|\eta\|_2\leq\epsilon$, write the hidden
fixed and dynamic bases as
\[
    P_F = F R,
    \qquad
    D = A_D R',
\]
where $R,R'\sim {Haar}(d/2)$ are sampled independently. Define
\[
    a := \|\eta P_F\|_2,
    \qquad
    b := \|\eta D\|_2.
\]
Since $P_F$ and $D$ are orthonormal bases for orthogonal $d/2$-dimensional
subspaces,
\[
    a^2+b^2
    =
    \|\eta P_F\|_2^2+\|\eta D\|_2^2
    \leq
    \|\eta\|_2^2
    \leq
    \epsilon^2.
\]

We first consider the fixed component. Since right multiplication by an
orthogonal matrix preserves norm,
\[
    \|\eta F\|_2
    =
    \|\eta F R\|_2
    =
    \|\eta P_F\|_2
    =
    a.
\]
Moreover, since $R$ is sampled uniformly from $O(d/2)$, 
\[
    \Delta\!\left(
        \eta P_F,
        aU_{d/2}
    \right)=0.
\]
where $U_{d/2}$ is uniform on $\mathbb S^{d/2-1}$.

Similarly, since $D=A_D R'$ and $R'\sim {Haar}(d/2)$ is the
one-time rotation for the dynamic component,
\[
    \|\eta A_D\|_2
    =
    \|\eta A_D R'\|_2
    =
    \|\eta D\|_2
    =
    b,
\]
and hence
\[
    \Delta\!\left(
        \eta D,
        bV_{d/2}
    \right)=0.
\]
where $V_{d/2}$ is sampled uniformly from $\mathbb S^{d/2-1}$.
Therefore, the effective perturbation in detector coordinates satisfies
\[
    \Delta(\eta\Lambda_i^\top,
    (aU_{d/2}, bV_{d/2}))=0.
\]
That is, the adversary is limited to isotropic white-noise of magnitude $a$ and $b$ within the two subspaces. 

Since $a,b\leq \epsilon$, the perturbation $(aU_{d/2},bV_{d/2})$ is dominated
by 
\[
    (\epsilon U_{d/2},\epsilon V_{d/2}).
\]
Under the well-behaved threshold detector model, increasing the isotropic
perturbation magnitude in either half can only increase the probability that
the corresponding sign-flip threshold is exceeded. Hence,
\[
    \Pr[\textnormal{reject under }(aU_{d/2},bV_{d/2})]
    \leq
    \Pr[\textnormal{reject under }(\epsilon U_{d/2},\epsilon V_{d/2})].
\]

Furthermore,  $(\epsilon U_{d/2},\epsilon V_{d/2})$ has total squared
norm
\[
    \epsilon^2+\epsilon^2=2\epsilon^2,
\]
and hence total magnitude $\sqrt2\epsilon$. 
Let $\mathcal A_w$ output isotropic white-noise $W$ over the entire $d$-dimensions
 with $\|W\|_2=\sqrt2\epsilon$. Let $s$ be a random subset of $d/2$ elements from $W$. Standard concentration of measure results gives
\[
    \Pr\left[
        \left|\|s\|_2^2-\epsilon^2\right|
        >
        \epsilon^2 d^{-1/3}
    \right]
    \leq
    \exp(-\Omega(d^{1/3})).
\]
By a union bound, with probability at least
\[
    1-2\exp(-\Omega(d^{1/3})),
\]
the split isotropic perturbation
$(\epsilon U_{d/2},\epsilon V_{d/2})$ and the full-space isotropic perturbation
of magnitude $\sqrt2\epsilon$ induce the same detector behavior up to a
negligible concentration error. It follows that
\[
    \mathsf{Succ}(\mathcal A,\epsilon)
    \leq
    \mathsf{Succ}(\mathcal A_w,\sqrt2\epsilon)+\mathrm{negl}(d).
\]
\end{proof}

\begin{lemma*} [\emph{Lemma \ref{lemma:most_powerful}}]
Continuing from the setting in Lemma \ref{lemma:basis_ind}. If there exists a adversary $\mathcal{A}$ operating under the REM2 threat model such that  $\mathsf{Adv}(\mathcal{A}, \epsilon)$ is non-negligible, then there exists a measurable subset $S \subseteq \mathcal{F}$, a rotation $R \in Haar(d/2)$, and polynomial $p(\cdot)>0$ such that:
\[
   \Big| \IP[F \in S \mid \mathbf{Y}, \mathcal{O}(q_1), \dots] - \IP[F \in S_R \mid \mathbf{Y}, \mathcal{O}(q_1), \dots] \Big| \geq \frac{1}{p(d)},
\]
where $S_R = \{ M R : M \in S\}$. That is, $\mathcal{A}$ can distinguish the specific orientation of the fixed basis $F$ from a rotated counterpart with non-negligible probability.
\end{lemma*}
\begin{proof}
We prove the contrapositive. Suppose that $\mathcal A$ obtains negligible orientation information about $F$. That is, for all polynomial
$p(\cdot)>0$ and adversary's view $Z_{\mathcal A}$, 
\[
\begin{aligned}
    \Pr_{Z_{\mathcal A}} \Bigg[
        &\exists\,\mathcal S\subseteq\mathcal F,\,
        \exists\,R\in Orth(d/2)
        \text{ such that } \\
        &\left|
            \Pr[F\in\mathcal S\mid Z_{\mathcal A}]
            -
            \Pr[F\in\mathcal S_R\mid Z_{\mathcal A}]
        \right|
        \geq \frac{1}{p(d)}
    \Bigg]
\end{aligned}
\]
is negligible, where $Orth(m)$ denote the orthonormal group in $m$ dimensions

Let
\[
    \eta = \mathcal A(Z_{\mathcal A}), \qquad \|\eta\|_2\leq\epsilon,
\]
and $P_F$ and $P_D$ denote the orthogonal projectors onto the fixed and dynamic
half-subspaces, respectively.  Define
\[
    a := \|\eta P_F\|_2,
    \qquad
    b := \|\eta P_D\|_2.
\]
Since $P_F$ and $P_D$ project onto orthogonal half-subspaces,
\[
    a^2+b^2\leq \|\eta\|_2^2\leq \epsilon^2.
\]

Because the conditional distribution of $F$ remains invariant under right
multiplication by rotations in $Orth(d/2)$, the adversary does not know the
orientation of the detector basis inside the fixed subspace. Therefore. the total variation distance
\[
    \Delta(\eta F^\top, aU_{d/2}) \leq negl(d),
\]
where $U_m$ is uniform on the unit
sphere $\mathbb S^{d/2-1}$.
Similarly, for the challenge sample, the dynamic detector basis is masked by a
fresh one-time Haar rotation. Hence
\[
    \Delta(\eta D^\top,  bU_{d/2}) \leq negl(d),
\]
where $D$ denotes the dynamic basis for the challenge sample.

Thus, conditioned on this realization of $Z_{\mathcal A}$, the adversary's
perturbation is equivalent to an isotropic perturbation in the two hidden
half-spaces, with some energy split $(a,b)$ satisfying
$a^2+b^2\leq\epsilon^2$. The adversary cannot target individual detector
coordinates; it can only choose how much energy to allocate to each half.

By the subspace-isotropization and $\sqrt2\epsilon$ baseline comparison lemma,
every such rotation-hidden perturbation is dominated by the full-space
isotropic white-noise adversary of magnitude $\sqrt2\epsilon$. Hence
\[
    \mathsf{Succ}(\mathcal A,\epsilon\mid Z_{\mathcal A})
    \leq
    \mathsf{Succ}(\mathcal A_w,\sqrt2\epsilon)+negl(d)
\]
where $\mathsf{Succ}(\mathcal A,\epsilon\mid Z_{\mathcal A})$ denotes the success probability of $ \mathcal{A}$ conditional on view $Z_\mathcal{A}$

Averaging over $Z_{\mathcal A}$, and using that bad realizations occur with
negligible probability, gives
\[
    \mathsf{Succ}(\mathcal A,\epsilon)
    \leq
    \mathsf{Succ}(\mathcal A_w,\sqrt2\epsilon)+negl(d).
\]
Therefore
\[
    \mathsf{Adv}(\mathcal A,\epsilon)\leq negl(d).
\]
Taking the contrapositive proves the lemma.
\end{proof}

\begin{theorem*} [\emph{Theorem \ref{thm:computational_rem2}}]
If the underlying PRC encoder, PRF, and orthonormal matrix generators are secure pseudorandom primitives, and the condition in (\ref{eqn:threshold_tau}) is met, then under the REM2 threat model, for all PPT adversaries $\mathcal{A}$, 
\[
{Adv}(\mathcal{A}, \epsilon) \leq {negl}(\lambda),
\]
where $\lambda$ is the security parameter.
\end{theorem*}
\begin{proof}
We proceed via a sequence of hybrid experiments, starting from the real-world instantiation and replacing the pseudorandom matrices with their idealized counterparts. Let $\mathcal{A}$ be a PPT adversary bounded by a perturbation budget $\epsilon$, and assume $d \geq \lambda$. We denote the advantage of $\mathcal{A}$ in hybrid $H_x$ as $\text{Adv}_{H_x}(\mathcal{A}, \epsilon)$.

\begin{itemize}
    \item {Hybrid 0 ($H_0$):} This corresponds to the real-world watermarking system defined in Algorithms 1, 2, and 3. By definition, ${Adv}_{H_0}(\mathcal{A}, \epsilon) = Adv(\mathcal{A}, \epsilon)$.

    \item {Hybrid 1 ($H_1$):} We modify $H_0$ by replacing the long-term secret matrix $A \leftarrow \text{SRandOrthonormal}(d, k_2)$ with a truly random $A^* \sim \text{Haar}(d)$, applied consistently in both sample generation and oracle detection. This transition is valid because $k_2$ is never revealed to $\mathcal{A}$. Any PPT distinguisher between $H_0$ and $H_1$ directly breaks the pseudorandomness of \texttt{SRandOrthonormal}:
    \[
    \begin{split}
    |\text{Adv}_{H_0}(\mathcal{A}, \epsilon) - \text{Adv}_{H_1}(\mathcal{A}, \epsilon)| \leq negl_1(\lambda)
    \end{split}
    \]

    \item {Hybrid 2 ($H_2$):} We modify $H_1$ by replacing each per-sample matrix $B_i = \text{SRandOrthonormal}(d/2, r)$ with a perfectly random matrix $B_i^* \sim \text{Haar}(d/2)$, consistently in both generation and detection. This transition is valid because the nonce $r$ is computationally hidden from $\mathcal{A}$; it appears in the adversary's view only through $C_1 = \text{PRC}(k_1, r)$ and $C_2 = f_{k_3}(r)$. Recovering $r$ requires breaking the computational security of the PRC or the PRF. Since $\mathcal{A}$ cannot recover $r$, the output of \texttt{SRandOrthonormal} is computationally indistinguishable from a Haar-random matrix:
    \[
    \begin{split}
    &|\text{Adv}_{H_1}(\mathcal{A}, \epsilon) - \text{Adv}_{H_2}(\mathcal{A}, \epsilon)| \\
    &\leq negl_2(\lambda) + negl_3(\lambda) + negl_4(\lambda)
    \end{split}
    \]
    where $negl_2$ is the cost of the \texttt{SRandOrthonormal} reduction, and $negl_3, negl_4$ bound the probability of $r$ being recovered from $C_1$ or $C_2$.
\end{itemize}

Note that $H_2$  matches the conditions of the information-theoretic result in Theorem \ref{thm:rem2-secure}. Therefore,
\[
\text{Adv}_{H_2}(\mathcal{A}, \epsilon) \leq negl_5(d) \leq negl_5(\lambda),
\]
where the latter is due to $\lambda \leq d$. 
By the triangle inequality,
\[
\begin{split}
&|\text{Adv}(\mathcal{A}, \epsilon) - \text{Adv}_{H_2}(\mathcal{A}, \epsilon)|\\
&=|\text{Adv}_{H_0}(\mathcal{A}, \epsilon) - \text{Adv}_{H_2}(\mathcal{A}, \epsilon)|\\
&\leq |\text{Adv}_{H_0}(\mathcal{A}, \epsilon) - \text{Adv}_{H_1}(\mathcal{A}, \epsilon)|  + |\text{Adv}_{H_1}(\mathcal{A}, \epsilon) - \text{Adv}_{H_2}(\mathcal{A}, \epsilon)| \\
&\leq  negl_1(\lambda)+negl_2(\lambda) + negl_3(\lambda) + negl_4(\lambda)
\end{split}
\]
which is negligible in $\lambda$. It follows that,
\[
\text{Adv}(\mathcal{A}, \epsilon)  \leq |\text{Adv}(\mathcal{A}, \epsilon) - \text{Adv}_{H_2}(\mathcal{A}, \epsilon)| + \text{Adv}_{H_2}(\mathcal{A}, \epsilon)
\]
which is negligible in $\lambda$. 
\end{proof}

\section{Additional Empirical Studies}

\subsection{Validating Threshold Gap Condition} \label{appendix_threshold_gap}
We empirically verify that the threshold gap condition in (\ref{eqn:threshold_tau}) is fulfilled. To empirically estimate $\tau_1$, we proceed as follows:
\begin{enumerate}
    \item Randomly sample a PRC latent codeword $C$ of length $d/2$, embedding a uniformly sampled 100-bit message.
    \item Generate 1,000 noisy copies $C_{1, p}, \dots, C_{1000, p}$, where each $C_{i, p}$ is a copy of $C$ with each entry's sign randomly flipped with probability $p$.
    \item Find the maximum value of $p'$ such that all $C_{1, p'}, \dots, C_{1000, p'}$ successfully decode. We set this $p'$ as our estimate for $\tau_1$.
\end{enumerate}
In our experiments, the value of $\tau_1$ is estimated to be 0.31 for $d=16,384$ (SD2.1), and 0.37 for $d=65,536$ (SDXL).

\subsubsection*{Single-Bit Watermarking} 
By setting $t < \frac{d}{2} \left(\tau_1 - \sqrt{\frac{\ln d \ln \ln d}{d}} \right)$, the condition in (\ref{eqn:threshold_tau}) is satisfied. 

\subsubsection*{Multi-Bit Watermarking}
For multi-bit watermarking, we use the parameters suggested in the \texttt{pyldpc} package. Similar to the PRC case, we empirically estimate the threshold $\tau_2$ using the same sign-flipping noise injection method described above. In our experiments, the value of $\tau_2$ is estimated to be 0.26 for $d=16,384$ (SD2.1), and 0.34 for $d=65,536$ (SDXL). 
This meets the threshold gap condition in (\ref{eqn:threshold_tau}).

\subsection{Performance Overhead} \label{appendix_performance_overhead}

We evaluate the memory and computational overhead of our proposed defense on both SD2.1 and SDXL. 
\begin{table}
\centering
\resizebox{\columnwidth}{!}{ 
\begin{tabular}{l c c c c}
\toprule
\textbf{Metric} & \textbf{SD2.1} & \textbf{SDXL} & \textbf{Empirical} & \textbf{Theoretical} \\
 & & & \textbf{Scaling} & \textbf{Scaling} \\
\midrule
Total Memory & $\sim$1.3 GB & $\sim$20.0 GB & 15.0$\times$ & 16.0$\times$ \\
Sample Time (s)  & 0.1    &  1.2   & 12.0$\times$ & 16.0$\times$ \\
Decode Time  (s) & 0.2    &   1.9   & 9.5$\times$ & 16.0$\times$ \\
\bottomrule
\end{tabular}
}
\caption{Average memory and time overhead for sampling and decoding phases (averaged over 100 trials), rounded to the nearest 1 decimal place. The empirical scaling closely tracks the theoretical $16\times$ increase expected from the $O(d^2)$ matrix operations. As a comparison, the average unwatermarked per-image generation times are approximately 2.5s and 7.3s for SD2.1 and SDXL, respectively}
\label{tab:performance_overhead}
\end{table}
\subsubsection*{Memory Overhead} 
For efficiency, the fixed $d \times d$ transformation matrix $A$ and the one-time $(d/2) \times (d/2)$ matrix $B$ are stored in memory. This combined storage results in a total memory footprint of approximately 1.3 GB for SD2.1 and 20.0 GB for SDXL, representing an approximate $15.4\times$ increase in memory requirements.

\subsubsection*{Time Overhead}
We measured the average execution time for the sampling and decoding phases across 100 independent trials. In both Algorithm \ref{alg:sample} (Line \ref{alg_sample_line_mult}) and Algorithm \ref{alg:decode_detect} (Lines \ref{alg_decode_line_mult1} and \ref{alg_decode_line_mult2}), the primary computational bottlenecks are dense matrix multiplications, which require $O(d^2)$ operations. 

Because the latent dimension $d$ in SDXL is 4 times larger than in SD2.1, the $O(d^2)$ complexity yields a theoretical $16\times$ increase in computational workload. As shown in Table \ref{tab:performance_overhead}, our empirical latency measurements closely align with this expected scaling factor, demonstrating $12.0\times$ and $9.5\times$ increases for sampling and decoding, respectively. The deviation from the theoretical $16 \times$ likely stems from other fixed overhead incurred during the sampling and decoding steps.

\subsection{Samples of Watermarked Images} \label{appendix_im_sample_wm}
Samples of watermarked images generated are in Figures \ref{fig:sd21_appendix_full} and \ref{fig:img/sdxl_appendix_full} for Stable Diffusion 2.1 and Stable Diffusion XL, respectively. 

\begin{figure*}[t]
    \centering


    \begin{subfigure}{0.19\linewidth}
        \includegraphics[width=\linewidth]{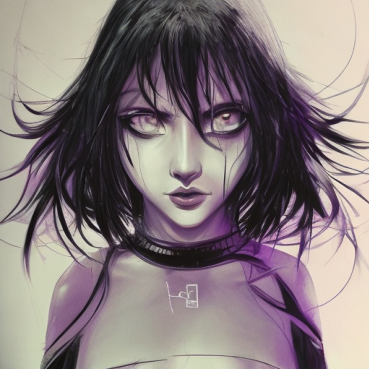}
    \end{subfigure}\hfill
    \begin{subfigure}{0.19\linewidth}
        \includegraphics[width=\linewidth]{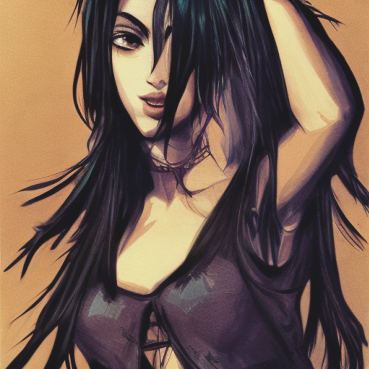}
    \end{subfigure}\hfill
    \begin{subfigure}{0.19\linewidth}
        \includegraphics[width=\linewidth]{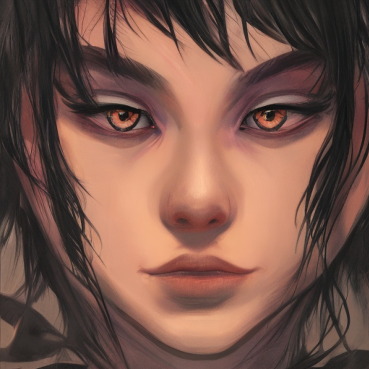}
    \end{subfigure}\hfill
    \begin{subfigure}{0.19\linewidth}
        \includegraphics[width=\linewidth]{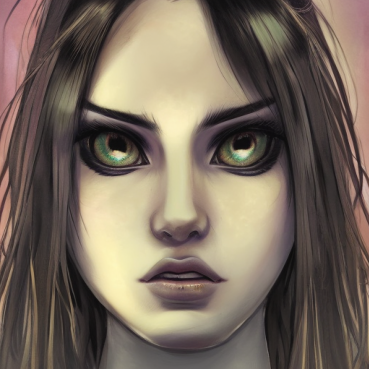}
    \end{subfigure}\hfill
    \begin{subfigure}{0.19\linewidth}
        \includegraphics[width=\linewidth]{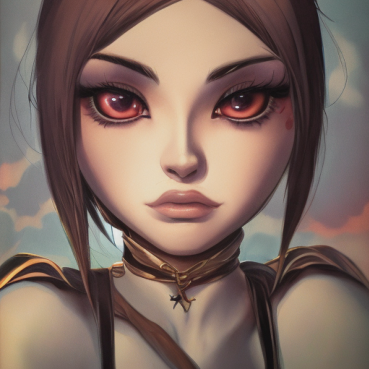}
    \end{subfigure}
    \caption*{a comic potrait of a female necromamcer with big and cute eyes, fine - face, realistic shaded perfect face, fine details. night setting. very anime style. realistic shaded lighting poster by ilya kuvshinov katsuhiro, magali villeneuve, artgerm, jeremy lipkin and michael garmash, rob rey and kentaro miura style, trending on art station}
    \vspace{1em}

    \begin{subfigure}{0.19\linewidth}
        \includegraphics[width=\linewidth]{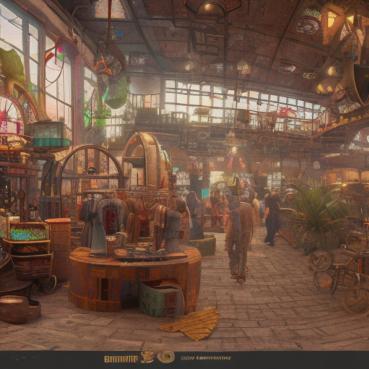}
    \end{subfigure}\hfill
    \begin{subfigure}{0.19\linewidth}
        \includegraphics[width=\linewidth]{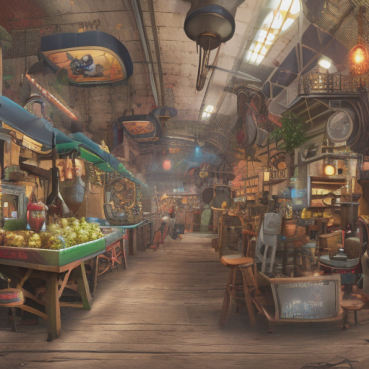}
    \end{subfigure}\hfill
    \begin{subfigure}{0.19\linewidth}
        \includegraphics[width=\linewidth]{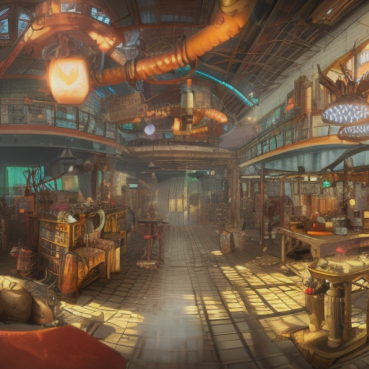}
    \end{subfigure}\hfill
    \begin{subfigure}{0.19\linewidth}
        \includegraphics[width=\linewidth]{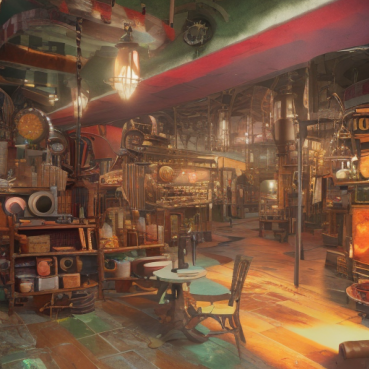}
    \end{subfigure}\hfill
    \begin{subfigure}{0.19\linewidth}
        \includegraphics[width=\linewidth]{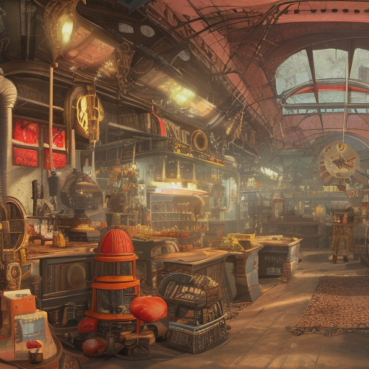}
    \end{subfigure}
    \caption*{steampunk market interior, colorful, 3 d scene, greg rutkowski, zabrocki, karlkka, jayison devadas, trending on artstation, 8 k, ultra wide angle, zenith view, pincushion lens effect}
    \vspace{1em}

    \begin{subfigure}{0.19\linewidth}
        \includegraphics[width=\linewidth]{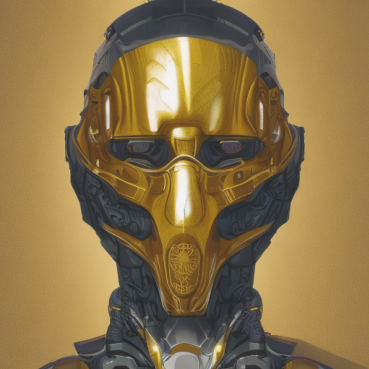}
    \end{subfigure}\hfill
    \begin{subfigure}{0.19\linewidth}
        \includegraphics[width=\linewidth]{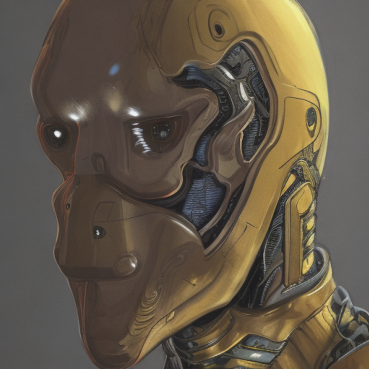}
    \end{subfigure}\hfill
    \begin{subfigure}{0.19\linewidth}
        \includegraphics[width=\linewidth]{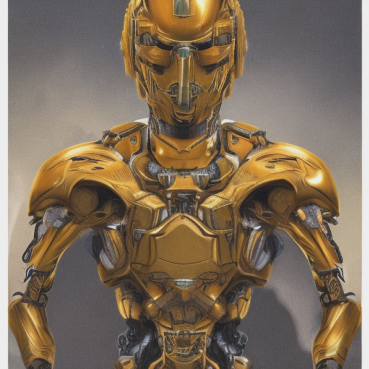}
    \end{subfigure}\hfill
    \begin{subfigure}{0.19\linewidth}
        \includegraphics[width=\linewidth]{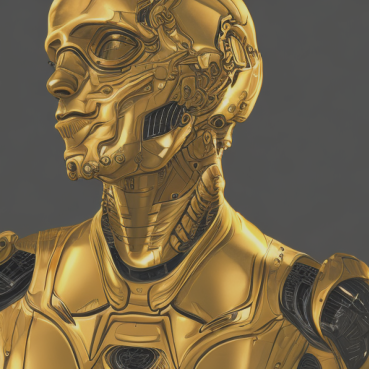}
    \end{subfigure}\hfill
    \begin{subfigure}{0.19\linewidth}
        \includegraphics[width=\linewidth]{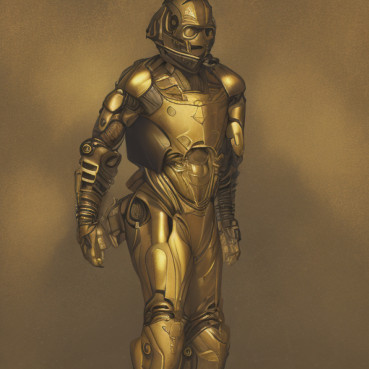}
    \end{subfigure}
    \caption*{A portrait of a cyborg in a golden suit, D\&D sci-fi, artstation, concept art, highly detailed illustration.}
    \vspace{1em}

    \begin{subfigure}{0.19\linewidth}
        \includegraphics[width=\linewidth]{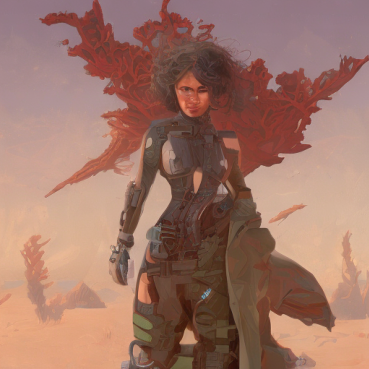}
    \end{subfigure}\hfill
    \begin{subfigure}{0.19\linewidth}
        \includegraphics[width=\linewidth]{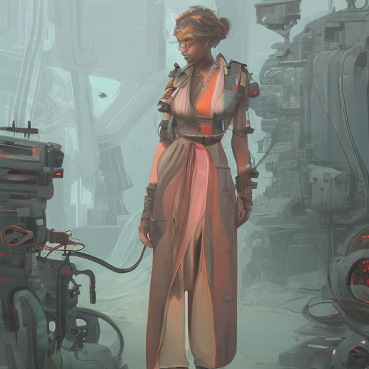}
    \end{subfigure}\hfill
    \begin{subfigure}{0.19\linewidth}
        \includegraphics[width=\linewidth]{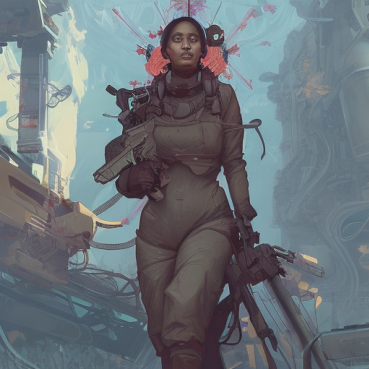}
    \end{subfigure}\hfill
    \begin{subfigure}{0.19\linewidth}
        \includegraphics[width=\linewidth]{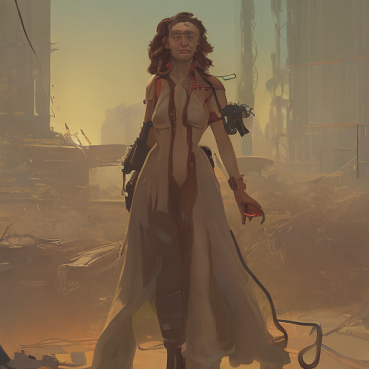}
    \end{subfigure}\hfill
    \begin{subfigure}{0.19\linewidth}
        \includegraphics[width=\linewidth]{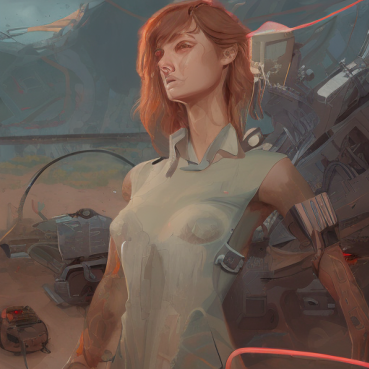}
    \end{subfigure}
    \caption*{A full portrait of a beautiful post apocalyptic offworld nanotechnician, intricate, elegant, highly detailed, digital painting, artstation, concept art, smooth, sharp focus, illustration, art by Krenz Cushart and Artem Demura and alphonse mucha}

    \vspace{0.5em}
    \caption{Watermarked images generated using Stable Diffusion 2.1 using prompts from the Stable Diffusion Prompts dataset}
    \label{fig:sd21_appendix_full}
\end{figure*}

\begin{figure*}[t]
    \centering


    \begin{subfigure}{0.19\linewidth}
        \includegraphics[width=\linewidth]{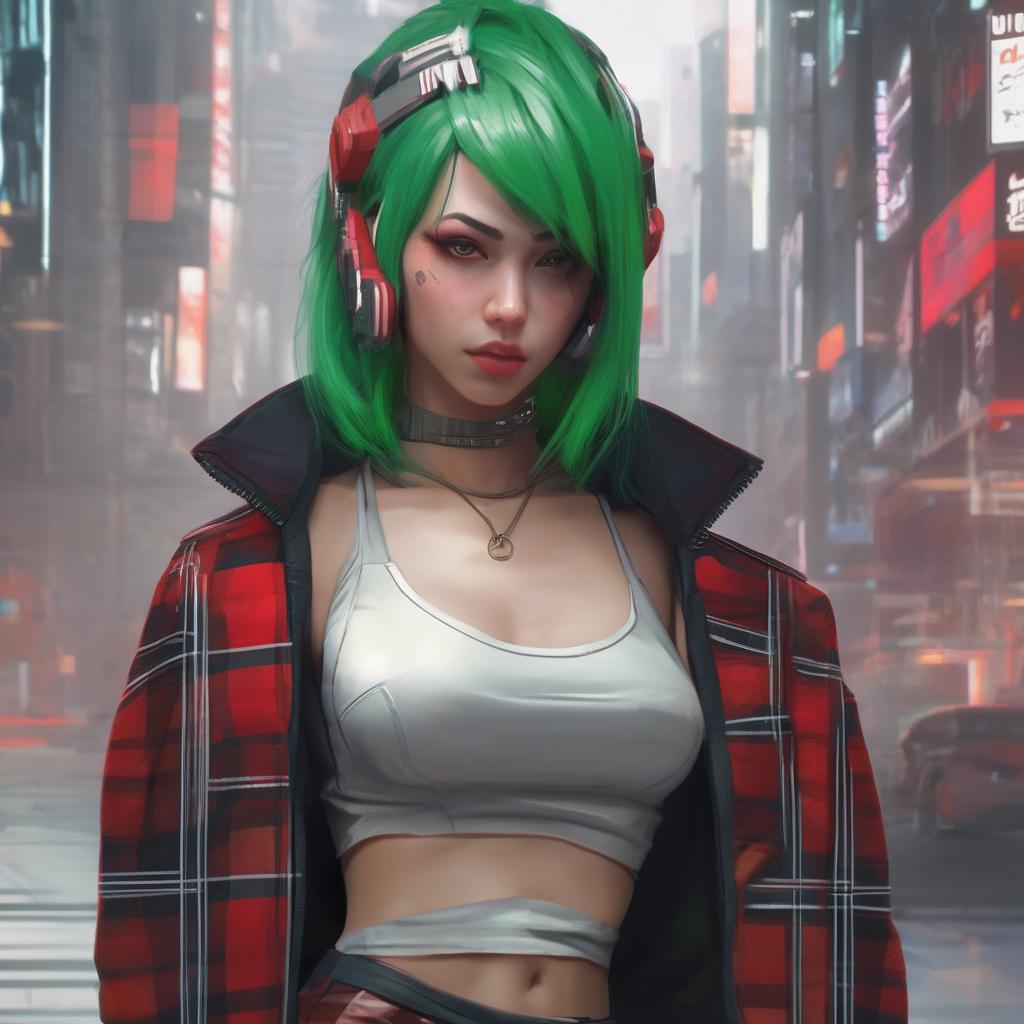}
    \end{subfigure}\hfill
    \begin{subfigure}{0.19\linewidth}
        \includegraphics[width=\linewidth]{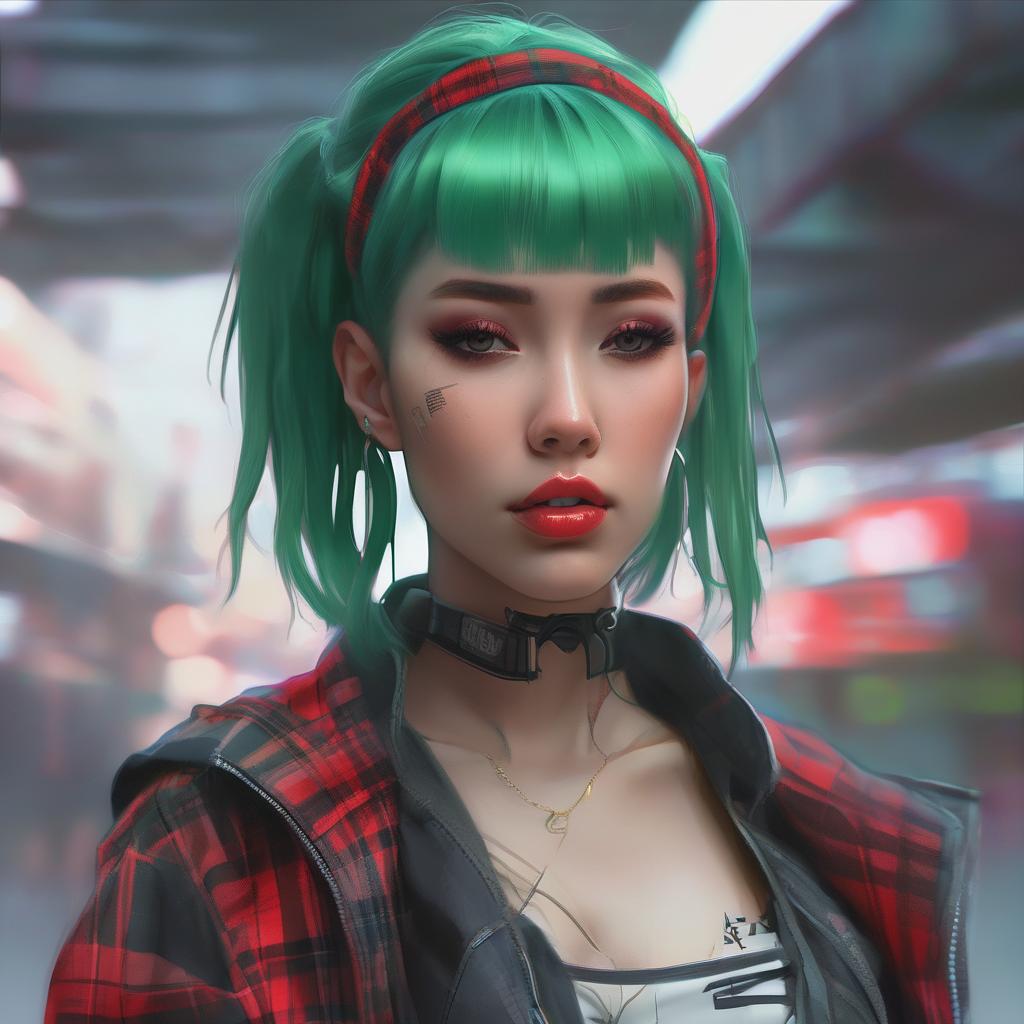}
    \end{subfigure}\hfill
    \begin{subfigure}{0.19\linewidth}
        \includegraphics[width=\linewidth]{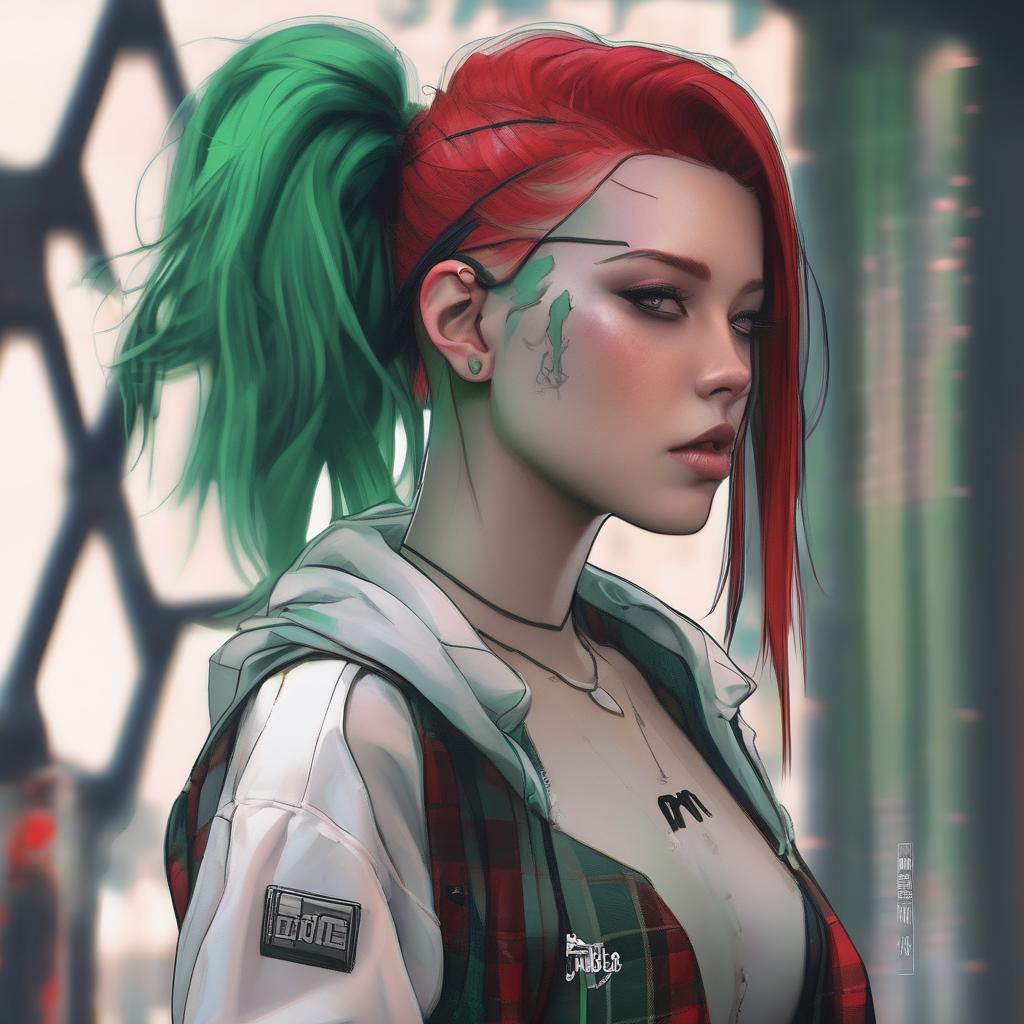}
    \end{subfigure}\hfill
    \begin{subfigure}{0.19\linewidth}
        \includegraphics[width=\linewidth]{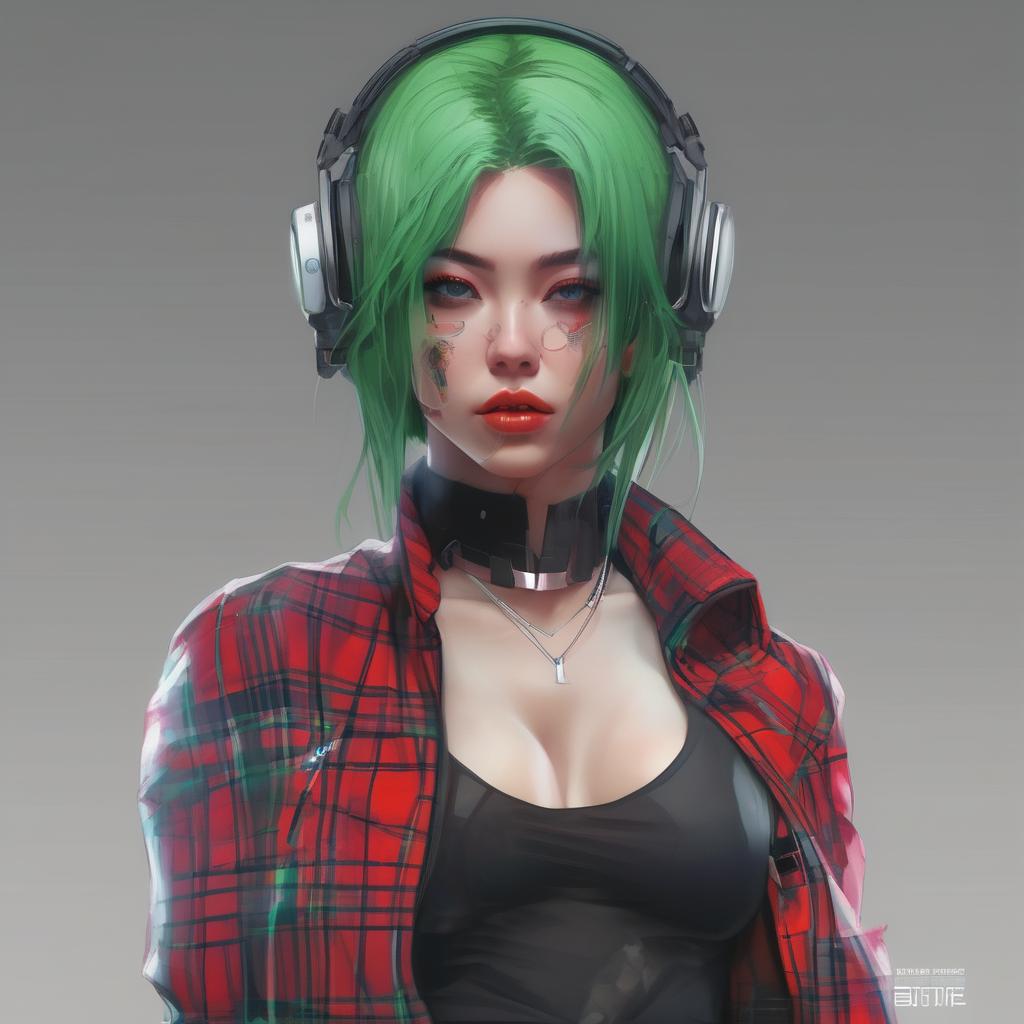}
    \end{subfigure}\hfill
    \begin{subfigure}{0.19\linewidth}
        \includegraphics[width=\linewidth]{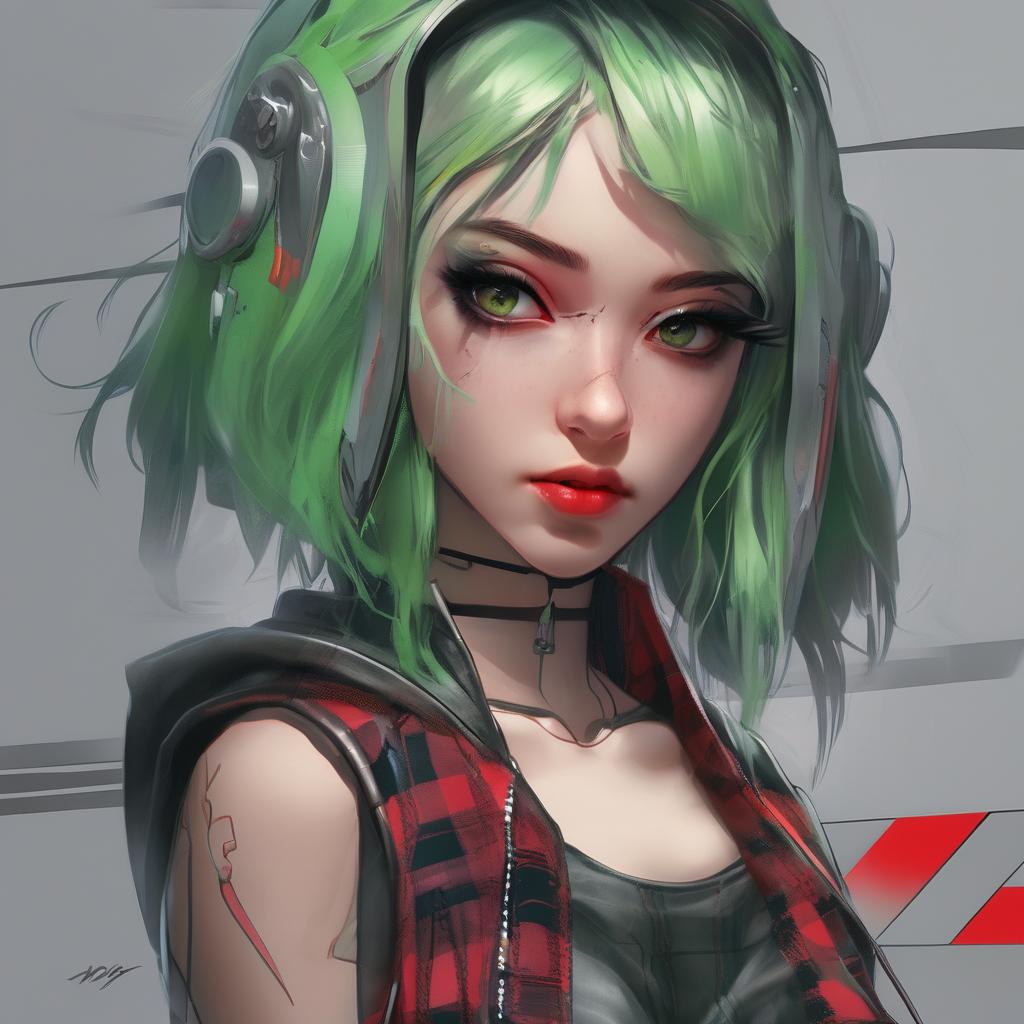}
    \end{subfigure}
    \caption*{cyberpunk woman with green hair wearing futuristic clothes and a red plaid miniskirt, beautiful face, by wlop, by artgerm, by yoshitaka amano, digital art, matte art, octane render, lineart, pop art}
    \vspace{1em}

    \begin{subfigure}{0.19\linewidth}
        \includegraphics[width=\linewidth]{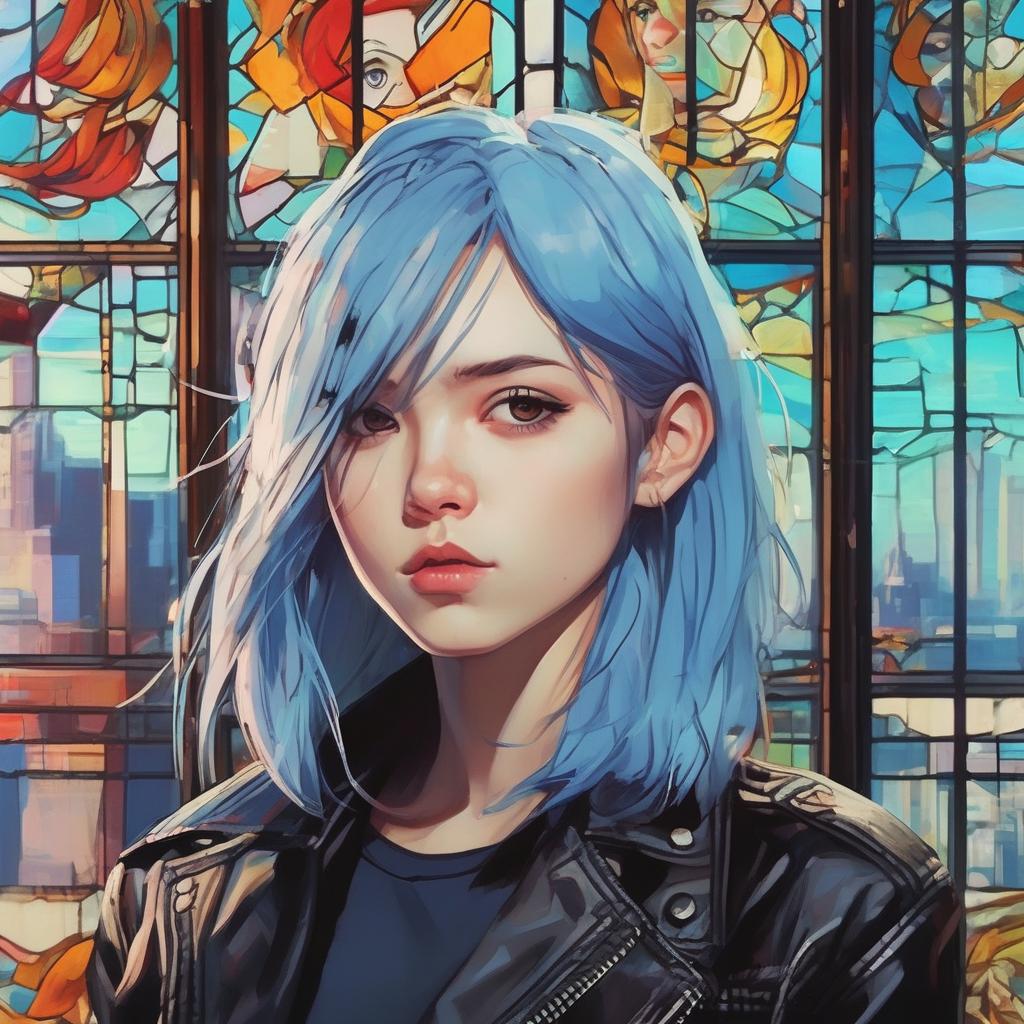}
    \end{subfigure}\hfill
    \begin{subfigure}{0.19\linewidth}
        \includegraphics[width=\linewidth]{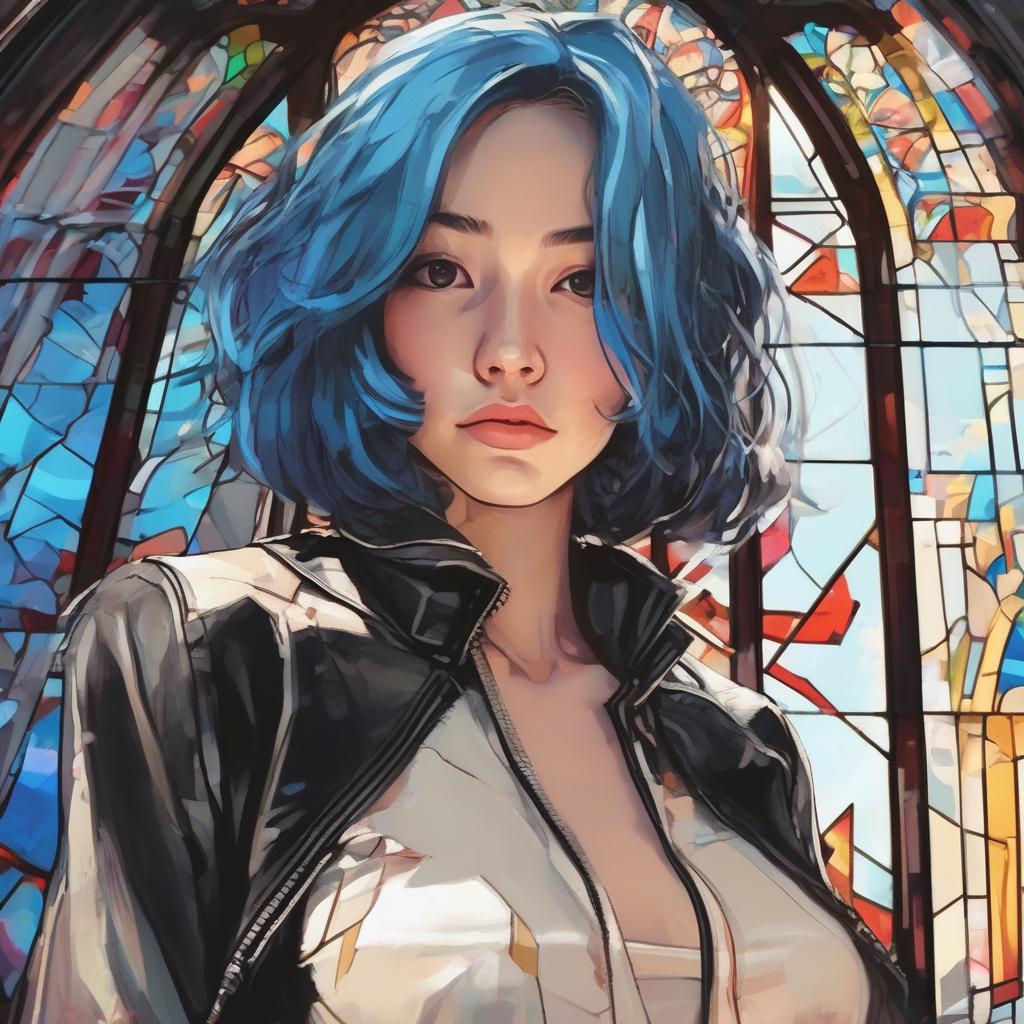}
    \end{subfigure}\hfill
    \begin{subfigure}{0.19\linewidth}
        \includegraphics[width=\linewidth]{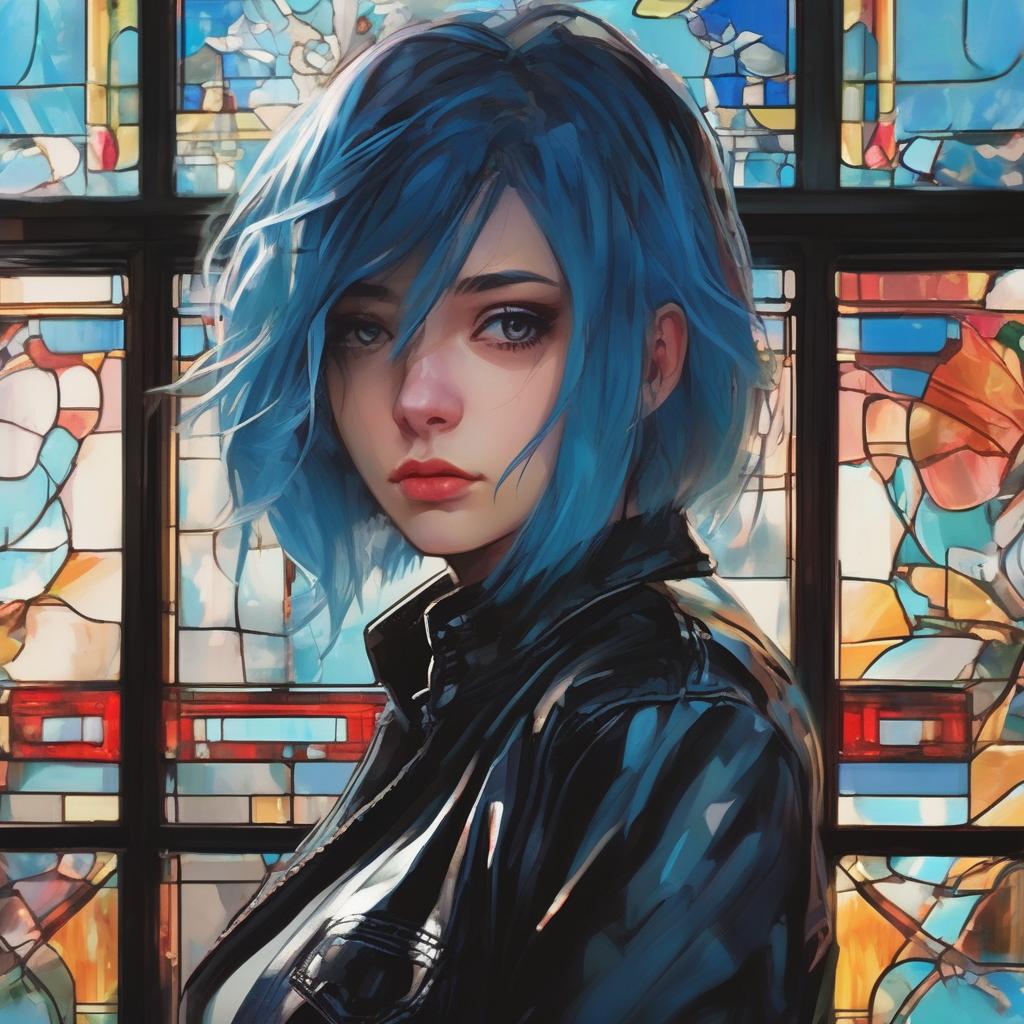}
    \end{subfigure}\hfill
    \begin{subfigure}{0.19\linewidth}
        \includegraphics[width=\linewidth]{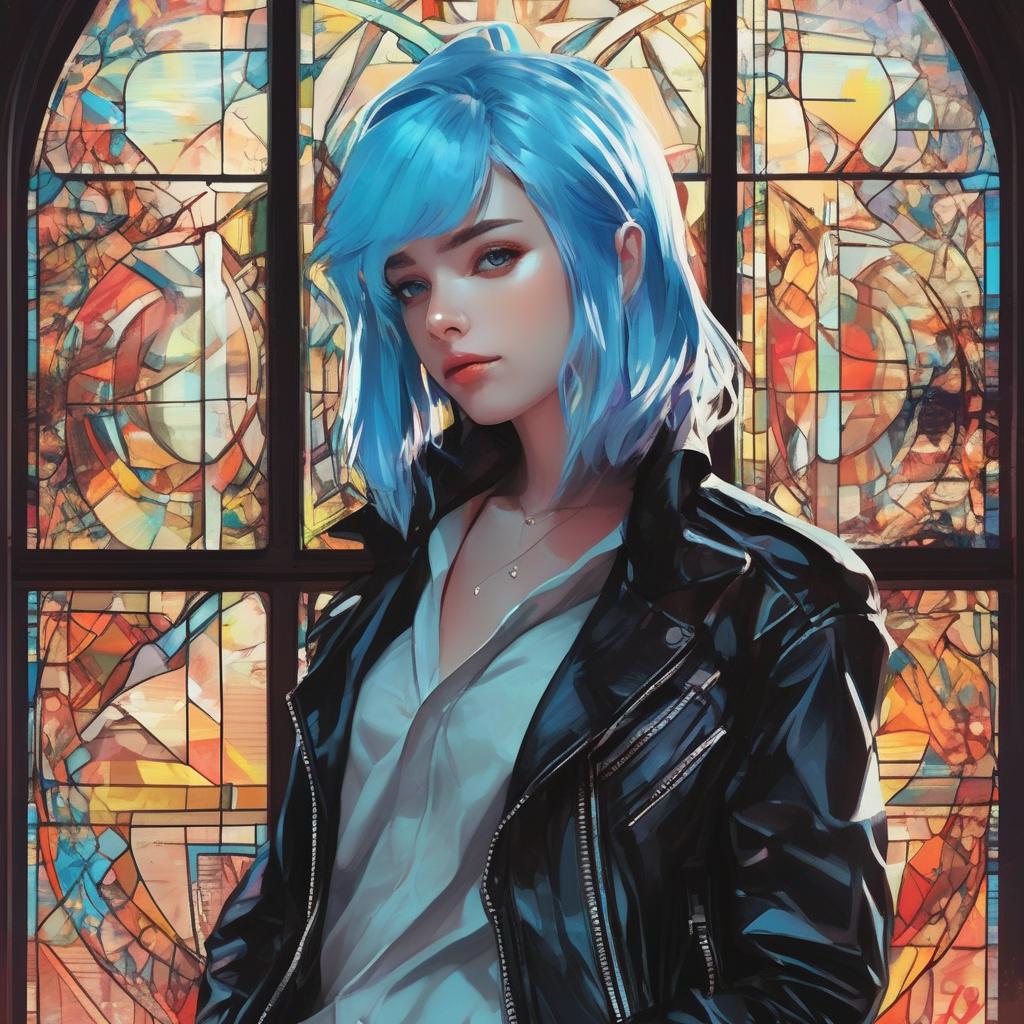}
    \end{subfigure}\hfill
    \begin{subfigure}{0.19\linewidth}
        \includegraphics[width=\linewidth]{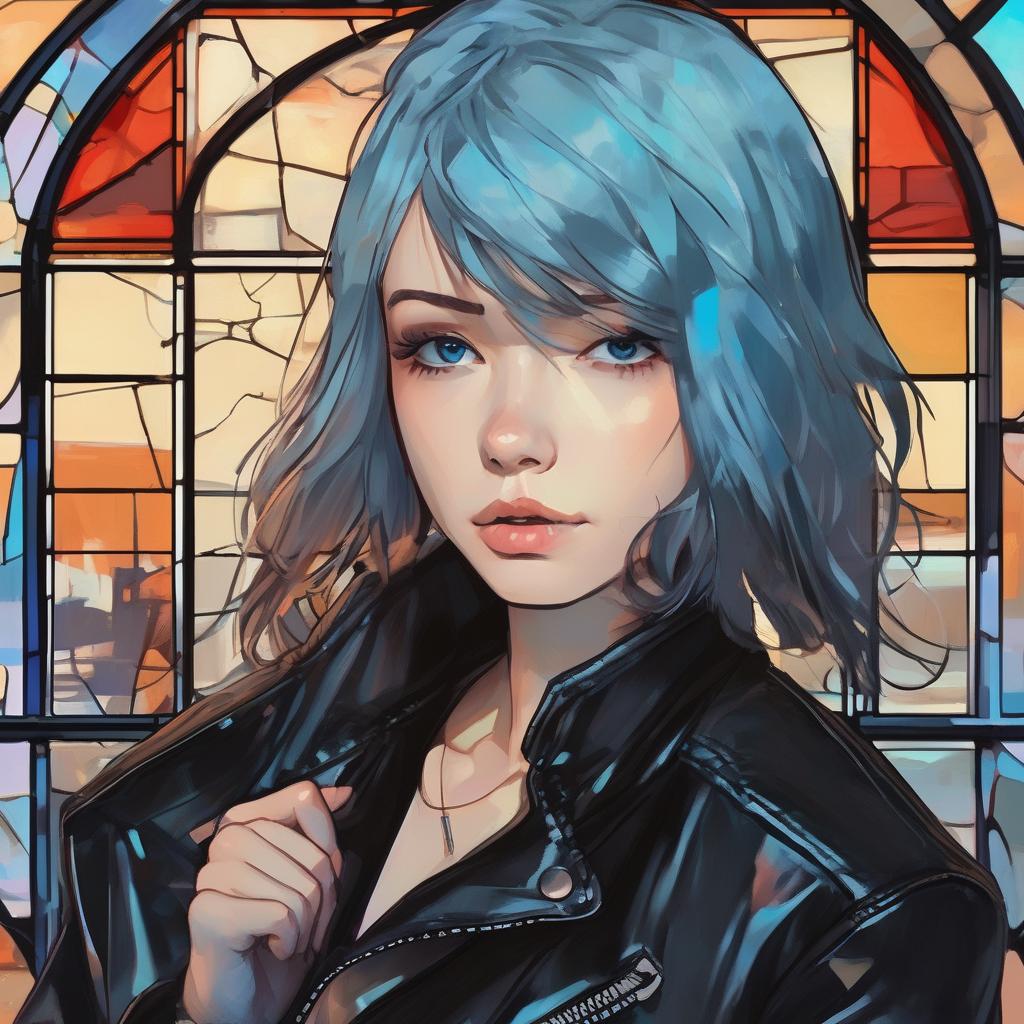}
    \end{subfigure}
    \caption*{ilya kuvshinov with long hair, sky blue hair, hazel eyes, boyish face, black leather jacket, professional digital painting, concept art, award - winning photography, cinematic, stained glass window, wlop, art by andy warhol, pixiv art, yoshitaka amano}
    \vspace{1em}

    \begin{subfigure}{0.19\linewidth}
        \includegraphics[width=\linewidth]{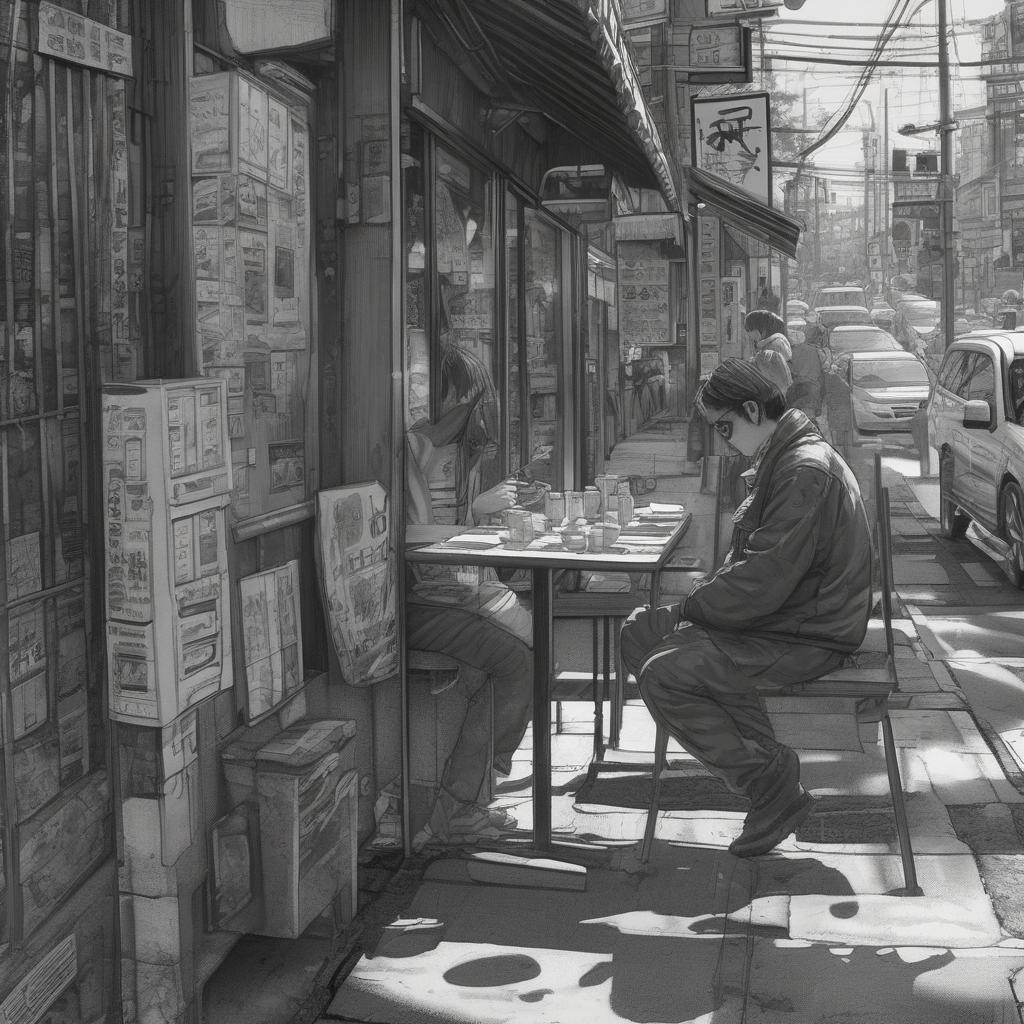}
    \end{subfigure}\hfill
    \begin{subfigure}{0.19\linewidth}
        \includegraphics[width=\linewidth]{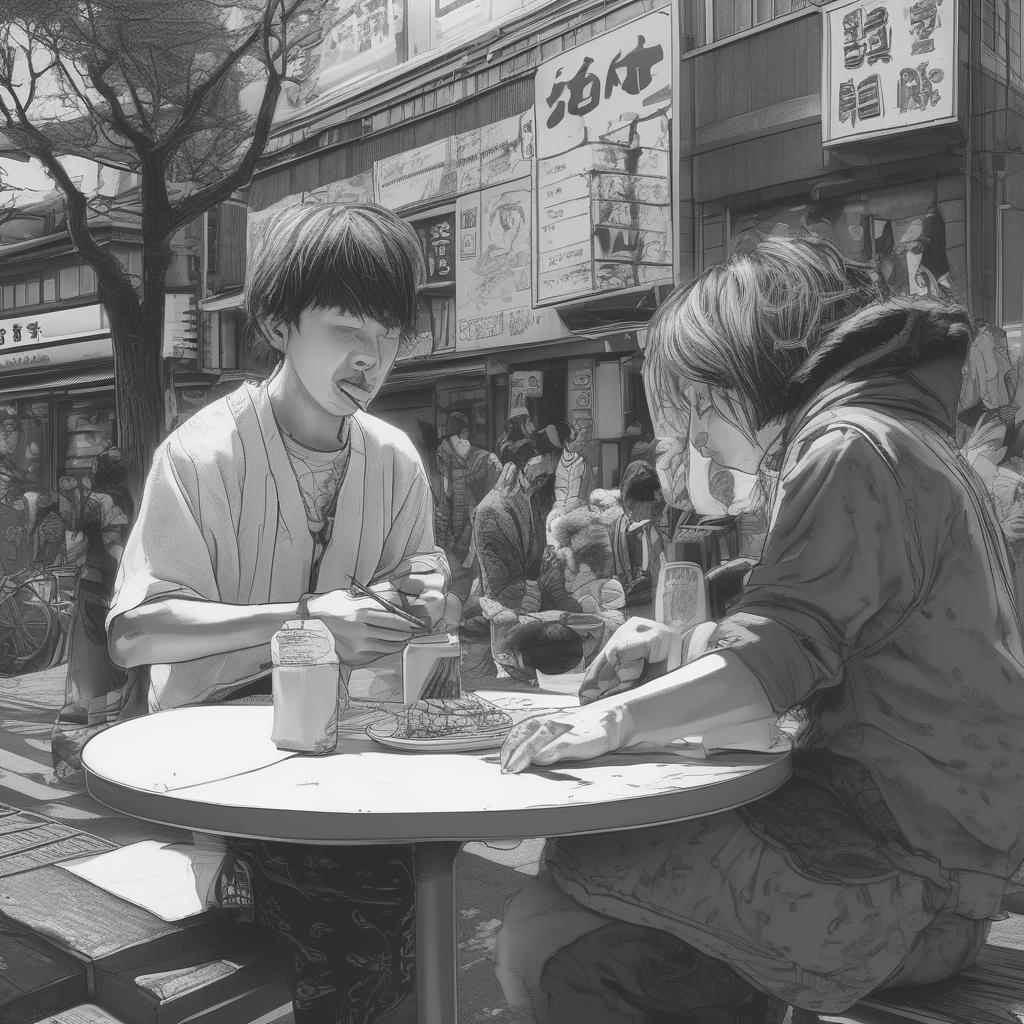}
    \end{subfigure}\hfill
    \begin{subfigure}{0.19\linewidth}
        \includegraphics[width=\linewidth]{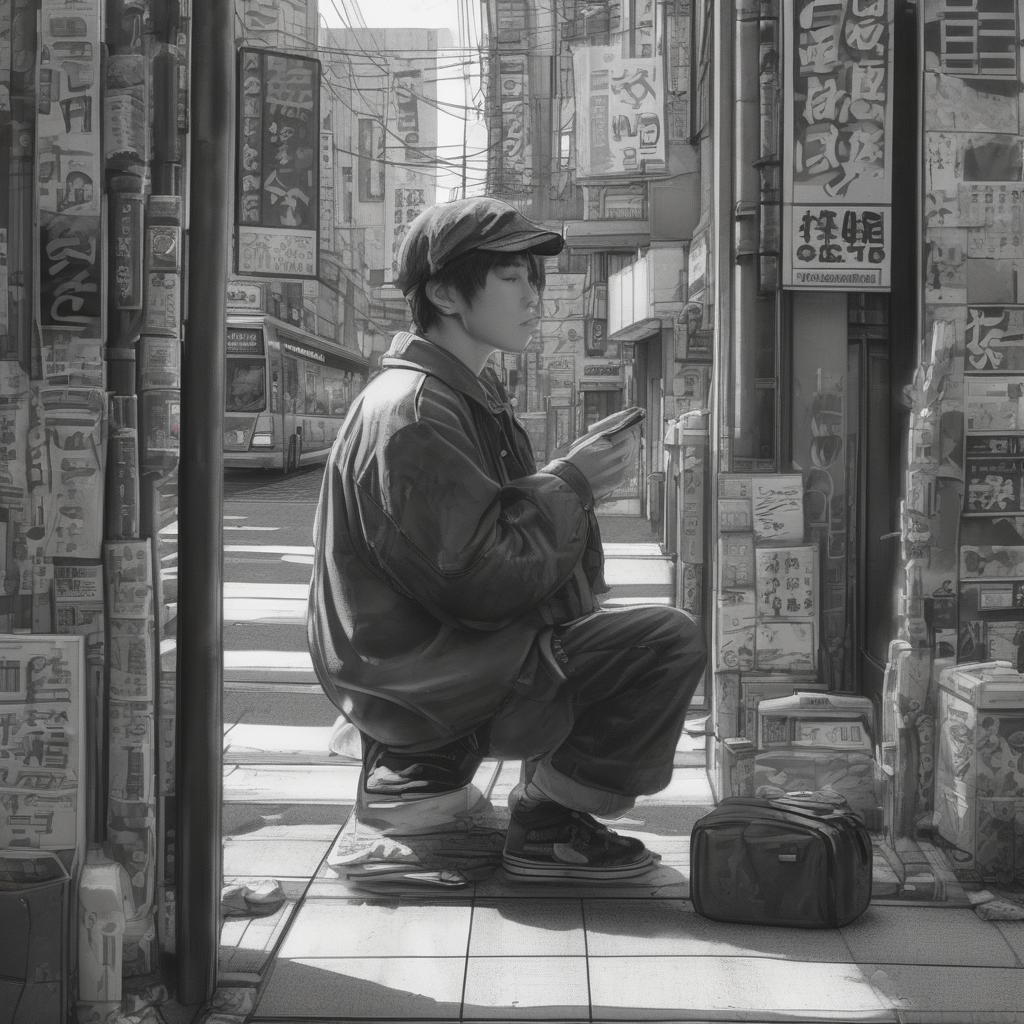}
    \end{subfigure}\hfill
    \begin{subfigure}{0.19\linewidth}
        \includegraphics[width=\linewidth]{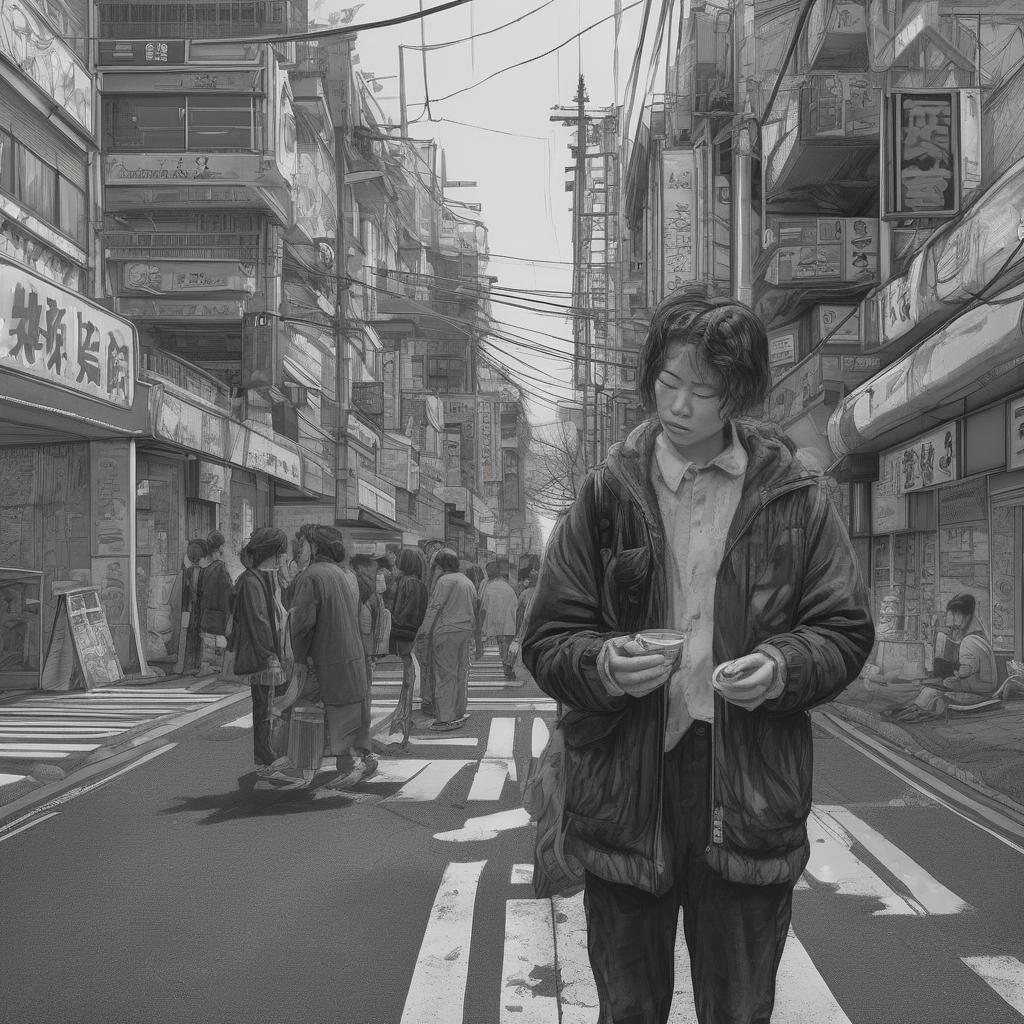}
    \end{subfigure}\hfill
    \begin{subfigure}{0.19\linewidth}
        \includegraphics[width=\linewidth]{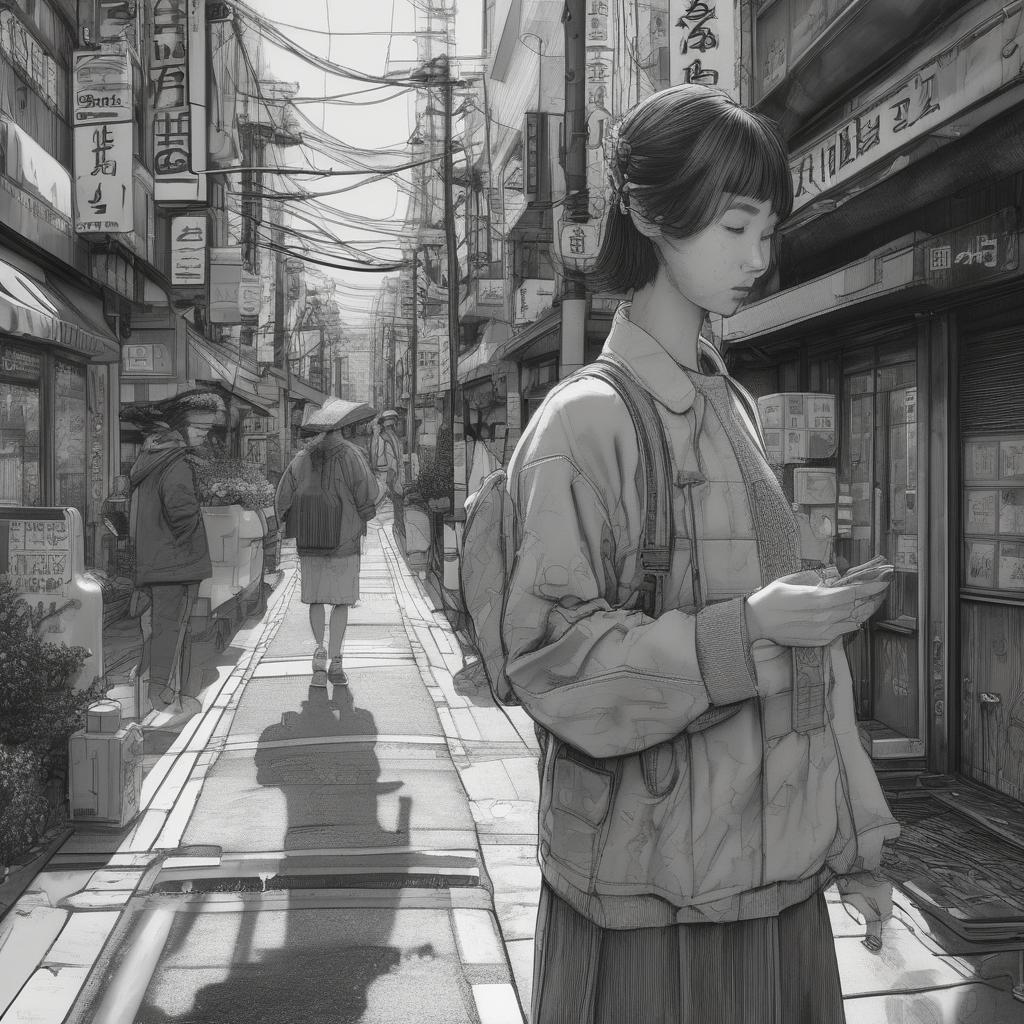}
    \end{subfigure}
    \caption*{eating crayons and being reborn in the loving warmth of jehovah on the streets of harajuku, in the style of mondo grosso killian eng kawase hasui james jean, artstation trending, 8 k, photorealistic, volumetric lighting caustics, black and white, detailed af}
    \vspace{1em}

    \begin{subfigure}{0.19\linewidth}
        \includegraphics[width=\linewidth]{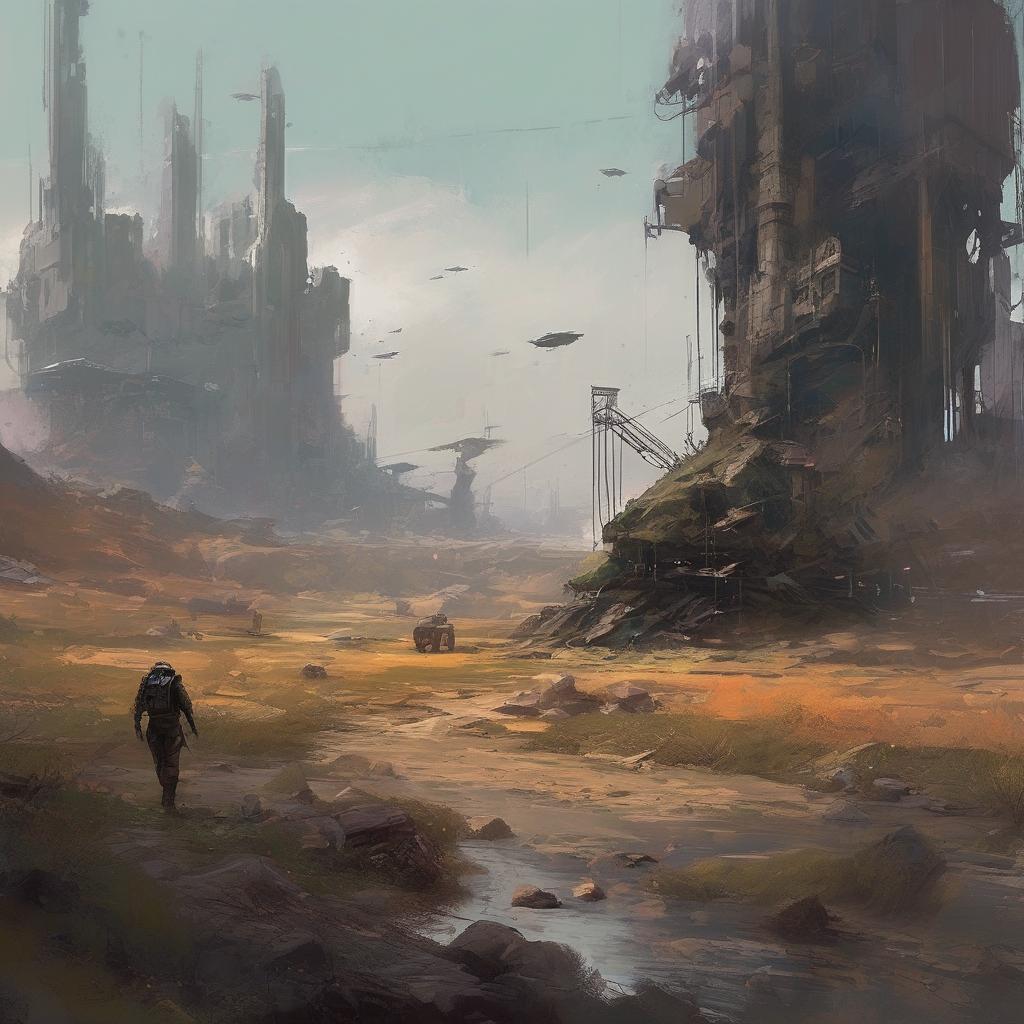}
    \end{subfigure}\hfill
    \begin{subfigure}{0.19\linewidth}
        \includegraphics[width=\linewidth]{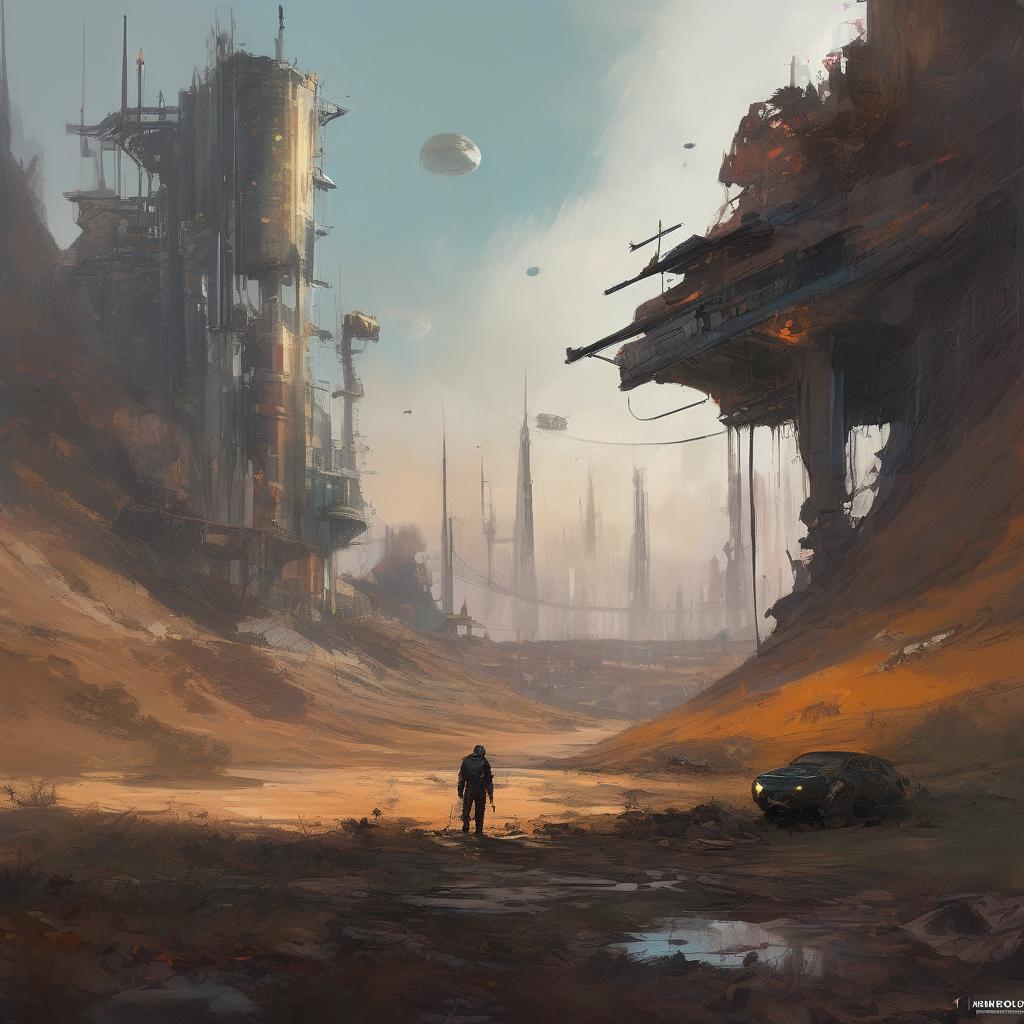}
    \end{subfigure}\hfill
    \begin{subfigure}{0.19\linewidth}
        \includegraphics[width=\linewidth]{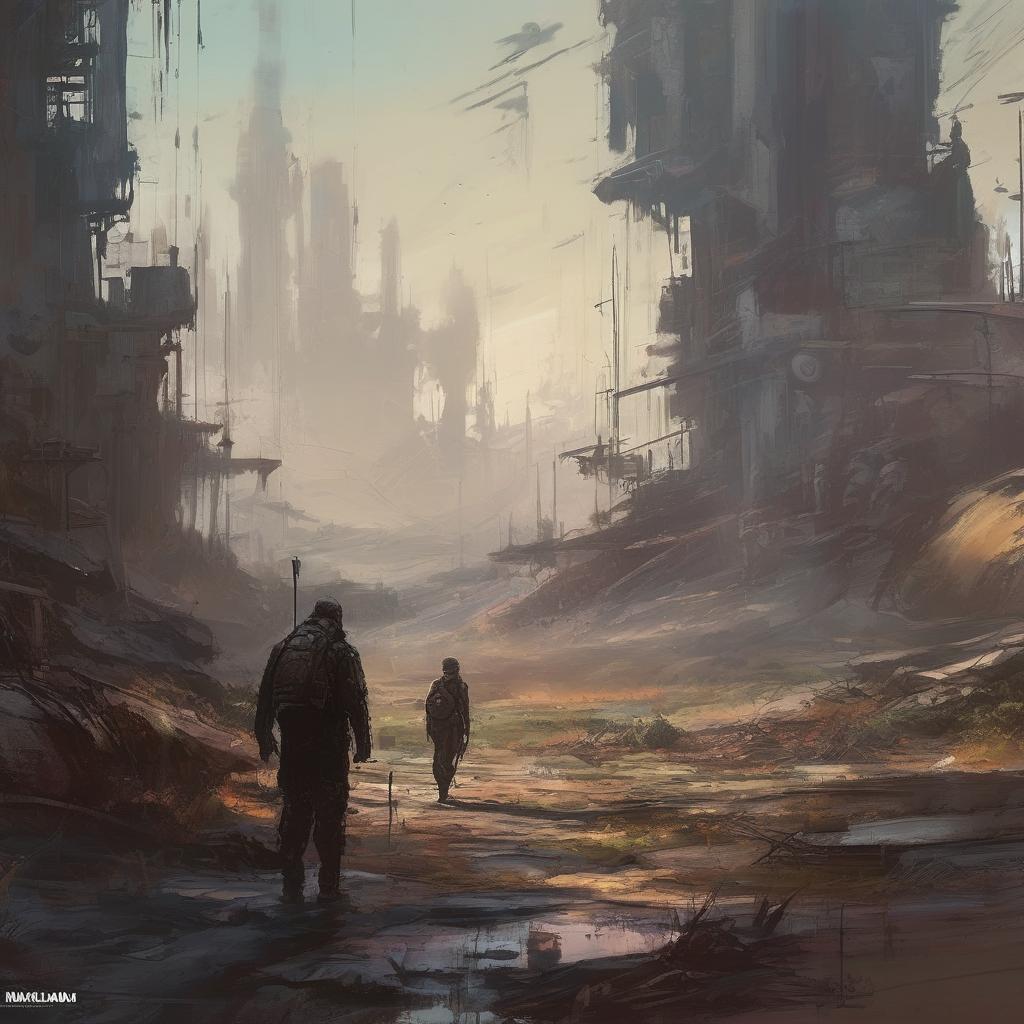}
    \end{subfigure}\hfill
    \begin{subfigure}{0.19\linewidth}
        \includegraphics[width=\linewidth]{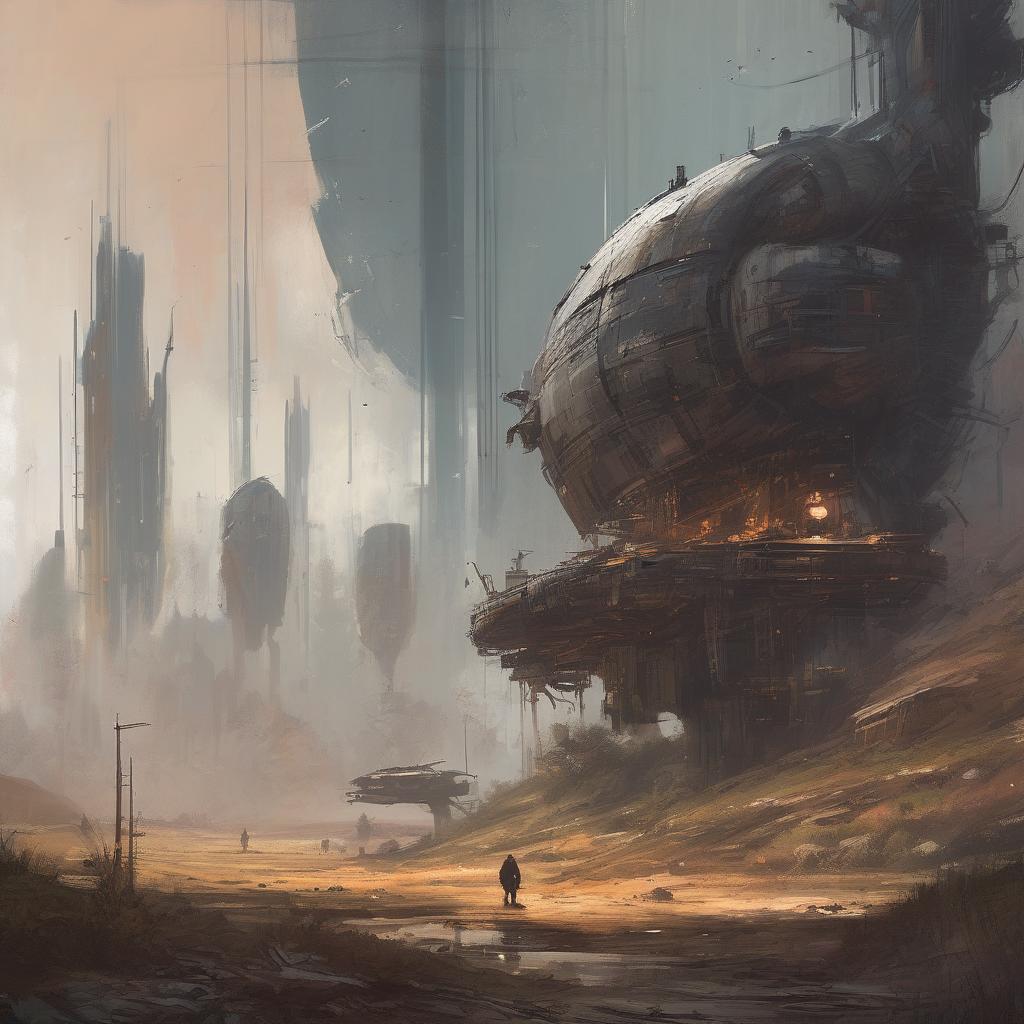}
    \end{subfigure}\hfill
    \begin{subfigure}{0.19\linewidth}
        \includegraphics[width=\linewidth]{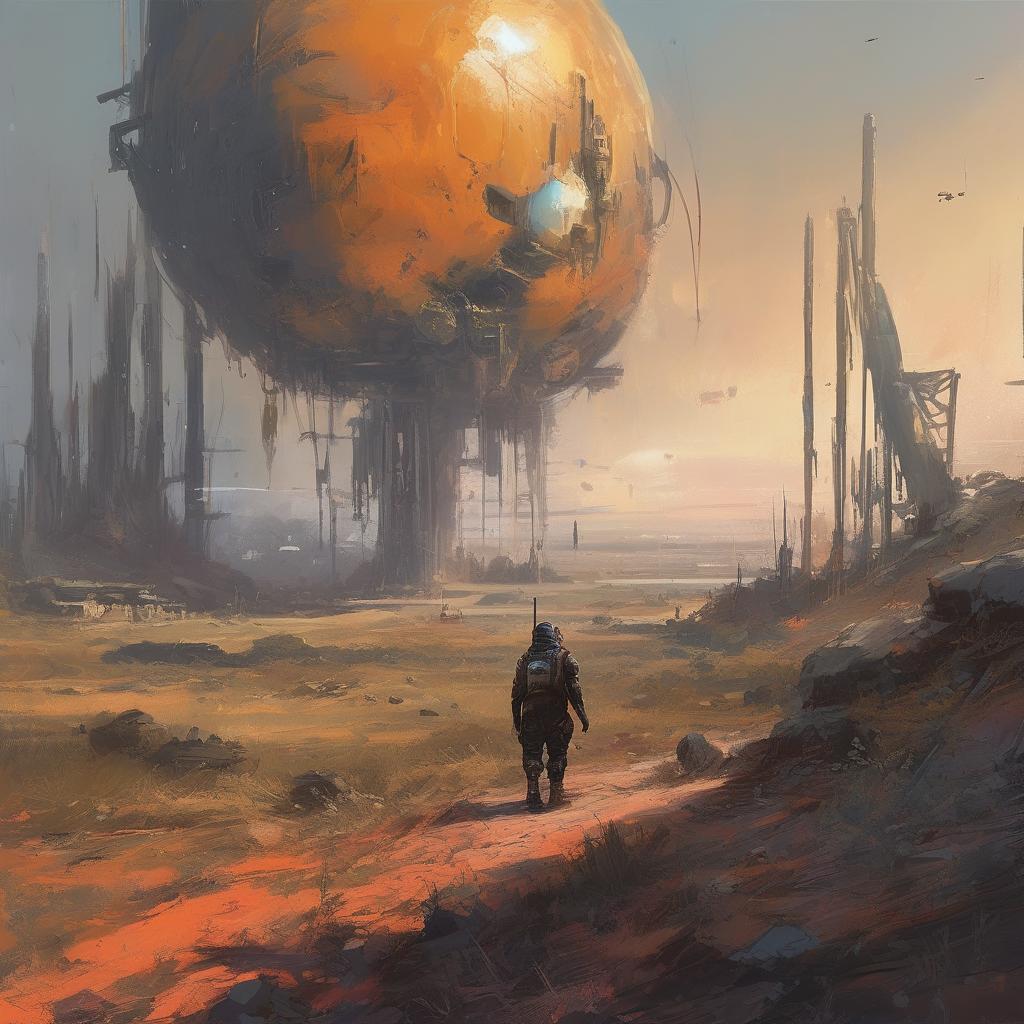}
    \end{subfigure}
    \caption*{ismail inceoglu epic oil on canvas painting of a sci - fi landscape, line art!!, ink, art concept for a book cover!!, harmonic colours, magic the gathering art, sci - fi movie still, game of thrones cardgame style, highly detailed, centered, hyperrealism, romantic, concept art, smooth, sharp focus trending on artstation, by greg manchess and by craig mullins and by kilian eng and by jake parker}

    \vspace{0.5em}
    \caption{Watermarked images generated using Stable Diffusion XL using prompts from the Stable Diffusion Prompts dataset}
    \label{fig:img/sdxl_appendix_full}
\end{figure*}

\end{document}